\documentclass[12pt]{article}%

\usepackage{fullpage}
\usepackage{orcidlink}

\usepackage{float}
\usepackage[linesnumbered,boxed,noend]{algorithm2e}
\usepackage{graphicx}
\usepackage{amsfonts}
\usepackage{amsthm}
\usepackage{amssymb}
\usepackage{tikz}
\usepackage{amsmath}
\usepackage{multirow}
\usepackage{url}
\usetikzlibrary{patterns,intersections}
\usetikzlibrary{decorations.pathreplacing}
\usetikzlibrary{arrows.meta}
\usepackage[bibliography=common]{apxproof}

\providecommand{\keywords}[1]
{
 \textbf{\textit{Keywords---}} #1
}

\renewcommand{\P}{\mathcal{P}}

\newcommand{\G}{\mathcal{G}}
\newcommand{\M}{\mathcal{M}}
\newcommand{\C}{\mathcal{C}}

\newcommand{\I}{\mathcal{I}}

\newcommand{\K}{\mathcal{K}}
\newcommand{\gmw}{\G{\text -}mc}
\newcommand{\gmparam}{modular cardinality}

\newcommand{\he}{\hat{e}}

\newcommand{\hu}{\hat{u}}
\newcommand{\ha}{\hat{a}}
\newcommand{\charge}{\gamma}
\newcommand{\mcc}{\textsc{Symmetric Multicolored Clique}}
\newcommand{\bdd}{\textsc{Bounded Degree Deletion}}

\newcommand{\ilp}{\textsc{Integer Linear Programming}}
\newcommand{\milp}{\textsc{Mixed Integer Linear Programming}}
\newcommand{\ilpwf}{\textsc{MIP with Functions}}

\newcommand{\ld}{\textsc{$(\alpha,\beta)$-Linear Degree Domination}}
\newcommand{\ad}{\textsc{$\alpha$-Domination}}
\newcommand{\kd}{$k$-\textsc{Domination}}
\newcommand{\mipwffull}{\textsc{Mixed Integer Programming with Functions}}
\newcommand{\cographs}{\textnormal{\texttt{cograph}}}
\newcommand{\clusters}{\textnormal{\texttt{cluster}}}
\newcommand{\stars}{\textnormal{\texttt{stars}}}
\newcommand{\linearforest}{\textnormal{\texttt{linear forest}}}
\newcommand{\binaryforest}{\textnormal{\texttt{binary forest}}}

\newtheorem{theorem}{Theorem}
\newtheorem{lemma}{Lemma}
\newtheorem{observation}{Observation}
\newtheorem{proposition}{Proposition}
\newtheorem{corollary}{Corollary}
\newtheorem{definition}{Definition}
\newtheorem{claim}{Claim}
\newenvironment{claimproof}{\noindent \emph{Proof.}}{\qed}

\title{Parameterized Complexity of Domination Problems Using Restricted Modular Partitions} 

\author{{Manuel Lafond
\orcidlink{0000-0002-5305-7372}\footnote{Department of Computer Science, Université de Sherbrooke, \textit{Email: manuel.lafond@usherbrooke.ca}}} 
\and 
{Weidong Luo
\orcidlink{0009-0003-5300-606X}\footnote{Department of Computer Science, Université de Sherbrooke, \textit{Email: weidong.luo@yahoo.com}} }}

\begin{document}

\maketitle

\begin{abstract}
For a graph class $\G$, we define the $\G$-\gmparam~of a graph $G$ as the minimum size of a vertex partition of $G$ into modules that each induces a graph in $\G$.  This generalizes other module-based graph parameters such as neighborhood diversity and iterated type partition.  Moreover, if $\G$ has bounded modular-width, the W[1]-hardness of a problem in $\G$-\gmparam~implies hardness on modular-width, clique-width, and other related parameters.  On the other hand, fixed-parameter tractable (FPT) algorithms in $\G$-\gmparam~may provide new ideas for algorithms using such parameters.

Several FPT algorithms based on modular partitions compute a solution table in each module, then combine each table into a global solution.  This works well when each table has a succinct representation, but as we argue, when no such representation exists, the problem is typically W[1]-hard. We illustrate these ideas on the generic $(\alpha, \beta)$-domination problem, which asks for a set of vertices that contains at least a fraction $\alpha$ of the adjacent vertices of each unchosen vertex, plus some (possibly negative) amount $\beta$. This generalizes known domination problems such as \bdd, \kd, and \ad. We show that for graph classes $\G$ that require arbitrarily large solution tables, these problems are W[1]-hard in the $\G$-\gmparam, whereas they are fixed-parameter tractable when they admit succinct solution tables.  This leads to several new positive and negative results for many domination problems parameterized by known and novel structural graph parameters such as clique-width, modular-width, and \clusters-\gmparam.
\end{abstract}


\section{Introduction}

Modular decompositions of graphs have played an important role in algorithms since their inception \cite{Gallai_1967}. 
In the world of parameterized complexity \cite{DBLP:books/sp/CyganFKLMPPS15, DBLP:series/txcs/DowneyF13}, Gajarsk{\'{y}} et. al. \cite{DBLP:conf/iwpec/GajarskyLO13}  proposed the notion of \emph{modular-width}, or $mw$ for short, which can be defined as the maximum degree of a prime node in the modular decomposition tree of $G$.
Unlike other structural parameters such as treewidth \cite{DBLP:conf/sirocco/Bodlaender07}, $mw$ can be bounded on certain classes of dense graphs, making it comparable to the clique-width ($cw$) parameter~\cite{DBLP:journals/mst/CourcelleMR00}.  In fact, $cw$ is at most $mw$ plus two, and $mw$ can sometimes be arbitrarily larger than $cw$.
It is known that several problems that are hard in $cw$ are fixed-parameter tractable (FPT) in $mw$, with popular examples  including \textsc{Hamiltonian Cycle}, \textsc{Graph Coloring} \cite{DBLP:journals/siamcomp/FominGLS10, DBLP:journals/talg/FominGLSZ19, DBLP:conf/iwpec/GajarskyLO13}, and \textsc{Metric Dimension} \cite{DBLP:journals/siamdm/BelmonteFGR17, DBLP:journals/algorithmica/BonnetP21}.
In particular in \cite{DBLP:conf/iwpec/GajarskyLO13}, the main technique used to design such algorithms is dynamic programming over the modular decomposition. In essence, the values of an optimal solution are found recursively in each module of the graph $G$, which are then combined into a solution for the graph itself, often using small integer linear programs based on the prime graph of $G$.
Another technique was recently introduced in~\cite{fomin2018algorithms}, 
where the authors show that the number of potential maximal cliques of a graph is at most $O^*(1.73^{mw})$, a fact that can be combined with results of \cite{DBLP:journals/siamcomp/BouchitteT01, DBLP:journals/siamcomp/FominTV15} to obtain FPT algorithms for several problems such as \textsc{Treewidth} and \textsc{Minimum Fill-in}.

Despite these efforts, there are still several problems that are known to be hard on $cw$, for instance \textsc{Max-Cut} \cite{DBLP:journals/siamcomp/FominGLS14} and \textsc{Bounded Degree Deletion} \cite{DBLP:journals/dam/BetzlerBNU12, DBLP:journals/algorithmica/GanianKO21}, but unknown to be hard or FPT in $mw$. Recently, an XP algorithm is given for $k$-\textsc{Domination} in parameter treewidth ($tw$)  \cite{liu2021w}, but the W-hardness for $k$-\textsc{Domination} in parameter $tw$, $cw$ or $mw$ is unknown.  Also, the authors of \cite{DBLP:conf/iwpec/GajarskyLO13} conclude with the question of whether \textsc{Edge Dominating Set} and \textsc{Partition into Triangles} are FPT in $mw$, which are 10 years-old questions that are still unanswered. 
One promising direction to gain knowledge and tools for $mw$ algorithms is to study some of its related parameters.  We consider such variants in which the graph must be decomposed into modules that each induces a subgraph belonging to a specific graph class.  Notable examples include neighborhood diversity ($nd$) \cite{DBLP:journals/algorithmica/Lampis12}, which can be seen as the minimum size of a partition into modules that induce edgeless graphs or cliques, and iterated type partition ($itp$) \cite{DBLP:conf/iwoca/CordascoGR20}, in which each module must induce a cograph.  This idea was also used in \cite{DBLP:journals/tcs/Guo09} from which we borrow our terminology, where a partition into modules inducing cliques is used to obtain linear kernels for the cluster editing problem. Meanwhile, as Knop wrote in~\cite{DBLP:journals/dam/Knop20}, `another important task in this area is to understand the boundary (viewed from the parameterized complexity perspective) between modular-width on one side, and neighborhood diversity, twin-cover number, and clique-width on the other side'. Our work resides in this boundary.



In this work, we propose to generalize the above ideas by 
restricting modules to a given graph class $\G$.  That is, we define the $\G$-\gmparam~of a graph $G$, denoted $\gmw(G)$, as the size of the smallest partition of its vertices into modules that each induces a subgraph in $\G$.  For a graph class $\G$ of bounded modular-width (such as cographs), the hardness of a problem in $\G$-\gmparam~implies its hardness in $mw$ (and thus $cw$).  On the other hand, FPT techniques for $\gmw$ may shed light towards developing better algorithms for $mw$.  Moreover, by considering graph classes of unbounded modular-width (e.g. paths, grids), $\gmw$ may be incomparable with $mw$ or even $cw$, leading to FPT algorithms for novel types of graphs.  
To the best of our knowledge, such a generalization had not been studied, although it is worth mentioning that in~\cite{DBLP:conf/stoc/JansenK021}, the authors propose a similar concept for treewidth, where some bags of a tree decomposition are allowed to induce subgraphs of a specified graph class.

\begin{figure}
\centering
\begin{tikzpicture}

\draw[color=black] (-3,0) node[scale=0.9] {$cw$};

\draw[stealth-] (-2.7,0) -- (-1.6,0);
\draw[color=black] (-1.2,0) node[scale=0.9] {$mw$};

\draw[stealth-] (-0.8,0) -- (0.4,0);
\draw[color=black] (1.5,0) node[scale=0.9] {\cographs-$mc$};

\draw[stealth-] (2.6,0) -- (3.8,0);
\draw[color=black] (4.7,0) node[scale=0.9] {\stars-$mc$};

\draw[stealth-] (2.6,0.1) -- (4.6,0.8);
\draw[color=black] (4.7,1) node[scale=0.9] {$itp$};

\draw[stealth-] (2.6,-0.1) -- (4.6,-0.8);
\draw[color=black] (4.7,-1) node[scale=0.9] {\clusters-$mc$};

\draw[stealth-] (4.8,-0.8) -- (6.7,-0.1);
\draw[stealth-] (4.8,0.8) -- (6.7,0.1);
\draw[stealth-] (5.5,0) -- (6.7,0);
\draw[color=black] (7,0) node[scale=0.9] {$nd$};

\draw[stealth-] (-2.7,0.1) -- (0.3,0.8);
\draw[color=black] (0.3,1) node[scale=0.9] {\linearforest-$mc$};

\end{tikzpicture}
\caption{Relation among the main parameters discussed in the paper, which are clique-width ($cw$), modular-width ($mw$), \cographs-\gmparam~(\cographs-$mc$), iterated type partitions ($itp$), \stars-\gmparam~(\stars-$mc$), \clusters-\gmparam~(\clusters-$mc$), neighborhood diversity ($nd$), and \linearforest-modular cardinality (\linearforest-$mc$). Arrows indicate generalizations, e.g. modular-width($mw$) generalizes \cographs-\gmparam~and thus is bounded by (a function of) \cographs-\gmparam.}
\label{fig:parameter-relationship}
\end{figure}
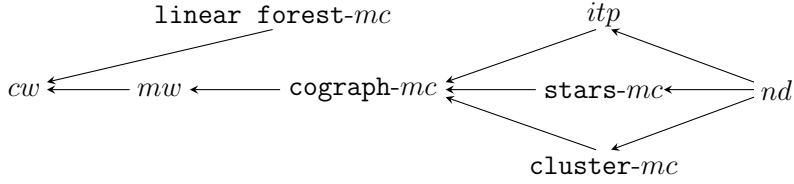

\begin{table}
\begin{center} 
\def\arraystretch{1.15}
\begin{tabular}{ | c | c | c | c | c |} 
\hline
$(\alpha,\beta)$-LDD problem &  $\alpha =0$  & $\alpha \in (0,1)$  & $\alpha =1$   & Parameters  \\
\quad & ($k$-dom) & ($\alpha$-dom $+ \beta$) &  (BDD) &  \\
\hline

\multirow{7}{*}{$\beta$ is in the input} &   W[1]-h  & W[1]-h & W[1]-h$^*$ & $cw$\\
& W[1]-h  & W[1]-h & W[1]-h & $mw$\\
& W[1]-h  & W[1]-h & W[1]-h & \cographs-$mc$\\
& W[1]-h  & W[1]-h & W[1]-h & $itp$\\
& \textbf{W[1]-h}  & W[1]-h & \textbf{W[1]-h} & \stars-$mc$\\
& open  & open & \textbf{FPT} & \clusters-$mc$\\
& \textbf{FPT}  & \textbf{FPT} & \textbf{FPT} & $nd$\\
\hline
			
\multirow{7}{*}{$\beta$ is any constant}&  \textbf{FPT}  & W[1]-h & \textbf{FPT} & $cw$\\
& FPT  & W[1]-h & FPT & $mw$\\
& FPT  & W[1]-h & FPT & \cographs-$mc$\\
& FPT  & W[1]-h & FPT & $itp$\\
& FPT  & \textbf{W[1]-h} & FPT & \stars-$mc$\\
& FPT  & open & FPT & \clusters-$mc$\\
& FPT  & FPT & FPT & $nd$\\
\hline		
\end{tabular}
\caption{The results proved in this paper for \ld~problem ($(\alpha,\beta)$-LDD) on different structural parameters (the item with mark $^*$ is implied by \cite{DBLP:journals/dam/BetzlerBNU12, DBLP:journals/algorithmica/GanianKO21}). 
The parameters include clique-width ($cw$), modular-width ($mw$), \cographs-\gmparam~(\cographs-$mc$), iterated type partitions ($itp$), \stars-\gmparam~(\stars-$mc$), \clusters-\gmparam~(\clusters-$mc$), and neighborhood diversity ($nd$), each of which is bounded by its succeeding one except $nd$. 
Results in boldface are those proved directly in this paper (other entries are implied by these results).
W[1]-hard (W[1]-h) or fixed-parameter tractable (FPT) shows the parameterized complexity of \ld~problem with specific $\alpha$, $\beta$ values.  Three cases on \clusters-$mc$ are left open.
Recall that $\alpha = 0$ is equivalent to the \kd~problem, $\alpha = 1$ to the \bdd~problem, and $\alpha \in (0, 1)$ to the \ad~problem --- we generalize the latter by adding the additive $\beta$ term (the original problem puts $\beta = 0$, and the results of the second column include this case). Not shown in the table: BDD is FPT in parameters \linearforest-$mc$ and \binaryforest-$mc$.}
\label{summary-table}
\end{center}
\end{table}

\noindent
\textbf{Our contributions.}
We first establish that if $\G$ is hereditary and has bounded $mw$, then $mw(G)$ is at most $\gmw(G)$ for a graph $G \notin \G$, allowing the transfer of hardness results.  We then show that for many graph classes, namely those that are \emph{easily mergeable}, there is a polynomial-time algorithm to compute $\gmw$ and obtain a corresponding modular partition.  

We then introduce techniques to obtain W[1]-hardness results and FPT algorithms for the $\gmw$ parameter.  
In essence, we argue that the dynamic programming technique on $mw$ algorithms works well when a small amount of information from each module is sufficient to obtain a solution for the whole graph (for instance, the algorithms of~\cite{DBLP:conf/iwpec/GajarskyLO13} require only a single integer from each module).  Such \emph{succinct solution tables} from each module can often be combined using integer programs with few integer variables \cite{DBLP:journals/combinatorica/FrankT87,DBLP:journals/mor/Kannan87,DBLP:journals/mor/Lenstra83}.
Conversely, when too much information is required from each module (e.g. linear in the size of the modules) to obtain a final solution, we are unable to use integer programming and this typically leads to W[1]-hardness.  This occurs when \emph{arbitrary solution tables} are possible in each module.

We use a large class of domination problems to illustrate these techniques.  Specifically, for a real number $\alpha \in [0, 1]$ and integer $\beta$ (possible negative), we introduce the \ld~problem.  
Given a graph $G$ and an integer $q$, this problem asks for a subset of vertices $X \subseteq V(G)$ of size at most $q$ such that, for each $v \in V(G) \setminus X$, we have $|N(v) \cap X| \geq \alpha |N(v)| + \beta$.  In other words, each unchosen vertex has at least a fraction $\alpha$ of its neighbors dominating it, plus some number $\beta$.
The problem is equivalent to the \bdd~problem \cite{DBLP:journals/dam/BetzlerBNU12, DBLP:journals/algorithmica/GanianKO21} if $\alpha = 1$ and $\beta \leq 0$; equivalent to  the \kd~problem \cite{DBLP:journals/gc/ChellaliFHV12, DBLP:journals/dam/LanC13} if $\alpha = 0$ and $\beta \geq 1$; and equivalent to the \ad~problem \cite{bakhshesha2015generalization, DBLP:journals/dmgt/DahmeRV04, DBLP:journals/dm/DunbarHLM00, DBLP:journals/ijcmcst/SaadiBAM22} if $\alpha \in (0,1]$ and $\beta=0$. Additionally, if $\alpha = 0$ and $\beta \leq 0$, then all well-coded instances of the problem are yes instances, and if $\alpha = 1$ and $\beta \geq 1$, then all instances of the problem are no instances.  We exclude those cases.


Table \ref{summary-table} illustrates some of the main results of this paper.  Most of the hardness results follow from a more general result on arbitrary solution tables, detailed in Theorem \ref{ld-main-theorem} below.  The statement is slightly technical, so we describe its high-level implication on the case $\alpha = 1$ (\bdd).  In this problem, a possible solution table is a function $f : [n] \rightarrow \mathbb{N}$ such that, for each $i \in [n]$, $f(i)$ is the minimum maximum degree achievable in $G$ after deleting $i$ vertices.   The theorem states that for graph class $\G$, if for any such $f$ we can construct a graph in $\G$ whose solution table is $f$, then \bdd~is W[1]-hard in $\G$-\gmparam.
We show that the class of disjoint stars satisfies this property, which implies several other hardness results.  
On the other hand, several positive results for the \bdd~problem make use of succinct solution tables.  In essence, when $f$ can be represented by a constant number of linear functions, or by a convex function, then we can use integer programming to merge these tables and obtain positive results.
Finally, additional results not shown in the table can be deduced easily from this technique.  We show that \bdd~is FPT in  \linearforest-\gmparam~and \binaryforest-\gmparam~as parameters, which are of interest since they are incomparable with modular-width.

\section{Preliminary notions}

For an integer $n$, denote $[n] = \{1,\ldots,n\}$.  The maximum degree of a graph $G$ is denoted $\Delta(G)$.  
The \emph{complement} of $G$, denoted $\overline{G}$, is the graph with $V(\overline{G}) = V(G)$ and, for distinct $u, v \in V(G)$, $uv \in E(\overline{G})$
if and only if $uv \notin E(G)$.  The neighborhood of $v \in V(G)$ is $N(v)$.
The set of connected components of a graph $G$ is denoted $CC(G)$.  For $X \subseteq V(G)$, $G[X]$ denotes the subgraph of $G$ induced by $X$ and $G - X = G[V(G) \setminus X]$.
If $X = \{v\}$, we may write $G - v$.
Slightly abusing notation, we may also write $v \in G$ instead of $v \in V(G)$, $|G|$ instead of $|V(G)|$, and $X \cap G$ instead of $X \cap V(G)$.

A \emph{graph class} $\G$ is a (possibly infinite) set of graphs containing at least one non-empty graph.
We say that $\G$ is \emph{hereditary} if, for any $G \in \G$, any induced subgraph of $G$ is also in $\G$.
Note that if $\G$ is hereditary, 
the graph consisting of an isolated vertex is in $\G$.
We say that $\G$ is a \emph{polynomial-time recognition graph class} if there is a polynomial-time algorithm that decides whether a given graph $G$ is in $\G$.
Some popular graph classes that we will use throughout this paper: $\I$ is the set of all edgeless graphs; $\K$ is the set of all complete graphs; \clusters~is the set of graphs in which every connected component induces a complete graph; \stars~is the set of graphs in which every connected component is a star graph; \cographs~is the set of cographs, where a cograph is either a single vertex, or a graph obtained by applying either a join or a disjoint union of two cographs \cite{brandstadt1999graph}.
Observe that $\I \subseteq \stars \subseteq \cographs$.

\medskip

\noindent
\textbf{Modular parameters.}
First, let us introduce the concepts of modular decomposition and modular-width.
For a graph $G=(V,E)$, a \emph{module} of $G$ is a set of vertices $M \subseteq V$ 
such that for every $v\in V \setminus M$, either all vertices of $M$ are adjacent to $v$ or all vertices of $M$ are not adjacent to $v$. The empty set, the vertex set $V$, and all one element sets $\{v\}$ for $v\in V$ are called the \emph{trivial modules}. In a \textit{prime} graph, all modules are trivial modules.  A \emph{factor} is a subgraph induced by a module.
A module $M$ is \emph{strong} if for any module $M'$ of $G$, either $M'$ is a subset of $M$, $M$ is a subset of $M'$, or the intersection of $M$ and $M'$ is empty. $M$ is  \emph{maximal} if $M \subsetneq V$ and there is no module $M'$ such that $M \subsetneq M' \subsetneq V$. 
 A partition $\M$ of $V(G)$ is called a \emph{modular partition} if every element of $\M$ is a module of $G$.  If $\M$ only contains maximal strong modules, then it is a \emph{maximal modular partition}.  It is known that for every graph, this partition is unique.  Two modules $M$ and $M'$ are \emph{adjacent} in $G$ if every vertex of $M$ is adjacent to every vertex of $M'$, and \emph{non-adjacent} if every vertex of $M$ is not adjacent to any vertex of $M'$. For a modular partition $\M$ of $V$, the \emph{quotient graph} $G_{/\M}$ is defined by $V(G_{/\M}) = \{v_M : M \in \M\}$ and $v_{M_1} v_{M_2} \in E(G_{/P})$ if and only if modules $M_1, M_2$ are adjacent.

We can represent all strong modules $M$ of $G$ by an inclusion tree $MD(G)$. Moreover, we call $MD(G)$ the \textit{modular decomposition tree} of $G$, in which each vertex $v_M$ is corresponding to a strong module $M$. More specifically, each leaf $v_{\{v\}}$ of the inclusion tree corresponds to a vertex $v$ of $G$ and the root vertex $v_V$ corresponds to $V$. Moreover, for any two strong modules $M$ and $M'$,  $M'$ is a proper subset of $M$ if and only if $v_{M'}$ is a descendant of $v_M$ in $MD(G)$. 
An internal vertex $v_M$ of $MD(G)$ is \textit{parallel} if $G[M]$ is disconnected, is \textit{series} if $\overline{G[M]}$ is disconnected, and is \textit{prime} otherwise.
The \textit{modular-width} of a graph $G$ is the maximum number of children of a prime vertex in $MD(G)$.
We recommend~\cite{DBLP:journals/csr/HabibP10} for more information on the topic.

Next, let us provide the definition of the parameter  \emph{neighborhood diversity} \cite{DBLP:journals/algorithmica/Lampis12}. A modular partition of $V(G)$ is called a \textit{neighborhood partition} if every module of the modular partition is either a clique or an independent set. The width of the partition is its cardinality. The \textit{minimum neighborhood partition} of $V(G)$ is the \textit{neighborhood partition} of $V(G)$ with the minimum width. The \textit{neighborhood diversity}, denoted by $nd(G)$ or $nd$, of $G$ is the width of the minimum neighborhood partition of $V(G)$. Theorem 7 of \cite{DBLP:journals/algorithmica/Lampis12} proves that minimum neighborhood partition, thus also the neighborhood diversity, can be obtained in polynomial time. Clearly, modular-width generalizes neighborhood diversity.

Our variant of modular-width restricted to graph classes follows.

\begin{definition}
Let $\G$ be a graph class.
For a given graph $G$ (not necessarily in $\G$), a module $M$ of $G$ is a \emph{$\G$-module} if $G[M]$ belongs to $\G$.  
A modular partition $\M = \{M_1, \ldots, M_k\}$ of a graph $G$ is called a \emph{$\G$-modular partition} if each $M_i$ is a $\G$-module.  

The \emph{$\G$-\gmparam}~of $G$, denoted by $\gmw(G)$, is the minimum cardinality of a $\G$-modular partition of $G$.
\end{definition}

The neighborhood diversity ($nd$) is equivalent to the $(\K \cup \I)$-\gmparam.
The iterated type partition ($itp$) parameter~\cite{DBLP:conf/iwoca/CordascoGR20} is the number of vertices of the graph obtained through the following process: 
start with the smallest modular partition into cliques and edgeless graphs;
contract each module into a single vertex; repeat until no more contraction is possible.  It can be shown that the remaining vertices represent modules that are cographs.  
Thus, $itp(G)$ is not smaller than the \cographs-\gmparam~of $\G$.

\section{Properties and tractability of \boldmath\texorpdfstring{$\mathcal{G}$}{\mathcal{G}}-\gmparam}

In this section, we show that $mw$ is not larger than $\gmw$ for graph classes of bounded modular-width. This allows hardness results on $\gmw$ to also apply to $mw$.  We then show that $\gmw$ is polynomial-time computable for ``easily mergeable'' graph classes.

Observe that for a graph $G$, if one $G$ or $\overline{G}$ is disconnected, then $CC(G)$ or $CC(\overline{G})$, respectively, form the maximal modular partition of $G$ (since the modules in this partition are defined to be strong).

Before we can compare $\gmw$ and $mw$, we will need the following. 

\begin{lemma}\label{lem:inducedmw}
Let $G$ be a graph with at least two vertices.  Then the following holds:
\begin{itemize}
\item if $G$ is disconnected, then $mw(G) = \max_{C \in CC(G)} mw(G[C])$;

\item if $\overline{G}$ is disconnected, then $mw(G) = \max_{C \in CC(\overline{G})} mw(G[C])$;

\item 
if both $G$ and $\overline{G}$ are connected, let $\M^*$ be the maximal modular partition of $G$.  Then $mw(G) = \max( |\M^*|, \max_{M \in \M^*} mw(G[M]) )$. 
\end{itemize}
\end{lemma}

\begin{proof}
According to the definitions of modular-width, the first and second cases are trivial. Let us consider the third case. Suppose $G = (V,E)$ and $MD(G)$ is the modular decomposition tree of $G$. Since both $G$ and $\overline{G}$ are connected, the root $v_V$ of $MD(G)$ is a prime vertex with $|\M^*|$ children, moreover, each child $v_M$ of $v_V$ is corresponding to a module $M \in \M^*$. Modular decomposition tree $MD(G[M])$ of $G[M]$ is a subtree of $MD(G)$ such that the root $v_M$ of $MD(G[M])$ is a child of $v_V$, and $mw(G[M])$ is the largest prime vertex in the $MD(G[M])$. Thus, $\max_{M \in \M^*} mw(G[M])$ is the largest prime node among all subtree of $MD(G)$, where the roots of the subtree are children of $v_V$. This means that $\max_{M \in \M^*} mw(G[M])$ is the largest prime vertex of $MD(G)$ except for the root vertex of $MD(G)$. Overall, $\max( |\M^*|, \max_{M \in \M^*} mw(G[M]) )$ equals the modular-width of $G$ according the definition of modular-width.
\end{proof}

We heavily rely on the relationships between $\gmw$, modules, and induced subgraphs.

\begin{lemma}\label{lem:mwcontainment}
Suppose graph $G$ is not in graph class $\G$. Let $\M$ and $\M^*$ be a $\G$-modular partition and the maximal modular partition of $G$, respectively.
Then the following holds:
\begin{itemize}

\item 
if $G$ or $\overline{G}$ is disconnected, then, for any $M \in \M$ and $M'\in \M^*$, $M \cap M' \in \{M, M', \emptyset\}$.

\item 
if $G$ and $\overline{G}$ are connected, then for any $M \in \M$, there is a $M' \in \M^*$ such that $M \subseteq M'$.
\end{itemize}
\end{lemma}

\begin{proof}
The first case is straightforward since $M'$ is a strong module. 

Consider the second case.
Note that since we assume $G \notin \G$, we have $\gmw(G) > 1$. Suppose that there are $a, b \in M$
and distinct $M^*_a, M^*_b \in \M^*$ such that 
$a \in M^*_a$, $b \in M^*_b$.  Because $M^*_a$ is the unique maximal module of $G$ that contains $a$, $M$ must be equal to $V(G)$, as otherwise, this would imply the existence of another maximal module containing $a$ (this other maximal module is either $M$, or the largest module not equal to $V(G)$ that contains $M$).  Thus $\M = \{M\}$, contradicts $\gmw(G) > 1$.  Therefore, $M$ must be a subset of some module of $\M^*$.
\end{proof}

\begin{observation}\label{obs:gmw-induced}
Let $\G$ be a hereditary graph class.  Then for any graph $G$ and any $X \subseteq V(G)$, $\gmw(G[X]) \leq \gmw(G)$.
\end{observation}
\begin{proof}
Let $\M$ be a $\G$-modular partition of $G$ of size $\gmw(G)$.  If we delete a vertex $v$ of a module $M_i \in \M$, 
the resulting subgraph $G[M_i] - v$ belongs to $\G$ by heredity.  Thus $\M' = (\M \cup \{M_i \setminus \{v\}\}) \setminus \{M_i, \emptyset\}$ 
is a $\G$-modular partition of $G - v$ (even if $M_i = \{v\}$) and $|\M'| \leq |\M|$.  
\end{proof}

Notice that without the heredity condition, Observation~\ref{obs:gmw-induced} may be false, for instance if $\G$ is the graph class 
``graphs whose number of vertices is either $1$, or a multiple of $100$'' and $G$ is an edgeless graph with $100$ vertices.  
It is now possible to show inductively that is our setting, for graph classes of fixed modular-width, having bounded $\gmw$ implies bounded $mw$.  

\begin{theorem} \label{modularwidth-smaller-than-gmw}
Let $\G$ be an hereditary graph class and define $\omega_\G := \max_{H \in \G} mw(H)$.  Then for any graph $G$, $mw(G) \leq \max(\gmw(G), \omega_\G )$. 
\end{theorem}
\begin{proof}
If $G \in \G$, then $\gmw(G) = 1$ and $mw(G) \leq \max(1, \omega_\G)$ holds trivially because $mw(G) \leq \omega_\G$ by definition.  Hence we will assume that $G \notin \G$.
We prove the statement by induction on $|V(G)|$.  As a base case, suppose that $|V(G)| = 1$.  Then $mw(G) = 0$ and $\gmw(G) = 1$ and the statement holds.
Assume that $|V(G)| > 1$ and that the statement holds for smaller graphs.  
If $\gmw(G) = 1$, then $G \in \G$ and $mw(G) \leq \omega_\G$, in which case the statement holds. Suppose for the remainder that $\gmw(G) > 1$.
Let us note that for any subset $X \subsetneq V(G)$, we may assume by induction that
\begin{align}
mw(G[X]) \leq \max(\gmw(G[X]), \omega_\G) \leq \max(\gmw(G), \omega_\G)	\label{eq:mwgx}
\end{align}
where Observation~\ref{obs:gmw-induced} was used for the last inequality.

Suppose first that $G$ (resp. $\overline{G}$) is disconnected.  By Lemma~\ref{lem:inducedmw}, 
there is $C \in CC(G)$ (resp. $C \in CC(\overline{G})$ such that $mw(G) = mw(G[C])$.
Since $C \subsetneq V(G)$, we can use (\ref{eq:mwgx}) to deduce that $mw(G) = mw(G[C]) \leq \max(\gmw(G), \omega_\G)$.

Let us then assume that $G$ and $\overline{G}$ are both connected.  Let $\M$ be a $\G$-modular partition of $G$ of size $\gmw(G)$.  
Let $\M^*$ be the maximal modular partition of $G$.
Since $\gmw(G) > 1$, we may apply  Lemma~\ref{lem:mwcontainment} and deduce that each module $M_i \in \M$ is contained in some module of $\M^*$, which implies 
$|\M^*| \leq |\M| = \gmw(G)$.  
Now consider some $M^*_j \in \M^*$.  
By (\ref{eq:mwgx}), we have $mw(G[M^*_j]) \leq \max(\gmw(G), \omega_\G)$.
Finally, by Lemma~\ref{lem:inducedmw} we know that $mw(G)$ is the maximum of $|\M^*|$ or one of $mw(G[M^*_i])$, 
and in either case the statement holds.
\end{proof} 

Let us note that the above bound is tight, in the sense that $mw(G)$ can be at least as large as either $\gmw(G)$ or $\omega(G)$.  
For instance, if $G \notin \G$ is a prime graph, then $mw(G) = \gmw(G)$, and if $G = \arg \max_{H \in \G} mw(H)$, then $mw(G) = \omega_\G$.

\medskip 

\noindent
\textbf{Computing $\gmw$ for easily mergeable graph classes.}
A graph $G$ is a \emph{$\G$-join} if $\overline{G}$ is disconnected and $G[C] \in \G$ for each $C \in CC(\overline{G})$.
Likewise, $G$ is a \emph{$\G$-union} if $G$ is disconnected and $G[C] \in \G$ for each $C \in CC(G)$.
If $G$ is a $\G$-join (resp. $\G$-union), a \emph{$\G$-merge} is a $\G$-modular partition $\M$ of $G$ such that 
for each $C \in CC(\overline{G})$ (resp. $C \in CC(G)$), there is some $M \in \M$ that contains $C$.
We say that $\M$ is a \emph{minimum} $\G$-merge if no other $\G$-merge has a size strictly smaller than $\M$.

We say that a graph class $\G$ is \emph{easily mergeable} if there exists a polynomial time algorithm that,
given a graph $G$ such that $G$ is either a $\G$-join or a $\G$-union, 
outputs a minimum $\G$-merge of $G$. 
We say that $\G$ is \emph{trivially mergeable} if, for any $\G$-join or $\G$-union $G$, one of $\{V(G)\}, CC(G)$, or $CC(\overline{G})$ is a minimum $\G$-merge of $G$.

For easily mergeable classes, the $\G$-\gmparam~can be computed easily.  If $G \in \G$, we can simply return $\{V(G)\}$.  If $G$ and $\overline{G}$ are connected, we know by Lemma~\ref{lem:mwcontainment} that each module of a minimum $\G$-modular partition is contained in a maximal module.  It thus suffices to find $\G$-modules recursively into each module of the maximal modular partition.
The complex cases occur when $G$ or $\overline{G}$ is disconnected. 
It does not suffice to find the $\G$-modules recursively in each connected component. 
Indeed, the minimum $\G$-modular partition might merge several of these connected components into one module.
As we show, we may safely recurse into the connected components that do not induce a graph in $\G$.  For the connected components that do induce a graph in $\G$, we must decide how to merge them optimally, and this is where the easily mergeable property is needed.  Algorithm~\ref{alg:gmw} describes this idea.

\vspace{3mm}
\begin{algorithm}[h]
\SetAlgoNoLine
\SetKwProg{Fn}{function}{}{}
\Fn{compute-$\G$-mc($G$)}{
    \lIf{$G \in \G$}
    {
        return $\{V(G)\}$
    }
    $\M = \emptyset$  \quad //modules to return\;
    \uIf{both $G$ and $\overline{G}$ are connected}
    {
        Let $\M^* = \{M^*_1, \ldots, M^*_l\}$ be the maximal modular partition of $G$\;
        $\M = \bigcup_{i=1}^l$ compute-$\G$-mc($G[M^*_i]$)\;
    }
    \uElseIf{$G$ or $\overline{G}$ is disconnected}
    {
        Let $\C = CC(G)$ if $G$ is disconnected, and let $\C = CC(\overline{G})$ otherwise\;
        Let $\C_1 = \{C \in \C : G[C] \in \G\}$ and $\C_2 = \C \setminus \C_1$\;
        Let $\M$ be a minimum $\G$-merge of $G[ \bigcup_{C \in \C_1} C]$\;
        \For{$C \in C_2$}
        {
            $\M_C = $ compute-$\G$-mc($G[C]$)\;
            Add every module of $\M_C$ to $\M$\;
        }
        
    }
    return $\M$\;
   
 }
\caption{Computing a $\G$-modular partition for a graph class $\G$.}
\label{alg:gmw}
\end{algorithm}
\vspace{3mm}

By combining Lemma~\ref{lem:mwcontainment} and induction, we can show the following.

\begin{theorem}
Suppose that $\G$ is a hereditary graph class.  Suppose further that $\G$ is polynomial-time recognizable and easily mergeable.
Then Algorithm~\ref{alg:gmw} returns a $\G$-modular partition of $G$ of minimum size in polynomial time.
\end{theorem}
\begin{proof}
We focus on the correctness of the algorithm (the polynomial time is not too difficult to verify).  The proof of correctness is by induction on $|V(G)|$.  If $|V(G)| = 1$, 
then by heredity, $G \in\G$ and the algorithm correctly returns $\{V(G)\}$.

Suppose that $|V(G)| > 2$ and that the algorithm is correct on smaller graphs.  If $G \in \G$, then it is clear that returning $\{V(G)\}$ is correct.  So assume that $G \notin \G$. 
Let $\M$ be the partition of $V(G)$ returned by the algorithm, and let $\M'$ be a $\G$-modular partition of $G$ of minimum size.  It is not hard to check that in all cases, $\M$ is a $\G$-modular partition: the modules added to $\M$ either come from recursive calls (and are thus $\G$-modules by induction), or are explicitly checked to be in $\G$ (the $\C_1$ set combined with the definition of a $\G$-merge). Hence, we focus on showing that $|\M| = |\M'|$.

Consider the case in which both $G$ and $\overline{G}$ are connected.  
Let $\M^*$ be the maximal modular partition of $G$.
By Lemma~\ref{lem:mwcontainment}, each module of $\M'$ is a subset of some module of $\M^*$.
For $M^*_i \in \M^*$, let $\M'_i = \{M' \in \M' : M' \subseteq M^*_i\}$.  
Then $\M'_i$ must be a $\G$-modular partition of $G[M^*_i]$ of minimum size, because if there is 
a smaller such partition $\M''_i$, then $(\M' \setminus \M'_i) \cup \M''_i$ is a better $\G$-partition of $G$.
By induction, Algorithm~\ref{alg:gmw} returns a $\G$-modular partition of $G[M^*_i]$ of size $|\M'_i|$ and, 
since this is true for each $M^*_i \in \M^*$, $|\M| = |\M'|$.

So, consider the case in which $G$ or $\overline{G}$ is disconnected.   
Let $\C = CC(G)$ if $G$ is disconnected, and let $\C = CC(\overline{G})$ otherwise.  
As a consequence of Lemma~\ref{lem:mwcontainment}, for any $C \in \C$, either some $M'_i \in \M'$ contains $C$, or there is $\M'_i \subseteq \M'$ of size at least two that partitions $C$.  
As in the algorithm, let $\C_1 = \{C \in \C : G[C] \in \G\}$ and let 
$\C_2 = \C \setminus \C_1$.

Let $C \in \C_1$.  Since $G[C] \in \G$, we may assume that there is some $M'_i \in \M'$  that contains $C$, since otherwise by Lemma~\ref{lem:mwcontainment}, there is some $\M'_i \subseteq \M$ of size at least two that partitions $\C$.  If the latter occurs, $(\M \setminus \M'_i) \cup \{C\}$ is a better $\G$-modular partition, a contradiction.
On the other hand, let $C \in \C_2$.  Since $G[C] \notin \G$, there cannot be $M'_i \in \M'$ that contains $C$.  Indeed, if such a $M'_i$ existed, we would have $G[M'_i] \in \G$ and, by heredity, $G[C] \in \G$, a contradiction to the definition of $\C_2$.
Therefore, for each $C \in \C_2$, there is $\M'_C \subseteq \M'$ of size at least two that partitions $C$.

We can thus deduce that if we define $\M'_1$ as the set of modules of $\M$ that contain some element of $\C_1$, then 
$\M'_2 := \M' \setminus \M'_1$ consists of the modules that are properly contained in some element of $\C_2$. 
Moreover, since $\M'_1$ is a $\G$-modular partition of $G[\bigcup_{C \in \C_1} C]$ and $M' \in \G$ for each $M' \in \M'_1$, it follows that $\M'_1$ is a $\G$-merge of 
$G[\bigcup_{C \in \C_1} C]$.
Since the $\G$-merge $\M_1$ computed by Algorithm~\ref{alg:gmw} is minimum, we have $|\M_1| \leq |\M'_1|$.
In a similar fashion, consider $C \in \C_2$ and let $\M'_C$ be the set of modules of $\M'$ contained in $C$. 
By induction, compute-$\G$-mc($G[C]$) returns a minimum $\G$-modular partition $\M_C$ of $G[C]$, and thus $|\M_C| \leq |\M'_C|$.  
Therefore, every set of module added by Algorithm~\ref{alg:gmw} to $\M$ corresponds to a distinct set of modules in $\M'$ that is at least as large, and it follows that $\M$ is of minimum size as well.
\end{proof}

\begin{corollary}
For $\G \in \{\I, \K, \cographs, \clusters\}$, a $\G$-modular partition of minimum size of a graph $G$ can be computed in polynomial time.
\end{corollary}

\begin{proof}
We argue that the graph classes mentioned above are trivially mergeable. 
For cographs, it follows from the definition of a cograph that if $G$ is a \cographs-join or a \cographs-union, then $G$ is itself a cograph.   Thus $\{V(G)\}$ is always a $\G$-merge.
For cliques, if $G$ is a $\K$-join, then $\{V(G)\}$ is a $\G$-merge, and if $G$ is a $\K$-union, then $CC(G)$ is the unique minimum $\G$-merge.
For edgeless graphs, if $G$ is an $\I$-join then $CC(\overline{G})$ is a $\G$-merge, and if $G$ is an $\I$-union, then $\{V(G)\}$ is a $\G$-merge.

Let $\G = \clusters$.  If $G$ is a \clusters-union, then since each $C \in CC(G)$ is a set of disjoint cliques, it follows that the same holds for $G$ and thus $\{V(G)\}$ is a $\G$-merge of $G$.
Suppose that $G$ is a \clusters-join and let $\M$ be a minimum $\G$-merge of $G$.  
Let $C \in CC(\overline{G})$.  We claim that if $C$ contains at least two connected components $K_a$ and $K_b$ (which are cliques), then $C \in \M$.  Recall that by Lemma~\ref{lem:mwcontainment}, the modules of $\M$ that intersect with $C$ either contain $C$, or are contained in $C$.  Notice that there is no $M \in \M$ that contains $C$ and another vertex from another $C' \in CC(\overline{G})$, because if that was the case, we could form a $P_3$ by taking a vertex in each of $K_a, K_b$, and $C'$, contradicting that $G[M]$ is a cluster graph.
Thus the modules of $\M$ that intersect with $C$ are either proper subsets of $C$ or equal to $C$. Because $C \in \G$ and $\M$ is minimum, we may assume that $C \in \M$, as claimed.
It only remains to deal with $\C' = \{C \in CC(\overline{G}) : C$ is a clique$\}$. 
Clearly, $\bigcup_{C' \in \C'} C'$ is a cluster graph, and it follows that a minimum $\G$-merge can be obtained by taking each $C \in CC(\overline{G}) \setminus \C'$, and the union of the vertices in $\C'$. 
\end{proof}

Note that not all graph classes are equally easy to merge.  Consider the class ``graphs with at most $100$ vertices''.  Merging such graphs amounts to solving a bin packing problem with a fixed capacity of $100$, and current polynomial-time algorithms require time in the order of $n^{100}$.  This is easily mergeable nonetheless, and it would be interesting to find graph classes that are hereditary and polynomial-time recognizable, but not easily mergeable.

\section{Hardness of domination problems with arbitrary solution tables}

Let us recall our generic domination problem of interest.  For $\alpha \in [0, 1]$ and $\beta \in \mathbb{Z}$ we define:

\medskip 

\noindent 
The~\ld~problem\\
\textbf{Input: } a graph $G = (V, E)$ and a non-negative integer $q$;\\
\textbf{Question:} does there exist a subset $X\subseteq V$ of size at most $q$ such that for every $v \in V \setminus X$,  $|N(v)\cap X| \geq \alpha |N(v)| + \beta$?

\medskip

In the above, the vertex set $X$ is called a \emph{$(\alpha, \beta)$-linear degree dominating set} of $G$.  In addition, for convenience, we also call $X$ the deletion part of $G$. 
For any vertex $v\in V(G -  X)$, the degree of $v$ in $X$, denoted by $\deg (v,G,X)$, equals the number of vertices in $X\cap V(G)$ adjacent to $v$. The minimum degree of $G - X$ in $X$ equals $\min \{\deg(v,G,X) : v\in G - X\}$, which is denoted by $\delta (G - X, X)$. 

Our goal is to formalize the intuition that graph classes with arbitrary solution tables lead to W[1]-hardness in $\gmw$.
For \ld, this takes the form of \emph{arbitrary deletion tables} for $\alpha \in (0, 1]$ and \emph{arbitrary retention tables} for $\alpha = 0$.

\begin{definition}\label{def:kd-degdeletion-bdd-121}
Let $(a_0, I, c)$ be a triple with $\{a_0\} \cup I \subseteq \mathbb{N}$ and $c : \mathbb{N} \rightarrow \mathbb{N}$.  We call $c$ a cost function.  We say that $(a_0, I, c)$ is \emph{decreasing valid} if there is an ordering $a_1, \ldots, a_{|I|}$ of $I$ such that $a_0 > a_1 > \ldots > a_{|I|} >1$, and $0= c(a_0) < c(a_1) < \ldots < c(a_{|I|})$.

For a decreasing valid triple $(a_0, I, c)$, we say that 
a graph $G =(V,E)$ is a \emph{$(a_0, I, c)$-degree deletion graph} if all of the following conditions hold:
\begin{enumerate}
    
    \item 
    $G$ has maximum degree $a_0$ and at most $(a_0c(a_{|I|}))^{10}$ vertices;
    
    \item 
    for any $a_i \in I$, there exists $X \subseteq V$ of size $c(a_i)$ such that $G - X$ has maximum degree $a_i$;
    
    \item \label{def:degdeletion-3}
    for any $a_i \in I$ and any $X \subseteq V$ of size strictly less than $c(a_i)$, $G - X$ has maximum degree at least $a_{i - 1}$.
    
\end{enumerate}   
In addition, if $G \in \stars$ and satisfies the above three conditions, then we say that $G$ is a \emph{$(a_0, I, c)$-degree deletion star graph}.
\end{definition}

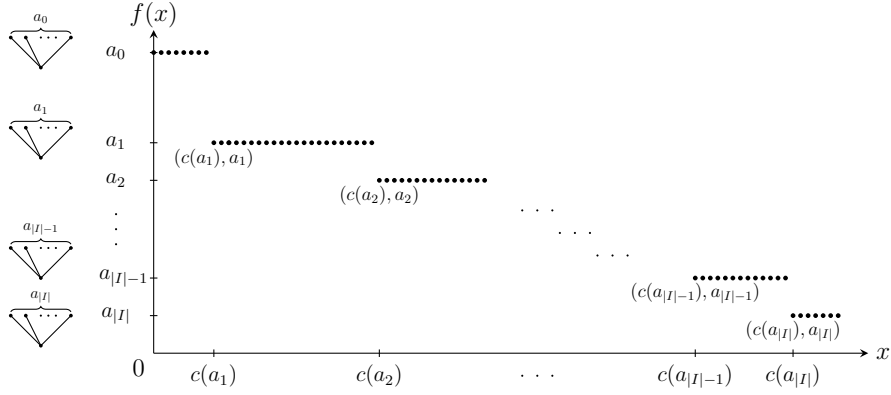
\begin{figure}
\centering
\begin{tikzpicture}
\fill [color=black] (-3.5,2.8) circle (0.7pt);

\fill [color=black] (-3.9,3.2) circle (0.7pt);
\fill [color=black] (-3.7,3.2) circle (0.7pt);
\fill [color=black] (-3.5,3.2) circle (0.4pt);
\fill [color=black] (-3.4,3.2) circle (0.4pt);
\fill [color=black] (-3.3,3.2) circle (0.4pt);
\fill [color=black] (-3.1,3.2) circle (0.7pt);
\draw (-3.5,2.8) -- (-3.9,3.2);
\draw (-3.5,2.8) -- (-3.7,3.2);
\draw (-3.5,2.8) -- (-3.1,3.2);
\draw [decorate,decoration={brace,amplitude=2pt},xshift=0pt,yshift=0pt]
(-3.9,3.25) -- (-3.1,3.25) node [black,midway,xshift=-0.6cm]{};
\draw[color=black] (-3.5,3.45) node[scale=0.5]{$a_0$};


\fill [color=black] (-3.5,2.8-1.2) circle (0.7pt);

\fill [color=black] (-3.9,3.2-1.2) circle (0.7pt);
\fill [color=black] (-3.7,3.2-1.2) circle (0.7pt);
\fill [color=black] (-3.5,3.2-1.2) circle (0.4pt);
\fill [color=black] (-3.4,3.2-1.2) circle (0.4pt);
\fill [color=black] (-3.3,3.2-1.2) circle (0.4pt);
\fill [color=black] (-3.1,3.2-1.2) circle (0.7pt);
\draw (-3.5,2.8-1.2) -- (-3.9,3.2-1.2);
\draw (-3.5,2.8-1.2) -- (-3.7,3.2-1.2);
\draw (-3.5,2.8-1.2) -- (-3.1,3.2-1.2);
\draw [decorate,decoration={brace,amplitude=2pt},xshift=0pt,yshift=0pt]
(-3.9,3.25-1.2) -- (-3.1,3.25-1.2) node [black,midway,xshift=-0.6cm]{};
\draw[color=black] (-3.5,3.45-1.2) node[scale=0.5]{$a_1$};

\fill [color=black] (-3.5,2.8-2.8) circle (0.7pt);

\fill [color=black] (-3.9,3.2-2.8) circle (0.7pt);
\fill [color=black] (-3.7,3.2-2.8) circle (0.7pt);
\fill [color=black] (-3.5,3.2-2.8) circle (0.4pt);
\fill [color=black] (-3.4,3.2-2.8) circle (0.4pt);
\fill [color=black] (-3.3,3.2-2.8) circle (0.4pt);
\fill [color=black] (-3.1,3.2-2.8) circle (0.7pt);
\draw (-3.5,2.8-2.8) -- (-3.9,3.2-2.8);
\draw (-3.5,2.8-2.8) -- (-3.7,3.2-2.8);
\draw (-3.5,2.8-2.8) -- (-3.1,3.2-2.8);
\draw [decorate,decoration={brace,amplitude=2pt},xshift=0pt,yshift=0pt]
(-3.9,3.25-2.8) -- (-3.1,3.25-2.8) node [black,midway,xshift=-0.6cm]{};
\draw[color=black] (-3.5,3.45-2.8) node[scale=0.5]{$a_{|I|-1}$};

\fill [color=black] (-3.5,2.8-3.7) circle (0.7pt);

\fill [color=black] (-3.9,3.2-3.7) circle (0.7pt);
\fill [color=black] (-3.7,3.2-3.7) circle (0.7pt);
\fill [color=black] (-3.5,3.2-3.7) circle (0.4pt);
\fill [color=black] (-3.4,3.2-3.7) circle (0.4pt);
\fill [color=black] (-3.3,3.2-3.7) circle (0.4pt);
\fill [color=black] (-3.1,3.2-3.7) circle (0.7pt);
\draw (-3.5,2.8-3.7) -- (-3.9,3.2-3.7);
\draw (-3.5,2.8-3.7) -- (-3.7,3.2-3.7);
\draw (-3.5,2.8-3.7) -- (-3.1,3.2-3.7);
\draw [decorate,decoration={brace,amplitude=2pt},xshift=0pt,yshift=0pt]
(-3.9,3.25-3.7) -- (-3.1,3.25-3.7) node [black,midway,xshift=-0.6cm]{};
\draw[color=black] (-3.5,3.45-3.7) node[scale=0.5]{$a_{|I|}$};

\draw[-stealth] (-2,-1) -- (7.5,-1);
\draw[color=black] (7.7,-1) node[scale=0.8] {$x$};

\draw[color=black] (-2.2,-1.2) node[scale=0.8] {$0$};


\draw[color=black] (-1.2,-1.3) node[scale=0.7] {$c(a_1)$};
\draw (-1.2,-1.05) -- (-1.2,-0.95);

\draw[color=black] (1,-1.3) node[scale=0.7] {$c(a_2)$};
\draw (1,-1.05) -- (1,-0.95);

\fill [color=black] (2.9,-1.3) circle (0.4pt);
\fill [color=black] (3.1,-1.3) circle (0.4pt);
\fill [color=black] (3.3,-1.3) circle (0.4pt);

\draw[color=black] (5.2,-1.3) node[scale=0.7] {$c(a_{|I|-1})$};
\draw (5.2,-1.05) -- (5.2,-0.95);

\draw[color=black] (6.5,-1.3) node[scale=0.7] {$c(a_{|I|})$};
\draw (6.5,-1.05) -- (6.5,-0.95);

\draw[-stealth] (-2,-1) -- (-2,3.3);
\draw[color=black] (-2,3.5) node[scale=0.8] {$f(x)$};

\draw[color=black] (-2.5,3) node[scale=0.7] {$a_0$};
\draw (-2.05,3) -- (-1.95,3);

\draw[color=black] (-2.5,1.8) node[scale=0.7] {$a_1$}; 
\draw (-2.05,1.8) -- (-1.95,1.8);

\draw[color=black] (-2.5,1.3) node[scale=0.7] {$a_2$}; 
\draw (-2.05,1.3) -- (-1.95,1.3);

\fill [color=black] (-2.5,0.85) circle (0.4pt);
\fill [color=black] (-2.5,0.65) circle (0.4pt);
\fill [color=black] (-2.5,0.45) circle (0.4pt);

\draw[color=black] (-2.4,0) node[scale=0.7] {$a_{|I|-1}$}; 
\draw (-2.05,0) -- (-1.95,0);

\draw[color=black] (-2.5,-0.5) node[scale=0.7] {$a_{|I|}$}; 
\draw (-2.05,-0.5) -- (-1.95,-0.5);


\fill [color=black] (-2,3) circle (0.9pt);
\fill [color=black] (-1.9,3) circle (0.9pt);
\fill [color=black] (-1.8,3) circle (0.9pt);
\fill [color=black] (-1.7,3) circle (0.9pt);
\fill [color=black] (-1.6,3) circle (0.9pt);
\fill [color=black] (-1.5,3) circle (0.9pt);
\fill [color=black] (-1.4,3) circle (0.9pt);
\fill [color=black] (-1.3,3) circle (0.9pt);

\fill [color=black] (-1.2,1.8) circle (0.9pt);
\fill [color=black] (-1.1,1.8) circle (0.9pt);
\fill [color=black] (-1,1.8) circle (0.9pt);
\fill [color=black] (-1,1.8) circle (0.9pt);
\fill [color=black] (-0.9,1.8) circle (0.9pt);
\fill [color=black] (-0.8,1.8) circle (0.9pt);
\fill [color=black] (-0.7,1.8) circle (0.9pt);
\fill [color=black] (-0.6,1.8) circle (0.9pt);
\fill [color=black] (-0.5,1.8) circle (0.9pt);
\fill [color=black] (-0.4,1.8) circle (0.9pt);
\fill [color=black] (-0.3,1.8) circle (0.9pt);
\fill [color=black] (-0.2,1.8) circle (0.9pt);
\fill [color=black] (-0.1,1.8) circle (0.9pt);
\fill [color=black] (0,1.8) circle (0.9pt);
\fill [color=black] (0.1,1.8) circle (0.9pt);
\fill [color=black] (0.2,1.8) circle (0.9pt);
\fill [color=black] (0.3,1.8) circle (0.9pt);
\fill [color=black] (0.4,1.8) circle (0.9pt);
\fill [color=black] (0.5,1.8) circle (0.9pt);
\fill [color=black] (0.6,1.8) circle (0.9pt);
\fill [color=black] (0.7,1.8) circle (0.9pt);
\fill [color=black] (0.8,1.8) circle (0.9pt);
\fill [color=black] (0.9,1.8) circle (0.9pt);

\fill [color=black] (1,1.3) circle (0.9pt);
\fill [color=black] (1.1,1.3) circle (0.9pt);
\fill [color=black] (1.2,1.3) circle (0.9pt);
\fill [color=black] (1.3,1.3) circle (0.9pt);
\fill [color=black] (1.4,1.3) circle (0.9pt);
\fill [color=black] (1.5,1.3) circle (0.9pt);
\fill [color=black] (1.6,1.3) circle (0.9pt);
\fill [color=black] (1.7,1.3) circle (0.9pt);
\fill [color=black] (1.8,1.3) circle (0.9pt);
\fill [color=black] (1.9,1.3) circle (0.9pt);
\fill [color=black] (2,1.3) circle (0.9pt);
\fill [color=black] (2.1,1.3) circle (0.9pt);
\fill [color=black] (2.2,1.3) circle (0.9pt);
\fill [color=black] (2.3,1.3) circle (0.9pt);
\fill [color=black] (2.4,1.3) circle (0.9pt);

\fill [color=black] (2.9,0.9) circle (0.4pt);
\fill [color=black] (3.1,0.9) circle (0.4pt);
\fill [color=black] (3.3,0.9) circle (0.4pt);

\fill [color=black] (3.4,0.6) circle (0.4pt);
\fill [color=black] (3.6,0.6) circle (0.4pt);
\fill [color=black] (3.8,0.6) circle (0.4pt);

\fill [color=black] (3.9,0.3) circle (0.4pt);
\fill [color=black] (4.1,0.3) circle (0.4pt);
\fill [color=black] (4.3,0.3) circle (0.4pt);

\fill [color=black] (5.2,0) circle (0.9pt);
\fill [color=black] (5.3,0) circle (0.9pt);
\fill [color=black] (5.4,0) circle (0.9pt);
\fill [color=black] (5.5,0) circle (0.9pt);
\fill [color=black] (5.6,0) circle (0.9pt);
\fill [color=black] (5.7,0) circle (0.9pt);
\fill [color=black] (5.8,0) circle (0.9pt);
\fill [color=black] (5.9,0) circle (0.9pt);
\fill [color=black] (6,0) circle (0.9pt);
\fill [color=black] (6.1,0) circle (0.9pt);
\fill [color=black] (6.2,0) circle (0.9pt);
\fill [color=black] (6.3,0) circle (0.9pt);
\fill [color=black] (6.4,0) circle (0.9pt);

\fill [color=black] (6.5,-0.5) circle (0.9pt);
\fill [color=black] (6.6,-0.5) circle (0.9pt);
\fill [color=black] (6.7,-0.5) circle (0.9pt);
\fill [color=black] (6.8,-0.5) circle (0.9pt);
\fill [color=black] (6.9,-0.5) circle (0.9pt);
\fill [color=black] (7,-0.5) circle (0.9pt);
\fill [color=black] (7.1,-0.5) circle (0.9pt);

\draw[color=black] (-1.2,1.6) node[scale=0.6] {$(c(a_1),a_1)$};

\draw[color=black] (1,1.1) node[scale=0.6] {$(c(a_2),a_2)$};

\draw[color=black] (5.2,-0.2) node[scale=0.6] {$(c(a_{|I|-1}),a_{|I|-1})$};

\draw[color=black] (6.5,-0.7) node[scale=0.6] {$(c(a_{|I|}),a_{|I|})$};

\end{tikzpicture}
\caption{The degree deletion function of a $(a_0, I, c)$-degree deletion graph.  The $x$-axis represents the number of vertices to delete, while the $y$-axis represents the minimum maximum degree achievable by deleting $x$ vertices.  The first point of the $(i+1)$-th horizontal vertex set is $(c(a_i),a_i)$ for $0 \leq i \leq |I|$, and the last point of the $(i+1)$-th horizontal vertex set is $(c(a_{i+1}) - 1,a_i)$ for $0 \leq i \leq |I| -1$. Disjoint stars are used as an example here since the graph class \stars~admits arbitrary deletion tables. The number of stars with $a_i$ leaves is $c(a_{i+1}) - c(a_i)$ for $0 \leq i \leq |I| -1$. The degree deletion process removes the internal vertices of stars from large to small.}
\label{fig:def-del-table}
\end{figure}

We say that a graph class $\G$ \emph{admits arbitrary degree deletion tables} if, for any decreasing valid triple $(a_0, I, c)$, one can construct in time polynomial in $a_0c(a_{|I|})$ a graph $G \in \G$ such that $G$ is a $(a_0, I, c)$-degree deletion graph.
Note that the size of $G$ is only required to be a polynomial function of $a_0c(a_{|I|})$, but we fix it to $(a_0c(a_{|I|}))^{10}$ for convenience.
For an integer $x \in [0, |V|]$, we call $f(x) = \min \{ \Delta (G - X) : |X| = x \}$ the \textit{degree deletion function} of $G$, where $X$ is a subset of $V$.
Figure \ref{fig:def-del-table} demonstrates the degree deletion function of a $(a_0, I, c)$-degree deletion graph.  The intuition behind degree deletion graphs is that their deletion function has a stepwise behavior with many steps.

The above notion of an arbitrary solution table works well for $\alpha \in (0, 1]$.  For $\alpha = 0$, we need to replace ``deletion'' with ``retention''.  This is useful for the $\alpha = 0$ case, where the set $X$ must contain at least $\beta$ neighbors of each unchosen vertex.  Hence, the steps of the table describe, for each number $x$ of vertices to include in $X$, 
the maximum possible $\delta(G - X, X)$ that can be achieved with a subset of size $x$.

\begin{figure}
\centering
\begin{tikzpicture}
\draw[-stealth] (-2,-1) -- (7.5,-1);
\draw[color=black] (7.7,-1) node[scale=0.8] {$x$};

\draw[color=black] (-2.2,-1.2) node[scale=0.8] {$0$};

\draw[color=black] (1,-1.3) node[scale=0.7] {$p$};
\draw (1,-1.05) -- (1,-0.95);

\draw[color=black] (2.5,-1.3) node[scale=0.7] {$p + c(a_{|I|-1})$};
\draw (2.5,-1.05) -- (2.5,-0.95);

\fill [color=black] (3.8,-1.3) circle (0.4pt);
\fill [color=black] (4.0,-1.3) circle (0.4pt);
\fill [color=black] (4.2,-1.3) circle (0.4pt);

 \draw[color=black] (5.2,-1.3) node[scale=0.7] {$p + c(a_{1})$};
 \draw (5.2,-1.05) -- (5.2,-0.95);

 \draw[color=black] (6.5,-1.3) node[scale=0.7] {$p+c(a_{0})$};
 \draw (6.5,-1.05) -- (6.5,-0.95);

 \draw[-stealth] (-2,-1) -- (-2,3.3);
 \draw[color=black] (-2,3.5) node[scale=0.8] {$f(x)$};

 \draw[color=black] (-2.5,3) node[scale=0.7] {$a_0$};
 \draw (-2.05,3) -- (-1.95,3);

 \draw[color=black] (-2.5,2.8) node[scale=0.7] {$a_1$};
 \draw (-2.05,2.8) -- (-1.95,2.8);

 \draw[color=black] (-2.5,0.9) node[scale=0.7] {$a_{|I|-1}$}; 
 \draw (-2.05,0.9) -- (-1.95,0.9);

 \fill [color=black] (-2.5,1.65) circle (0.4pt);
 \fill [color=black] (-2.5,1.85) circle (0.4pt);
 \fill [color=black] (-2.5,2.05) circle (0.4pt);

 \draw[color=black] (-2.5,0) node[scale=0.7] {$a_{|I|}$};
 \draw (-2.05,0) -- (-1.95,0);

\draw[color=black] (-2.5,-0.7) node[scale=0.7] {$l$};
\draw (-2.05,-0.7) -- (-1.95,-0.7);


 \fill [color=black] (-2,-1) circle (0.9pt);
 \fill [color=black] (-1.9,-1) circle (0.9pt);
 \fill [color=black] (-1.8,-1) circle (0.9pt);
 \fill [color=black] (-1.7,-1) circle (0.9pt);
 \fill [color=black] (-1.6,-1) circle (0.9pt);
 \fill [color=black] (-1.5,-1) circle (0.9pt);
 \fill [color=black] (-1.4,-1) circle (0.9pt);
 \fill [color=black] (-1.3,-1) circle (0.9pt);
 \fill [color=black] (-1.2,-1) circle (0.9pt);
 \fill [color=black] (-1.1,-1) circle (0.9pt);
 \fill [color=black] (-1,-1) circle (0.9pt);
 \fill [color=black] (-1,-1) circle (0.9pt);
 \fill [color=black] (-0.9,-1) circle (0.9pt);
 \fill [color=black] (-0.8,-1) circle (0.9pt);
 \fill [color=black] (-0.7,-1) circle (0.9pt);
 \fill [color=black] (-0.6,-1) circle (0.9pt);
 \fill [color=black] (-0.5,-1) circle (0.9pt);
 \fill [color=black] (-0.4,-1) circle (0.9pt);
 \fill [color=black] (-0.3,-1) circle (0.9pt);
 \fill [color=black] (-0.2,-1) circle (0.9pt);
 \fill [color=black] (-0.1,-1) circle (0.9pt);
 \fill [color=black] (0,-0.98) circle (0.9pt);
 \fill [color=black] (0.1,-0.98) circle (0.9pt);
 \fill [color=black] (0.2,-0.96) circle (0.9pt);
 \fill [color=black] (0.3,-0.93) circle (0.9pt);
 \fill [color=black] (0.4,-0.9) circle (0.9pt);
 \fill [color=black] (0.5,-0.88) circle (0.9pt);
 \fill [color=black] (0.6,-0.82) circle (0.9pt);
 \fill [color=black] (0.7,-0.78) circle (0.9pt);
 \fill [color=black] (0.8,-0.75) circle (0.9pt);
 \fill [color=black] (0.9,-0.72) circle (0.9pt);

 \fill [color=black] (1,0) circle (0.9pt);
 \fill [color=black] (1.1,0) circle (0.9pt);
 \fill [color=black] (1.2,0) circle (0.9pt);
 \fill [color=black] (1.3,0) circle (0.9pt);
 \fill [color=black] (1.4,0) circle (0.9pt);
 \fill [color=black] (1.5,0) circle (0.9pt);
 \fill [color=black] (1.6,0) circle (0.9pt);
 \fill [color=black] (1.7,0) circle (0.9pt);
 \fill [color=black] (1.8,0) circle (0.9pt);
 \fill [color=black] (1.9,0) circle (0.9pt);
 \fill [color=black] (2,0) circle (0.9pt);
 \fill [color=black] (2.1,0) circle (0.9pt);
 \fill [color=black] (2.2,0) circle (0.9pt);
 \fill [color=black] (2.3,0) circle (0.9pt);
 \fill [color=black] (2.4,0) circle (0.9pt);

 \fill [color=black] (2.5,0.9) circle (0.9pt);
 \fill [color=black] (2.6,0.9) circle (0.9pt);
 \fill [color=black] (2.7,0.9) circle (0.9pt);
 \fill [color=black] (2.8,0.9) circle (0.9pt);
 \fill [color=black] (2.9,0.9) circle (0.9pt);
 \fill [color=black] (3,0.9) circle (0.9pt);

 \fill [color=black] (3.3,1.55) circle (0.4pt);
 \fill [color=black] (3.5,1.55) circle (0.4pt);
 \fill [color=black] (3.7,1.55) circle (0.4pt);

 \fill [color=black] (3.9,1.85) circle (0.4pt);
 \fill [color=black] (4.1,1.85) circle (0.4pt);
 \fill [color=black] (4.3,1.85) circle (0.4pt);

 \fill [color=black] (4.5,2.15) circle (0.4pt);
 \fill [color=black] (4.7,2.15) circle (0.4pt);
 \fill [color=black] (4.9,2.15) circle (0.4pt);

 \fill [color=black] (5.2,2.8) circle (0.9pt);
 \fill [color=black] (5.3,2.8) circle (0.9pt);
 \fill [color=black] (5.4,2.8) circle (0.9pt);
 \fill [color=black] (5.5,2.8) circle (0.9pt);
 \fill [color=black] (5.6,2.8) circle (0.9pt);
 \fill [color=black] (5.7,2.8) circle (0.9pt);
 \fill [color=black] (5.8,2.8) circle (0.9pt);
 \fill [color=black] (5.9,2.8) circle (0.9pt);
 \fill [color=black] (6,2.8) circle (0.9pt);
 \fill [color=black] (6.1,2.8) circle (0.9pt);
 \fill [color=black] (6.2,2.8) circle (0.9pt);
 \fill [color=black] (6.3,2.8) circle (0.9pt);
 \fill [color=black] (6.4,2.8) circle (0.9pt);

 \fill [color=black] (6.5,3) circle (0.9pt);
 \fill [color=black] (6.6,3) circle (0.9pt);
 \fill [color=black] (6.7,3) circle (0.9pt);
 \fill [color=black] (6.8,3) circle (0.9pt);
 \fill [color=black] (6.9,3) circle (0.9pt);
 \fill [color=black] (7,3) circle (0.9pt);
 \fill [color=black] (7.1,3) circle (0.9pt);

 \draw[color=black] (1,-0.5) node[scale=0.5] {$(p-1,l')$};
 \draw[color=black] (1,0.2) node[scale=0.5] {$(p,a_{|I|})$};
 \draw[color=black] (2.5,1.1) node[scale=0.5] {$(p+ c(a_{|I|-1}),a_{|I|-1})$};
 \draw[color=black] (5.2,3.0) node[scale=0.5] {$(p+c(a_1),a_1)$};
\draw[color=black] (6.5,3.2) node[scale=0.5] {$(p+c(a_0),a_0)$};
\end{tikzpicture}
\caption{The degree retention function of a $(a_0, I, c)$-degree retention graph. The first point and the last point of the first curved vertex set are $(0,0)$ and $(p-1,l')$, respectively, where $l' < l$. The first point of the $(i+1)$-th horizontal vertex set is $(p + c(a_{|I|-i}),a_{|I|-i})$ for $0 \leq i \leq |I|$, and the last point of the $(i+1)$-th horizontal vertex set is $(p + c(a_{|I|-i-1}) - 1,a_{|I|-i})$ for $0 \leq i \leq |I| -1$.}
\label{fig:def-ren-table}
\end{figure}
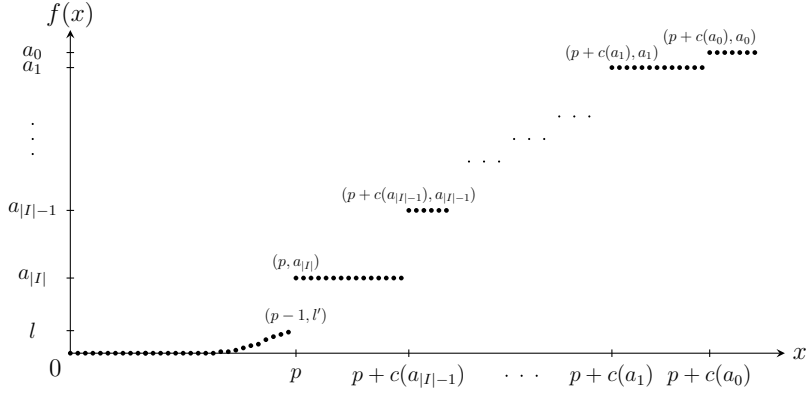

\begin{definition}\label{def:kd-degdeletion}
Let $(a_0, I, c)$ be a triple with $\{a_0\} \cup I \subseteq \mathbb{N}$ and $c : \mathbb{N} \rightarrow \mathbb{N}$.  We call $c$ a cost function.  We say that $(a_0, I, c)$ is \emph{increasing valid} if there is an ordering $a_1, \ldots, a_{|I|}$ of $I$ such that $a_0 > a_1 > \ldots > a_{|I|} > 100$,\footnote{In fact, any large constant here is enough for our proof in this paper, we fix it to be 100 for convenience.} and $c(a_0) > c(a_1) > \ldots > c(a_{|I|}) = 0$.

For a increasing valid triple $(a_0, I, c)$, we say that a graph $G=(V,E)$ is a \emph{$(a_0, I, c)$-degree retention graph above} $(p,l)$, where $p$ and $l$ are positive integers and $l< 100$, if all of the following conditions hold:
\begin{enumerate}
    
    \item 
    $G$ has maximum degree $a_0$, at most $(a_0c(a_{0}))^{10}$ vertices, and $p$ vertices of degree less than $l$;
    
    \item 
    for any $a_i \in \{a_0\} \cup I$, there exists $X \subseteq V$ of size $p + c(a_i)$ such that the minimum degree of $G - X$ in $X$ is $a_i$;

    \item \label{def:kd-degdeletion-3}
   for any $a_i \in \{a_0\} \cup I$ and any $X \subseteq V$ of size strictly less than $p + c(a_i)$, the minimum degree of $G - X$ in $X$ is at most $a_{i + 1}$, where here we define $a_{|I|+1} = l$.
\end{enumerate}
\end{definition}

We say that a graph class $\G$ \emph{admits arbitrary degree retention tables} if, for any increasing valid triple $(a_0, I, c)$, there exist integers $p$ and $l$ such that one can construct in time polynomial in $c(a_0)a_0$ a graph $G \in \G$ such that $G$ is a $(a_0, I, c)$-degree retention graph above $(p,l)$.
Note that condition 2 of Definition \ref{def:kd-degdeletion} imply that $p + c(a_0) < |V| \leq  (c(a_0)a_0)^{10}$.
For any integer $x\in [0, |V|-1]$, we call $f(x) = \max \{ \delta (G - X,X) : |X|= x\}$ the \textit{degree retention function} of $G$, where $X$ is a subset of $V$. 
Figure \ref{fig:def-ren-table} demonstrates the degree retention function of a $(a_0, I, c)$-degree retention graph.
Importantly, \stars~admit both types of tables.

\begin{proposition}\label{prop:stars-tables}
The graph class \stars~admits arbitrary deletion tables and arbitrary retention tables.
\end{proposition}

\begin{proof}
The proof is split into two parts.

\textit{The graph class of disjoint stars admits arbitrary degree deletion tables.}

Let $(a_0, I, c)$ be a decreasing valid triple.  Then we may order $I = \{a_1, \ldots, a_n\}$ so that $a_0 > a_1 > \ldots > a_n >1$ and $0= c(a_0) < c(a_1) < \ldots < c(a_n)$.
We construct a graph of disjoint stars $G$ that has this solution table.
For each $i \in [n]$, add to $G$ a set of $c(a_i) - c(a_{i-1})$ disjoint stars of degree $a_{i - 1}$.  Then, add one star of degree $a_n$.

It is clear that $G$ can be constructed in time polynomial in $a_0c(a_n)$.  
Also, $G$ has maximum degree $a_0$ and has at most $(c(a_n)+1)(a_0+1)$ vertices.  
We argue the second and third parts of Definition~\ref{def:kd-degdeletion-bdd-121} simultaneously.  
Now consider $i \in [n]$.  The number of stars of degree strictly more than $a_i$ is $\sum_{j = 1}^{i} (c(a_j) - c(a_{j-1})) = c(a_i)$.
If we delete the center of each of those stars, we achieve maximum degree $a_i$.  On the other hand, if we make less than $c(a_i)$ deletions, at least one of those stars will remain intact and the maximum degree will be strictly greater than $a_i$.

\textit{
The graph class of disjoint stars admits arbitrary degree retention tables.
}

Let $(a_0, I, c)$ be an increasing valid triple.  Then we may order $I = \{a_1, \ldots, a_n\}$ so that $a_0 > a_1 > \ldots > a_n > 100$ and $c(a_0) > c(a_1) > \ldots > c(a_n) = 0$. Let $p= a_0 + \sum_{i\in [n]} (c(a_{i-1}) - c(a_i))a_i$ and $l= 1$. 
We construct a graph of disjoint stars $G$ that is a $(a_0, I, c)$-degree retention graph above $(p,l)$.
For each $i \in [n]$, add to $G$ a set of $c(a_{i-1}) - c(a_i)$ disjoint stars of degree $a_{i}$.  Then, add one star of degree $a_0$.

It is clear that $G$ can be constructed in time polynomial in $a_0c(a_0)$.  
Also, $G$ has maximum degree $a_0$ and has at most $(c(a_0)+1)(a_0+1)$ vertices. Moreover, the number of vertices with degree at most $l$, which are the leaves of the stars, is $a_0 + \sum_{i\in [n]} (c(a_{i-1}) - c(a_i))a_i$. 
We argue the second and third parts of Definition~\ref{def:kd-degdeletion} simultaneously. We know the total number of leaves of the stars is $p$. Suppose $a_{n+1} = l$ and $c(a_{n+1}) = 0$. Now consider integer $i \in [0,n]$.  The number of stars of degree strictly less than $a_{i}$ is $p+ \sum_{j = i}^{n} (c(a_{j}) - c(a_{j+1})) = p + c(a_{i})$.
If we add the center of each of those stars and all leaves to $X$, the minimum degree of $G - X$ in $X$ is $a_i$, since the left vertices have degrees at least $a_i$ and all their neighbors are in $X$. On the other hand, assume we add less than $p + c(a_i)$ vertices to $X$. Since we have exactly $p + c(a_i)$ vertices with degrees at most $a_{i+1}$, there is a vertex $v$ with degrees at most $a_{i+1}$ that is in $V(G) \setminus X$. Thus, the minimum degree of $G-X$ in $X$ is at most $a_{i+1}$.
\end{proof}

We can now state one of our fundamental results.

\begin{theorem}
\label{ld-main-theorem}
The \ld~problem is W[1]-hard in the following cases:
\begin{enumerate}
    \item 
    $\alpha = 0$, $\beta$ is in the input,
    and the parameter is the $(\G \cup \I)$-\gmparam, where $\G$ is any graph class that admits arbitrary degree retention tables; 
    
    \item 
     $\alpha$ is any fixed constant in the interval $(0,1)$, $\beta$ is any fixed constant in $\mathbb{Z}$, and the parameter is the \stars-\gmparam;
    
    \item 
    $\alpha =1$, $\beta$ is in the input,
    and the parameter is the $(\G \cup \I)$-\gmparam,  where $\G$ is any graph class that admits arbitrary degree deletion tables.
\end{enumerate}
\end{theorem}

The proof of the Theorem \ref{ld-main-theorem} is quite long, and we put it in Appendix. The proof of the three cases all use the same construction and the same set of claims, but most claims require an argument for each case.  The next section illustrates the main techniques on the case $\alpha = 1$.

Together, Proposition \ref{prop:stars-tables} and Theorem \ref{ld-main-theorem}, along with the relationship demonstrated in Figure \ref{fig:parameter-relationship} implies the following corollary. 
Moreover, the \stars-\gmparam~of the output graph of the W[1]-hardness reduction in Theorem \ref{ld-main-theorem} equals the  iterated type partition of the graph.

\begin{corollary} \label{coro-w-hard-itp}
The \ld~problem parameterized by either clique-width, modular-width, \cographs-\gmparam,  iterated type partition, and \stars-\gmparam, is W[1]-hard in the following cases:
\begin{enumerate}

    \item 
    $\alpha = 0$, $\beta$ is in the input;
    
    \item 
     $\alpha$ is any fixed constant in the interval $(0,1)$, $\beta$ is any fixed constant in $\mathbb{Z}$;
    
    \item 
    $\alpha =1$, $\beta$ is in the input.
\end{enumerate}
In particular, \bdd~problem, \kd~problem, and \ad~problem are W[1]-hard in all these parameters.
\end{corollary}
\begin{proof}
For \stars-\gmparam, this follows from Proposition \ref{prop:stars-tables} and Theorem \ref{ld-main-theorem}.
The hardness on \cographs-\gmparam~then follows from the fact that stars are cographs, and the hardness on modular-width follows from the fact that \cographs-\gmparam~is at least as large as the modular-width because cographs have bounded modular-width (Theorem \ref{modularwidth-smaller-than-gmw}). The hardness on clique-width follows from the fact that $cw \leq mw +2$. 

Next, let us consider the parameter iterated type partition, whose definition refers to Definition 1 of \cite{DBLP:conf/iwoca/CordascoGR20}. 
Roughly speaking, the iterated type partition of a graph is defined by iteratively constructing type graphs until the type graph remains the same as the previous one, where the type graph of a graph is the quotient graph of a $\G$-modular partition such that $\G$ includes the complete graphs or the edgeless graphs. The base graph is the last one of the type graph sequence. Obviously, the main reduction works if we replace all $S_i$ and $U_{ij}$ in $H$ with their corresponding star graphs. Then, the type graph sequence of $H$ is $H^{0}$, $H^{1}$, $H^{2}$, $H^{3}$,  where $H^{0} = H$, $H^{1}$ is the graph that replaces every star in $S_i$ and $U_{ij}$ with a $K_2$ graph, and replaces every other edgeless factor with a vertex; $H^{2}$ is the graph that replaces every $K_2$ in $S_i$ and $U_{ij}$ of $H^{1}$ with a vertex; the base graph $H^{3}$ is the graph that replaces every $S_i$ and $U_{ij}$ of $H^{2}$ with a vertex. Clearly, the vertex number of $H^{3}$, which is iterated type partition of $H$, equals \stars-\gmparam~of $H$ and is $O(k^2)$.
\end{proof}

The next proposition is not related to the above hardness results but allows us to fill in some of the gaps that the above leaves in our results table.  In section \ref{FPT-algo-for-suc-section}, we will also show that 
for $\alpha = 1$ and $\beta$ in the input  (the $\bdd$ case), the problem is FPT in \clusters-\gmparam. Note that the design of FPT algorithm for $\bdd$ parameterized by \clusters-\gmparam~in section \ref{FPT-algo-for-suc-section} is of particular interest, the solution table for that case is non-convex and contains a larger number of blocks, so it seems that the solution table can not be succinct at first glance. But, we can reduce this special non-convex function to a convex function using ceilings to get the FPT algorithm.

\begin{proposition}
The \ld~problem is FPT in the following cases:
\begin{enumerate}
\item 
$\alpha \in [0,1]$, $\beta$ is in the input, and the parameter is the neighborhood diversity;
    
\item 
$\alpha \in \{0,1\}$, $\beta$ is a constant, and the parameter is the clique-width.
\end{enumerate}
\end{proposition}

The first case requires some work, whereas the second is a simple application of  Courcelle's theorem \cite{DBLP:journals/mst/CourcelleMR00}, as the problem admits a constant-length MSO$_1$ formula.

We first show that \ld~is FPT in parameter $nd$ when $\alpha \in [0 , 1]$ and $\beta$ is in the input.
Recall that the minimum neighborhood partition (neighborhood diversity) can be obtained in polynomial time.
To obtain our result, we need the following FPT algorithm for the \ilp~problem parameterized by the dimensions.

\begin{theorem}[\cite{DBLP:journals/combinatorica/FrankT87,DBLP:journals/mor/Kannan87,DBLP:journals/mor/Lenstra83}]
\label{ilp is FPT}
The \ilp~problem can be decided in $O(n^{2.5n+o(n)}\cdot I)$ time, where $I$ is the size of the input and $n$ is the number of variables.
\end{theorem}

\begin{lemma} \label{fpt-algorithm-para-nd}
\ld~problem is FPT parameterized by neighborhood diversity even if $\beta$ is a part of the input.
\end{lemma}
\begin{proof}
Suppose  $(G,q,\beta)$ is an input of 
the \ld~problem, where $G=(V,E)$. Compute the minimum neighborhood partition $\P = \{ M_1,\ldots, M_{r},$ $ M_{r+1}, \ldots, M_{k} \}$ of $V$ in polynomial time, where each $M_i$ for $i\in [r]$ is a clique, and each $M_i$ for $i\in [r+1,k]$ is an independent set. For a module $M_i \in \P$, assume that $N(M_i)$ consists of all modules adjacent to $M_i$, and that $N[M_i] = N(M_i) \cup \{M_i\}$. Suppose the vertex number of the intersection of $M_i$ and $X$ is $x_i$. In addition, for any integer $a$ and real number $b$, we have that $a \geq b$ if and only if $a \geq \left \lceil b \right \rceil$. Then, it is not hard to verify that $(G,q,\beta)$ is a yes instance of \ld~problem if and only if the following instance is yes for \ilp.
\begin{align*}
&(1) \: \sum_{1\leq i \leq k} x_{i} \leq q &\\
&(2) \: \sum_{M_j \in N[M_i]} x_j \geq \left \lceil \alpha (  \sum_{M_j \in N[M_i]} |M_j| - 1)\right \rceil +\beta & \forall i \in [0, r-1]\\
&(3) \: \sum_{M_j \in N(M_i)} x_j \geq  \left \lceil \alpha   \sum_{M_j \in N(M_i)} |M_j| \right \rceil +\beta & \forall i \in [r, r+s] \\
&(4) \: x_i \in \mathbb{N} & \forall i \in [k]\\
&(5) \: 0 \leq x_i \leq |M_i| & \forall  i \in [k]\\
\end{align*}
Clearly, the \ilp~has $k = nd(G)$ variables. Thus, \ld~problem is in FPT based on Theorem \ref{ilp is FPT}.
\end{proof}

Let us now focus on showing that for $\alpha \in \{0, 1\}$ and fixed $\beta$, the problem is FPT in clique-width.
This is easy using Courcelle's theorem \cite{DBLP:journals/mst/CourcelleMR00}.
For each $\alpha \in \{0,1\}$, we can write a constant-length MSO$_1$ formula with free variable $X$, denoted by $(\alpha, \beta)$-\textbf{LDD}$(X)$, to verify that $X$ is a $(\alpha, \beta)$-linear degree dominating set of $G$ as follows: (1) $(0, \beta)$-\textbf{LDD}(X) $=
\forall_{v\in V\setminus X} \exists_{v_1 \in X} \ldots $ $\exists_{v_{\beta} \in X} (\textbf{adj}(v,v_1) \land \ldots \land \textbf{adj}(v,v_{\beta}))$, and (2) $(1, \beta)$-\textbf{LDD}(X) $= \neg(\exists_{v\in V\setminus X} \exists_{v_0 \in V\setminus X} \ldots $ $\exists_{v_{|\beta|} \in V\setminus X} (\textbf{adj}(v,v_0) \land \ldots \land \textbf{adj}(v,v_{|\beta|})))$. According to the optimization version of the Courcelle's theorem \cite{DBLP:journals/mst/CourcelleMR00}, there exists an FPT algorithm parameterized by clique-width to minimize $|X|$ subject to $(\alpha, \beta)$-\textbf{LDD}$(X)$ for each $\alpha$.


\section{\boldmath\texorpdfstring{$(1, \beta)$}{}-\textsc{Linear Degree Domination} (BDD)}

We sketch the proof of the W[1]-hardness of $(1, \beta)$-\textsc{Linear Degree Domination} problem, which is enough to demonstrate the main idea of the reduction technique. Recall that we assume that $\alpha = 1$ and $\beta \leq 0$, which is the~\bdd~problem. That is, we must delete at most $q$ vertices such that the resulting subgraph has maximum degree at most $|\beta|$. 
 Furthermore, in the $(1, \beta)$-\textsc{Linear Degree Domination} problem, we also call $X$, the $(1, \beta)$-linear degree dominating set of $G$, the deletion part of $G$.




We provide a reduction from the \mcc~problem, which we define as follows.  A symmetric multicolored graph $G = (V^1\cup V^2 \ldots \cup V^k , E)$ is a connected 
graph such that, for all distinct $i,j \in [k]$,
\begin{enumerate}
    \item
    $V^i = \{v_1^i, \ldots, v_n^i\}$, where $n\geq k$;
    
    \item 
    all the vertices of $V^i$ are colored by color $i$;

    \item
    if $v_r^i v_s^j \in E(G)$, then $v_s^j v_r^i \in E(G)$ as well.
\end{enumerate}
Then, for the \mcc~problem, the input is a symmetry multicolored graph $G$ and an integer $k$, and the objective is to decide whether $G$ contains a $k$-clique with vertices of all $k$ colors. We also call $v_r^i v_s^j$ and $v_s^i v_r^j$ symmetry edges. 

The reduction in~\cite[Lemma~1]{DBLP:journals/tcs/FellowsHRV09}, which proves the W[1]-hardness of the multicolored clique problem, actually produces a symmetric multicolored graph.  Hence, \mcc~is W[1]-hard.  We now sketch the following.

\begin{lemma}
Case 3 for the W[1]-hardness results in Theorem \ref{ld-main-theorem} is correct. 
\end{lemma}

Let $(G,k)$ be an instance of \mcc, where  $G = (V^1\cup V^2 \cup \ldots \cup V^k , E)$. Without loss of generality, suppose $k\geq 100$, otherwise, the problem can be solved in polynomial time.
We will construct a corresponding instance $(H, \beta, q)$ of \textsc{$(1,\beta)$-Linear Degree Domination}, where $H$ is a graph whose $\G$-\gmparam~will be bounded by $O(k^2)$, $\beta= -(nk)^{10000}$, and $q$ is the maximum allowed size of $X$, the desired deletion part of $H$ (to be specified later). 
Before proceeding, we will make use of a $2$-sumfree-set, which is a set of positive integers in which every couple of elements has a distinct sum. That is, $I$ is a $2$-sumfree set if, for any $(a, b), (a', b') \in I \times I$, $a + b = a' + b'$ if and only if $\{a, b\} = \{a', b'\}$ (note that $a = b$ is possible).   
It is known that one can construct in time $O(n^3)$ a $2$-sumfree-set $I$ of cardinality $n$ in which the maximum value is $n^4$, which can be achieved with a greedy procedure (this is because $(n+1)^4-n^4 > n^3$ and $a_i+a_j-a_r$ has at most $n^3$ different values).  
We thus assume that we have built a $2$-sumfree set $I = \{a_1, \ldots, a_n\}$ for $H$, where $n^4 \geq a_1 > \ldots > a_n \geq 1$. Without loss of generality, we may multiply each $a_i \in I$ by an integer $r$, where here we choose $r = 2(k-1)^2k^3$. 
Then, we have that $n^4r \geq a_1 > \ldots > a_n = r$ for the updated $I$. 
Moreover, for any distinct pair $(a,b),(a',b')\in I \times I$, we have that the absolute value of $a+b-a'-b'$ is at least $r$.

For  each color class $V^i = \{v_1^i, \ldots, v_n^i\}$ of $G$, we provide a bijection $f_{i}$ from $V^i$ to $I$, such that $f_{i}(v^i_s) = a_s$ for every $s\in [n]$.
Clearly, we can use $f_{i}^{-1}(a_s)$ to denote the unique vertex $v^i_s$ of $V^i$ associated with $a_s \in I$. We have that, for all $s, t\in [n]$,  each pair of $a_s, a_t \in I$ has a unique sum. For distinct $i, j$ and any $u \in V^i, w \in V^j$ such that $uw \in E(G)$, if  $f_{i}(u) + f_{j}(w) = a_s + a_t$, then edge $uw$ is either $v^i_s v^j_t$ or $v^i_t v^j_s$. Moreover, for any distinct color classes $V^i, V^j$, edges $v^i_s v^j_t$ and $v^i_t v^j_s$, together, are both in $E(G)$ or both not in $E(G)$. Hence, by looking at a sum in $I$, we will be able to tell whether it is corresponding to a pair of symmetry edges between $V^i$ and $V^j$, or not.

Next, we define $s = n^{10}$, $q = ks + \binom{k}{2} s$, $a_0= q +1$, and $a_{n+1}= a_{n} - 1$. 
We construct $H$ in Figure \ref{fig:bdd} as follows. First, for each color class $V^i$ of $G$, add three factors  $ R_i,$ $ S_i,$ $ T_i$ to   $H $, where:
\begin{itemize}
  
    \item 
    $R_i$ is an edgeless graph of size $|\beta| -s$;

    \item 
      $S_i$ is a $(a_0, I \cup \{a_{n + 1}\}, c)$-degree deletion graph, where we put the costs $c(a_j) = s -  \frac{1}{2}a_j $ for $a_j \in I$ and $c(a_{n + 1}) = a_0$.   
 
    \item 
    $T_i$ is an edgeless graph of size $s$.
    
\end{itemize}
We then make $S_i$ adjacent with $R_i$ and $T_i$.
Secondly, for each pair of color classes $V^i, V^j$ with $i < j$, we add another two factors $U_{ij},  R_{ij}$, where:
\begin{itemize}
    \item 
    $R_{ij}$ is an edgeless graph of size $|\beta| - 2s$;

    \item 
    $U_{ij}$ is built as follows.
    Suppose integer set $I_{ij}$ consists of all $a + b$ such that $a$, $b \in$ $I$ and symmetry edges  $f_{i}^{-1}(b) f_{j}^{-1}(a) , f_{i}^{-1}(a) f_{j}^{-1}(b) \in E(G)$. Let $\ell_{ij} = \min(I_{ij})$.
    
   Then $U_{ij}$ is a ${(a_0, I_{ij} \cup \{\ell_{ij} - 1\}, c_{ij})}$-degree deletion graph, where we put the cost $c_{ij}(a + b) = s -  \frac{1}{2(k-1)} ( a+b )$, and we put $c_{ij}(\ell_{ij} - 1) = a_0$.

\end{itemize}
We then make $U_{ij}$ adjacent with $R_{ij}$, and adjacent with $T_i$ and $T_j$.  To avoid cumbersome notation, we define $U_{ij} = U_{ji}$, $R_{ij} = R_{ji}$, $c_{ij} = c_{ji}$, and $I_{ij} = I_{ji}$.
This completes the construction of $H$.
It is easy to see that  $H $ can be constructed in polynomial time.

\begin{figure}
\centering
\begin{tikzpicture}
\filldraw[color=red, fill=red!3] (0,0) ellipse (1 and 0.618); 
\filldraw[color=red, fill=red!3] (6,0) ellipse (1 and 0.618); 
\filldraw[color=blue, fill=blue!5] (-0.8,1.5) rectangle (0.8,2.5); 
\filldraw[color=blue, fill=blue!5]  (5.2,1.5) rectangle (6.8,2.5); 
\filldraw[color=blue, fill=blue!5] (0,3.7) ellipse (0.618 and 0.382); 
\filldraw[color=blue, fill=blue!5] (6,3.7) ellipse (0.618 and 0.382); 
\filldraw[color=blue, fill=blue!5]  (2.1,4.05) rectangle (3.9,5.15); 
\filldraw[color=red, fill=red!3]   (3,6.5) ellipse (1 and 0.618); 

\draw[color=black] (-1.5,0) node {$R_i$};
\draw[color=black] (7.5,0) node {$R_j$};
\draw[color=black] (-1.5,2) node {$S_i$};
\draw[color=black] (7.5,2) node {$S_j$};
\draw[color=black] (-1.2,3.7) node {$T_i$};
\draw[color=black] (7.2,3.7) node {$T_j$};
\draw[color=black] (3,3.7) node {$U_{ij}$};
\draw[color=black] (3,7.4) node {$R_{ij}$};

\fill [color=black] (-0.4,0) circle (1pt);
\fill [color=black] (-0.2,0) circle (1pt);
\fill [color=black] (0,0) circle (0.4pt);
\fill [color=black] (0.1,0) circle (0.4pt);
\fill [color=black] (0.2,0) circle (0.4pt);
\fill [color=black] (0.4,0) circle (1pt);

\fill [color=black] (6-0.4,0) circle (1pt);
\fill [color=black] (6-0.2,0) circle (1pt);
\fill [color=black] (6,0) circle (0.4pt);
\fill [color=black] (6.1,0) circle (0.4pt);
\fill [color=black] (6.2,0) circle (0.4pt);
\fill [color=black] (6.4,0) circle (1pt);

\fill [color=black] (-0.4,2.4) circle (1pt);
\fill [color=black] (-0.7,2.1) circle (1pt);
\draw (-0.4,2.4) -- (-0.7,2.1);
\fill [color=black] (-0.56,2.1) circle (1pt);
\draw (-0.4,2.4) -- (-0.56,2.1);
\fill [color=black] (-0.43,2.1) circle (0.4pt);
\fill [color=black] (-0.32,2.1) circle (0.4pt);
\fill [color=black] (-0.21,2.1) circle (0.4pt);
\fill [color=black] (-0.1,2.1) circle (1pt);
\draw (-0.4,2.4) -- (-0.1,2.1);

\fill [color=black] (0.4,2.4) circle (1pt);
\fill [color=black] (0.1,2.1) circle (1pt);
\draw (0.4,2.4) -- (0.1,2.1);
\fill [color=black] (0.26,2.1) circle (1pt);
\draw (0.4,2.4) -- (0.26,2.1);
\fill [color=black] (0.37,2.1) circle (0.4pt);
\fill [color=black] (0.48,2.1) circle (0.4pt);
\fill [color=black] (0.59,2.1) circle (0.4pt);
\fill [color=black] (0.7,2.1) circle (1pt);
\draw (0.4,2.4) -- (0.7,2.1);

\fill [color=black] (-0.29,1.7) circle (0.4pt);
\fill [color=black] (-0.4,1.7) circle (0.4pt);
\fill [color=black] (-0.51,1.7) circle (0.4pt);

\fill [color=black] (0.4,1.9) circle (1pt);
\fill [color=black] (0.1,1.6) circle (1pt);
\draw (0.4,1.9) -- (0.1,1.6);
\fill [color=black] (0.26,1.6) circle (1pt);
\draw (0.4,1.9) -- (0.26,1.6);
\fill [color=black] (0.37,1.6) circle (0.4pt);
\fill [color=black] (0.48,1.6) circle (0.4pt);
\fill [color=black] (0.59,1.6) circle (0.4pt);
\fill [color=black] (0.7,1.6) circle (1pt);
\draw (0.4,1.9) -- (0.7,1.6);

\fill [color=black] (5.6,2.4) circle (1pt);
\fill [color=black] (5.3,2.1) circle (1pt);
\draw (5.6,2.4) -- (5.3,2.1);
\fill [color=black] (6-0.56,2.1) circle (1pt);
\draw (5.6,2.4) -- (6-0.56,2.1);
\fill [color=black] (6-0.43,2.1) circle (0.4pt);
\fill [color=black] (6-0.32,2.1) circle (0.4pt);
\fill [color=black] (6-0.21,2.1) circle (0.4pt);
\fill [color=black] (6-0.1,2.1) circle (1pt);
\draw (5.6,2.4) -- (5.9,2.1);

\fill [color=black] (6.4,2.4) circle (1pt);
\fill [color=black] (6.1,2.1) circle (1pt);
\draw (6.4,2.4) -- (6.1,2.1);
\fill [color=black] (6.26,2.1) circle (1pt);
\draw (6.4,2.4) -- (6.26,2.1);
\fill [color=black] (6.37,2.1) circle (0.4pt);
\fill [color=black] (6.48,2.1) circle (0.4pt);
\fill [color=black] (6.59,2.1) circle (0.4pt);
\fill [color=black] (6.7,2.1) circle (1pt);
\draw (6.4,2.4) -- (6.7,2.1);

\fill [color=black] (6-0.29,1.7) circle (0.4pt);
\fill [color=black] (6-0.4,1.7) circle (0.4pt);
\fill [color=black] (6-0.51,1.7) circle (0.4pt);

\fill [color=black] (6.4,1.9) circle (1pt);
\fill [color=black] (6.1,1.6) circle (1pt);
\draw (6.4,1.9) -- (6.1,1.6);
\fill [color=black] (6.26,1.6) circle (1pt);
\draw (6.4,1.9) -- (6.26,1.6);
\fill [color=black] (6.37,1.6) circle (0.4pt);
\fill [color=black] (6.48,1.6) circle (0.4pt);
\fill [color=black] (6.59,1.6) circle (0.4pt);
\fill [color=black] (6.7,1.6) circle (1pt);
\draw (6.4,1.9) -- (6.7,1.6);

\fill [color=black] (-0.4,3.7) circle (1pt);
\fill [color=black] (-0.2,3.7) circle (1pt);
\fill [color=black] (0,3.7) circle (0.4pt);
\fill [color=black] (0.1,3.7) circle (0.4pt);
\fill [color=black] (0.2,3.7) circle (0.4pt);
\fill [color=black] (0.4,3.7) circle (1pt);

\fill [color=black] (6-0.4,3.7) circle (1pt);
\fill [color=black] (6-0.2,3.7) circle (1pt);
\fill [color=black] (6,3.7) circle (0.4pt);
\fill [color=black] (6.1,3.7) circle (0.4pt);
\fill [color=black] (6.2,3.7) circle (0.4pt);
\fill [color=black] (6.4,3.7) circle (1pt);

\fill [color=black] (3-0.4,2.6+2.4) circle (1pt);
\fill [color=black] (3-0.7,2.6+2.1) circle (1pt);
\draw (3-0.4,2.6+2.4) -- (3-0.7,2.6+2.1);
\fill [color=black] (3-0.56,2.6+2.1) circle (1pt);
\draw (3-0.4,2.6+2.4) -- (3-0.56,2.6+2.1);
\fill [color=black] (3-0.43,2.6+2.1) circle (0.4pt);
\fill [color=black] (3-0.32,2.6+2.1) circle (0.4pt);
\fill [color=black] (3-0.21,2.6+2.1) circle (0.4pt);
\fill [color=black] (3-0.1,2.6+2.1) circle (1pt);
\draw (3-0.4,2.6+2.4) -- (3-0.1,2.6+2.1);

\fill [color=black] (3.4,2.6+2.4) circle (1pt);
\fill [color=black] (3.1,2.6+2.1) circle (1pt);
\draw (3.4,2.6+2.4) -- (3.1,2.6+2.1);
\fill [color=black] (3.26,2.6+2.1) circle (1pt);
\draw (3.4,2.6+2.4) -- (3.26,2.6+2.1);
\fill [color=black] (3.37,2.6+2.1) circle (0.4pt);
\fill [color=black] (3.48,2.6+2.1) circle (0.4pt);
\fill [color=black] (3.59,2.6+2.1) circle (0.4pt);
\fill [color=black] (3.7,2.6+2.1) circle (1pt);
\draw (3.4,2.6+2.4) -- (3.7,2.6+2.1);

\fill [color=black] (3-0.29,2.6+1.7) circle (0.4pt);
\fill [color=black] (3-0.4,2.6+1.7) circle (0.4pt);
\fill [color=black] (3-0.51,2.6+1.7) circle (0.4pt);

\fill [color=black] (3.4,2.6+1.9) circle (1pt);
\fill [color=black] (3.1,2.6+1.6) circle (1pt);
\draw (3.4,2.6+1.9) -- (3.1,2.6+1.6);
\fill [color=black] (3.26,2.6+1.6) circle (1pt);
\draw (3.4,2.6+1.9) -- (3.26,2.6+1.6);
\fill [color=black] (3.37,2.6+1.6) circle (0.4pt);
\fill [color=black] (3.48,2.6+1.6) circle (0.4pt);
\fill [color=black] (3.59,2.6+1.6) circle (0.4pt);
\fill [color=black] (3.7,2.6+1.6) circle (1pt);
\draw (3.4,2.6+1.9) -- (3.7,2.6+1.6);

\fill [color=black] (3-0.4,6.5) circle (1pt);
\fill [color=black] (3-0.2,6.5) circle (1pt);
\fill [color=black] (3,6.5) circle (0.4pt);
\fill [color=black] (3.1,6.5) circle (0.4pt);
\fill [color=black] (3.2,6.5) circle (0.4pt);
\fill [color=black] (3.4,6.5) circle (1pt);

\draw[dashed]  (-0.3,0.6) -- (-0.3,1.5);
\draw  (-0.3,0.6) -- (-0.2,0.75);
\draw  (-0.3,0.6) -- (-0.1,0.75);
\draw  (-0.3,1.5) -- (-0.2,1.35);
\draw  (-0.3,1.5) -- (-0.1,1.35);
\fill [color=black] (-0.1,1.2) circle (0.4pt);
\fill [color=black] (0,1.2) circle (0.4pt);
\fill [color=black] (0.1,1.2) circle (0.4pt);
\fill [color=black] (-0.1,0.9) circle (0.4pt);
\fill [color=black] (0,0.9) circle (0.4pt);
\fill [color=black] (0.1,0.9) circle (0.4pt);
\draw[dashed]  (0.3,0.6) -- (0.3,1.5);
\draw  (0.3,0.6) -- (0.2,0.75);
\draw  (0.3,0.6) -- (0.1,0.75);
\draw  (0.3,1.5) -- (0.2,1.35);
\draw  (0.3,1.5) -- (0.1,1.35);

\draw[dashed]  (6-0.3,0.6) -- (6-0.3,1.5);
\draw  (6-0.3,0.6) -- (6-0.2,0.75);
\draw  (6-0.3,0.6) -- (6-0.1,0.75);
\draw  (6-0.3,1.5) -- (6-0.2,1.35);
\draw  (6-0.3,1.5) -- (6-0.1,1.35);
\fill [color=black] (6-0.1,1.2) circle (0.4pt);
\fill [color=black] (6,1.2) circle (0.4pt);
\fill [color=black] (6.1,1.2) circle (0.4pt);
\fill [color=black] (6-0.1,0.9) circle (0.4pt);
\fill [color=black] (6,0.9) circle (0.4pt);
\fill [color=black] (6.1,0.9) circle (0.4pt);
\draw[dashed]  (6.3,0.6) -- (6.3,1.5);
\draw  (6.3,0.6) -- (6.2,0.75);
\draw  (6.3,0.6) -- (6.1,0.75);
\draw  (6.3,1.5) -- (6.2,1.35);
\draw  (6.3,1.5) -- (6.1,1.35);

\draw[dashed]  (-0.3,2.5) -- (-0.3,3.37);
\draw  (-0.3,0.6+1.9) -- (-0.2,0.75+1.9);
\draw  (-0.3,0.6+1.9) -- (-0.1,0.75+1.9);
\draw  (-0.3,3.37) -- (-0.2,3.2);
\draw  (-0.3,3.37) -- (-0.1,3.2);
\fill [color=black] (-0.1,1.2+1.87) circle (0.4pt);
\fill [color=black] (0,1.2+1.87) circle (0.4pt);
\fill [color=black] (0.1,1.2+1.87) circle (0.4pt);
\fill [color=black] (-0.1,0.9+1.9) circle (0.4pt);
\fill [color=black] (0,0.9+1.9) circle (0.4pt);
\fill [color=black] (0.1,0.9+1.9) circle (0.4pt);
\draw[dashed]  (0.3,0.6+1.9) -- (0.3,1.5+1.87);
\draw  (0.3,0.6+1.9) -- (0.2,0.75+1.9);
\draw  (0.3,0.6+1.9) -- (0.1,0.75+1.9);
\draw  (0.3,1.5+1.87) -- (0.2,1.35+1.87);
\draw  (0.3,1.5+1.87) -- (0.1,1.35+1.87);

\draw[dashed]  (6-0.3,2.5) -- (6-0.3,3.37);
\draw  (6-0.3,0.6+1.9) -- (6-0.2,0.75+1.9);
\draw  (6-0.3,0.6+1.9) -- (6-0.1,0.75+1.9);
\draw  (6-0.3,3.37) -- (6-0.2,3.2);
\draw  (6-0.3,3.37) -- (6-0.1,3.2);
\fill [color=black] (6-0.1,1.2+1.87) circle (0.4pt);
\fill [color=black] (6,1.2+1.87) circle (0.4pt);
\fill [color=black] (6.1,1.2+1.87) circle (0.4pt);
\fill [color=black] (6-0.1,0.9+1.9) circle (0.4pt);
\fill [color=black] (6,0.9+1.9) circle (0.4pt);
\fill [color=black] (6.1,0.9+1.9) circle (0.4pt);
\draw[dashed]  (6.3,0.6+1.9) -- (6.3,1.5+1.87);
\draw  (6.3,0.6+1.9) -- (6.2,0.75+1.9);
\draw  (6.3,0.6+1.9) -- (6.1,0.75+1.9);
\draw  (6.3,1.5+1.87) -- (6.2,1.35+1.87);
\draw  (6.3,1.5+1.87) -- (6.1,1.35+1.87);

\draw[dashed]  (0.6,3.6) -- (2.1,4.4);
\draw  (0.6,3.6) -- (0.85,3.81);
\draw  (0.6,3.6) -- (0.8,3.85);
\draw  (1.8,4.3) -- (2.1,4.4);
\draw  (1.8,4.37) -- (2.1,4.4);
\fill[color=black] (1.02,3.9) circle (0.4pt);
\fill[color=black] (0.97,4) circle (0.4pt);
\fill[color=black] (0.92,4.1) circle (0.4pt);
\draw[dashed]  (0.48,3.935) -- (2.1,4.8);
\draw  (0.48,3.935) -- (0.8,4.02);
\draw  (0.48,3.935) -- (0.8,3.95);
\draw  (1.82,4.49) -- (2.1,4.8);
\draw  (1.8,4.57) -- (2.1,4.8);
\fill[color=black] (1.6,4.448) circle (0.4pt);
\fill[color=black] (1.65,4.348) circle (0.4pt);
\fill[color=black] (1.7,4.248) circle (0.4pt);

\draw[dashed]  (6-0.6,3.6) -- (6-2.1,4.4);
\draw  (6-0.6,3.6) -- (6-0.85,3.81);
\draw  (6-0.6,3.6) -- (6-0.8,3.85);
\draw  (6-1.8,4.3) -- (6-2.1,4.4);
\draw  (6-1.8,4.37) -- (6-2.1,4.4);
\fill[color=black] (6-1.02,3.9) circle (0.4pt);
\fill[color=black] (6-0.97,4) circle (0.4pt);
\fill[color=black] (6-0.92,4.1) circle (0.4pt);
\draw[dashed]  (6-0.48,3.935) -- (6-2.1,4.8);
\draw  (6-0.48,3.935) -- (6-0.8,4.02);
\draw  (6-0.48,3.935) -- (6-0.8,3.95);
\draw  (6-1.82,4.49) -- (6-2.1,4.8);
\draw  (6-1.8,4.57) -- (6-2.1,4.8);
\fill[color=black] (6-1.6,4.448) circle (0.4pt);
\fill[color=black] (6-1.65,4.348) circle (0.4pt);
\fill[color=black] (6-1.7,4.248) circle (0.4pt);

\draw[dashed]  (3-0.3,5.15) -- (3-0.3,5.91);
\draw  (3-0.3,0.6+4.55) -- (3-0.2,0.75+4.55);
\draw  (3-0.3,0.6+4.55) -- (3-0.1,0.75+4.55);
\draw  (3-0.3,1.5+4.41) -- (3-0.2,1.35+4.41);
\draw  (3-0.3,1.5+4.41) -- (3-0.1,1.35+4.41);
\fill [color=black] (3-0.1,1.2+4.41) circle (0.4pt);
\fill [color=black] (3,1.2+4.41) circle (0.4pt);
\fill [color=black] (3.1,1.2+4.41) circle (0.4pt);
\fill [color=black] (3-0.1,0.9+4.55) circle (0.4pt);
\fill [color=black] (3,0.9+4.55) circle (0.4pt);
\fill [color=black] (3.1,0.9+4.55) circle (0.4pt);
\draw[dashed]  (3.3,0.6+4.55) -- (3.3,1.5+4.41);
\draw  (3.3,0.6+4.55) -- (3.2,0.75+4.55);
\draw  (3.3,0.6+4.55) -- (3.1,0.75+4.55);
\draw  (3.3,1.5+4.41) -- (3.2,1.35+4.41);
\draw  (3.3,1.5+4.41) -- (3.1,1.35+4.41);

\draw  (-0.48,3.935) -- (-0.7,4.2);
\draw  (-0.48,3.935) -- (-0.7,4.1);
\draw  (-0.48,3.935) -- (-0.6,4.2);
\draw  (-0.28,4.042) -- (-0.3,4.3);
\draw  (-0.28,4.042) -- (-0.4,4.3);
\draw  (-0.28,4.042) -- (-0.45,4.22);
\draw  (0.2,4.055) -- (0.05,4.35);
\draw  (0.2,4.055) -- (0.15,4.35);
\draw  (0.2,4.055) -- (0.25,4.35);
\fill [color=black] (0.02,4.188) circle (0.4pt);
\fill [color=black] (-0.09,4.183) circle (0.4pt);
\fill [color=black] (-0.2,4.165) circle (0.4pt);

\draw  (6.48,3.935) -- (6.7,4.2);
\draw  (6.48,3.935) -- (6.7,4.1);
\draw  (6.48,3.935) -- (6.6,4.2);
\draw  (6.28,4.042) -- (6.3,4.3);
\draw  (6.28,4.042) -- (6.4,4.3);
\draw  (6.28,4.042) -- (6.45,4.22);
\draw  (6-0.2,4.055) -- (6-0.05,4.35);
\draw  (6-0.2,4.055) -- (6-0.15,4.35);
\draw  (6-0.2,4.055) -- (6-0.25,4.35);
\fill [color=black] (6-0.02,4.188) circle (0.4pt);
\fill [color=black] (6.09,4.183) circle (0.4pt);
\fill [color=black] (6.2,4.165) circle (0.4pt);
\end{tikzpicture}
\caption{Illustration of the construction of $H$.}
\label{fig:bdd}
\end{figure}
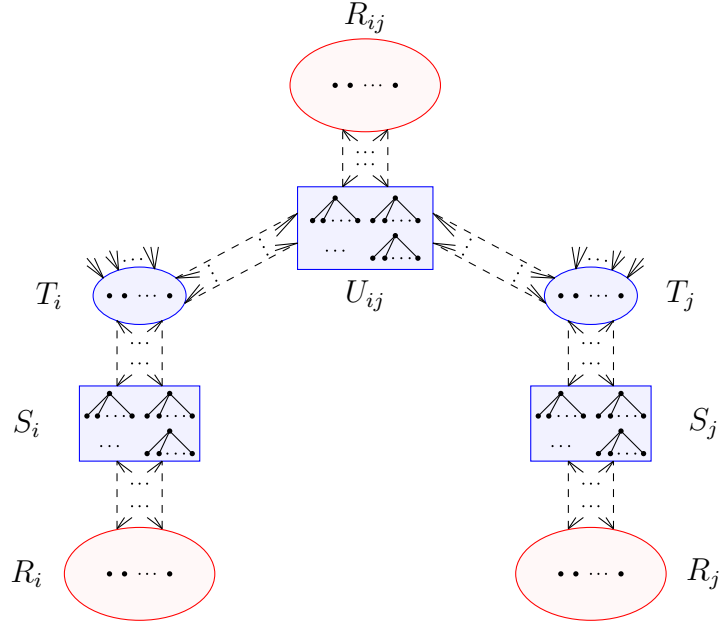

All that remains is to argue that the objective instance is equivalent to the original one.
The following lemma gives the forward direction, whose detailed proof can be found in case 2 of Lemma \ref{ldlem:first-direction} in the Appendix. We only provide a sketch here.  

\begin{lemma}\label{lem:first-direction}
Suppose that $G$ contains a multicolored clique $C$ of size $k$.  Then there is $X \subseteq V(H)$ of size at most $q = ks + \binom{k}{2} s$ such that $\Delta(H - X) \leq |\beta|$.
\end{lemma}

\begin{proofsketch}
For $i \in [k]$, let $f^{-1}_i(\ha_i)$ be the vertex of $V^i$ that belongs to the multicolored clique $C$, where $\ha_i \in I$ is the number associated with the vertex. For any $i\neq j$, we know that vertices $f^{-1}_i(\ha_i)$ and $f^{-1}_j(\ha_j)$ are in $C$, which means that $f^{-1}_i(\ha_i) f^{-1}_j(\ha_j) \in E(G)$.  
This implies that $\ha_i + \ha_j$ is in the $I_{ij}$ list that was used to construct $U_{ij}$.  
We then show how to construct the vertex set $X$ for $H$.
 The intersection of each $R_i$ and $X$ is empty. 
For each $T_i$, add  $\ha_i$ vertices to $X$.
For each $S_i$, add $c(\ha_i) = s - \frac{1}{2}\ha_i$ vertices to $X$.  This can be done so that $S_i - X$ has maximum degree $\ha_i$ according to Definition \ref{def:kd-degdeletion-bdd-121}.
The intersection of each $R_{ij}$ and $X$ is empty. 
For each $U_{ij}$, we can add
    $c_{ij}(\ha_i + \ha_j) =  s -  \frac{\ha_i   +    \ha_j }{2(k-1)}$
vertices to $X$ so that  $U_{ij}- X$ has maximum degree $\ha_i + \ha_j$ according to Definition \ref{def:kd-degdeletion-bdd-121}. It is not hard to verify that $X$ with exactly $q$ vertices satisfies that $\Delta(H - X) \leq |\beta|$.
\end{proofsketch}

The converse direction is much more difficult.

\begin{lemma}\label{lem:second-dir}
Suppose that there is $X \subseteq V(H)$ with $|X| \leq q$ such that $H - X$ has maximum degree at most $|\beta|$. 
Then $G$ contains a multicolored clique of size $k$.
\end{lemma}
\noindent \textit{Sketch:} Let $X \subseteq V(H)$ be of size at most $q$ such that $H - X$ has maximum degree $|\beta|$ or less.
To ease notation slightly, for a factor $M$ of $H$, we will write $X(M) := X \cap V(M)$ and $\chi(M) := |X(M)|$.
The proof is divided into a series of claims.
\begin{claim}\label{cl:nori}
For each $i \in [k]$, we may assume that $\chi(R_i) = 0$.  Moreover for each distinct $i, j \in [k]$, we may assume that $\chi(R_{ij}) = 0$.
\end{claim}

The detailed proof of this claim refers to case 2 of Claims \ref{cl:ld-module-4}, \ref{cl:ld-module-5} in the Appendix. 
The rough idea is that a vertex of $X(R_i)$ can always be replaced by a vertex of $T_i \setminus X$ if the latter  is non-empty. If $V(T_i) \setminus X$ is  empty, then after an unavoidable at least $c(a_1)$ vertices deletion from $S_i$, no deletion in $R_i$ is needed since each remaining vertex of $S_i$ in $H - X$ has maximum degree at most $|\beta| - s + a_1 < |\beta|$. The idea for  $\chi(R_{ij}) = 0$ goes the same way.

\begin{claim}\label{cl:schose}
For any $i \in [k]$, we may assume  ${\Delta(S_i - X) \in \{a_1, \ldots, a_n\}}$, and that $\chi(S_i) = c(\Delta(S_i - X))$.
\end{claim}

The detailed proof refers to case 2 of Claim \ref{cl:ld-module-6}  in the Appendix. 
The rough idea is that Case 2 of Definition \ref{def:kd-degdeletion-bdd-121} states that we can delete $c(a_j)$ vertices from $S_i$ to make $\Delta(S_i - X) = a_j$, where $0\leq j \leq n+1$. In fact, this is the smallest maximum degree we can achieve in $S_i$ by deleting between $c(a_j)$ and $c(a_{j+1}) - 1$ vertices, because of the stepwise behavior of arbitrary deletion tables. Moreover, we already know that $\chi(R_i) = 0$ based on Claim \ref{cl:nori}. So, if $\Delta(S_i - X) = a_0$ there is a vertex in $S_i - X$ with degree $|\beta| - s + a_0 > |\beta|$, and if $\Delta(S_i - X) = a_n - 1$ or less then $c(a_n - 1) = a_0 > q$ vertices have to be deleted from $S_i$.

\begin{claim}\label{cl:uchose}
For any distinct $i, j \in [k]$, we may assume that  ${\Delta(U_{ij} - X) \in I_{ij}}$, and that $\chi(U_{ij}) = c_{ij}(\Delta(U_{ij} - X))$.
\end{claim}
The idea for proving this claim is similar to that of Claim \ref{cl:schose}.

Our next step is to argue that the degree chosen by $U_{ij}$ must be the sum of the degrees chosen by $S_i$ and $S_j$.
More specifically, for $i \in [k]$, we say that $S_i$ \emph{chose} $a_j \in I$ if $\Delta(S_i - X) = a_j$ and $\chi(S_i) = c(a_j)$.  
Likewise, for distinct $i, j \in [k]$, we say that $U_{ij}$ \emph{chose} $a, b \in I$ if $\Delta(U_{ij} - X) = a + b$ and $\chi(U_{ij}) = c_{ij}(a + b)$.
Note that by Claim~\ref{cl:schose}, each $S_i$ chooses one $a_j$ and by Claim~\ref{cl:uchose}, each $U_{ij}$ chooses one pair $a, b$ such that symmetry edges  $f_{i}^{-1}(b) f_{j}^{-1}(a) , f_{i}^{-1}(a) f_{j}^{-1}(b) \in E(G)$.  
The point to make is that if $S_i$ and $S_j$ 
chose $a$ and $b$, respectively, then $U_{ij}$ must have chosen $a, b$.

\begin{claim}\label{cl:that-fking-claim}
For each $i \neq j$, 
if $S_i$ chose $a \in I$ and $S_j$ chose $b \in I$, then $U_{ij}$ chose $a, b$.
\end{claim}

The detailed proof of this claim refers to Claim \ref{cl:ld-module-10} in Appendix. We sketch it as follows. 
For each $i \in [k]$, we will denote by $\ha_i$ the element of $I$ that $S_i$ chose. 
We divide the $U_{ij}$'s into three groups:
\begin{align*}
    U^< &= \{ U_{ij} : \mbox{$U_{ij}$ chose $a', b'$ such that $a' + b' < \ha_i + \ha_j$} \} \\
    U^= &= \{ U_{ij} : \mbox{$U_{ij}$ chose $a', b'$ such that $a' + b' = \ha_i + \ha_j$} \} \\
    U^> &= \{ U_{ij} : \mbox{$U_{ij}$ chose $a', b'$ such that $a' + b' > \ha_i + \ha_j$} \}
\end{align*}

To prove the claim, it suffices to show that $U^<$ and $U^>$ are empty (this is because $U_{ij} \in U^=$ is only possible if $U_{ij}$ chose $\ha_i, \ha_j$, since all the sum pairs are distinct).
The rough idea is as follows.
If each $U_{ij}$ chose the correct $\ha_i, \ha_j$, then each of them will incur a deletion cost of $c_{ij}(\ha_i + \ha_j) = s - \frac{\ha_i + \ha_j}{2(k - 1)}$ and end up cancelling the deletion costs of the $S_i$ and $T_i$ factors.  
If $U^<$ is non-empty, it incurs extra deletion cost with respect to $c_{ij}(\ha_i + \ha_j)$ with no real benefit.  The complicated case is when $U^>$ is non-empty.  In this case, $U_{ij} - X$ has higher degree than if it had chosen $\ha_i, \ha_j$ and incurs less deletions than $c_{ij}(\ha_i + \ha_j)$.  However, this needs to be compensated with extra deletions in $T_i$ and $T_j$.  By using a charging argument, we can show that the sum of extra deletions required for all the $U^>$ members outweighs the deletions saved in the $U_{ij}$'s of $U^>$.

We can now construct a multicolored clique.  To this end,  define $C = \{f^{-1}_{i}(\ha_i) : i \in [k] \mbox{ and $S_i$ chose $\ha_i$} \}$.
We claim that $C$ is a clique.  
By Claim~\ref{cl:schose}, each $S_i$ chooses some $\ha_i$ and thus $|C| = k$.
Now let $f^{-1}_{i}(\ha_i), f^{-1}_{j}(\ha_j)$ be two vertices of $C$, where $i<j$.  
Then $\ha_i, \ha_j$ were chosen by $S_i$ and $S_j$, respectively, and by Claim~\ref{cl:uchose}
we know that $U_{ij}$ chose $\ha_i + \ha_j$.  
By the construction of the $U_{ij}$ solution table, this is only possible if symmetry edges  $f_{i}^{-1}(\ha_i) f_{j}^{-1}(\ha_j) , f_{i}^{-1}(\ha_j) f_{j}^{-1}(\ha_i) \in E(G)$.  Therefore, $f_{i}^{-1}(\ha_i) f_{j}^{-1}(\ha_j) \in E(G)$ and $C$ is a clique.

\section{FPT algorithms for succinct solution tables} \label{FPT-algo-for-suc-section}

In this section, we study the \bdd~problem (i.e. $\alpha = 1, \beta < 0$) on graph classes that admit succinct solution tables.
This notion is formalized below.

Let $G=(V,E)$ be a graph and $x$ be an integer in $[0,|V|]$, suppose that set $\mathbb{S}_x \subseteq 2^V$ consists of all subsets of $V$ with size $x$. A vertex set $X\in \mathbb{S}_x$ is called an \emph{$x$-deletion set} of $G$ if $\Delta(G-X) = \min \{\Delta(G-Y): Y\in \mathbb{S}_x\}$. An \emph{$x$-deletion} of $G$ is the process of deleting all the vertices of an $x$-deletion set from $G$. 
Recall that, for any integer $x\in [0, |V|]$, we call $f(x) = \min \{ \Delta (G - X) : |X| = x \}$ the degree deletion function of $G$. 
Thus, the degree deletion function of $G$ is the maximum degree of $G$ after an $x$-deletion.

A \textit{piecewise linear function}  $g: \mathbb{R} \rightarrow \mathbb{R}$ is a continuous function defined on a sequence of intervals, such that the function is linearly restricted to each of the intervals (each such linear function is called a \textit{sub-function} of $g$).  
In addition, a \emph{constant piecewise linear function} is a piecewise linear function that consists of a constant number of linear sub-functions. 

\begin{definition}\label{succict table of bdd}
Let $\G$ be a polynomial-time recognizable graph class.
For $G \in \G$, suppose that $f^G(x)$ is the degree deletion function of $G = (V,E)$, where $x\in [0, |V|]$. 
We say that $\G$ admits a succinct solution table for \bdd~if, for every $G\in \G$, there exists a function $g^G : \mathbb{R} \rightarrow \mathbb{R}$ that satisfies at least one of the following conditions:

\begin{enumerate}
    \item 
    $g^G$ is a constant piecewise linear function such that $f^G(x) = g^G(x)$ for every integer $x\in [0, |V(G)|]$,
    
    
    \item 
   $g^G$ is a piecewise convex linear function such that $f^G(x) = \left\lceil g^G(x) \right\rceil$ for every integer $x\in [0, |V(G)|]$.
\end{enumerate}
Moreover, $g^G$ can be described and constructed in polynomial time with respect to $|V(G)|$.
\end{definition}

For convenience, we say a graph class admits a succinct solution table of type $t$ for \bdd~if we refer to the condition $t$, where $t \in \{1,2\}$. The first type implies that the solution table can be divided into a small number of blocks, where the information of each block can be encoded into a small number of bits. 
In the second type, the solution table could be convex, but could also be non-convex and may 
contain a larger number of blocks, but we can reduce this special non-convex function to a convex function using ceilings (this occurs with \clusters-\gmparam).   

As we show, these definitions capture the intuition of solution tables that can be merged in FPT time.  The proof of Theorem \ref{succint table theorem for bdd} relies on the results of \cite{DBLP:journals/tcs/BredereckFNST20}, which show that \textsc{Mixed Integer Programming with Simple Piecewise Linear Transformations} (MIP/SPLiT) problem is FPT parameterized by the number of variables.

Let us first introduce the \milp~problem (MILP). In the MILP problem, the input is a matrix $A= \mathbb{Z}^{m \times (n+k)}$ with $m$ rows and $n+k$ columns and vector $\textbf{b} \in \mathbb{Q}^{n+k}$, the task is to decide whether there exists a vector $\textbf{x} =(x_1,\ldots,x_{n+k})$ such that $A\textbf{x}\leq \textbf{b}$, where variables $x_1, \ldots, x_n$ are integers and variables $x_{n+1},\ldots, x_{n+k}$ are rational numbers.
If $k=0$ then the problem is called \ilp~problem.
In addition, if we replace $x_i$ with some integer function on $x_i$, and $x_j$ with some real function on $x_j$, where $1\leq i\leq n$ and $n+1\leq j \leq n+k$, then we call the problem \mipwffull~problem (\ilpwf).

Recall that $\G$-\gmparam~is a generalization of the parameter neighborhood diversity which was first proposed by Lampis \cite{DBLP:journals/algorithmica/Lampis12}. Moreover, Lampis \cite{DBLP:journals/algorithmica/Lampis12} creates a nice technique, which uses MILP to combine the fact that each set of the neighborhood partition is able to be colored by either one color or the vertex number of the set colors, to design an FPT algorithm for \textsc{Chromatic Number} problem parameterized by neighborhood diversity. Then, Gajarsk{\'{y}} et. al. \cite{DBLP:conf/iwpec/GajarskyLO13} refine this technique and combine it with dynamic programming to obtain FPT algorithms for some important graph problems parameterized by modular-width. Moreover, the conclusion section of paper \cite{DBLP:conf/iwpec/GajarskyLO13} asks whether this technique can be further generalized to eventually obtain meta-theorem-like results, and suggests using `convex' solution tables for modules to generalize the technique. In the following, we demonstrate that \bdd~problem is FPT parameterized by $\G$-\gmparam~if $\G$ admits a `succinct' solution table for \bdd. The rough idea of our algorithm is to compute a succinct solution table for each module of a $\G$-modular partition of the graph and combine these tables using \ilpwf. Intuitively, a `succinct' solution table should be a `convex' one. However, our result goes further over this. We emphasize that condition 2 of our succinct solution table in definition \ref{succict table of bdd} allows some \emph{non-convex} functions.

The FPT algorithms in this section rely on the known FPT algorithm for \textsc{Mixed Integer Programming with Convex Constraints} \cite{DBLP:journals/tcs/BredereckFNST20}. So we define piecewise linear convex functions in the same way. A \textit{piecewise linear function}  $g : \mathbb{R} \rightarrow \mathbb{R}$ is a continuous function, which is defined on a sequence of intervals, such that the function is linearly restricted to each of the intervals. In addition, the linear function on each of the intervals is called a \textit{sub-function} of $g$. Function $g$ is called \textit{convex} if, for any real numbers $x_1,x_2$, there exists a real number $r$ with $0\leq r \leq 1$ such that $g(rx_1 + (1-r)x_2) \leq rg(x_1) + (1-r)g(x_2)$. A function $g$ is called concave if the function $-g$ is convex. 
Naturally, $g$ a \textit{piecewise linear convex (concave) functions} if $g$ is not only a piecewise linear function but also a convex (concave) function.  
Moreover, we restrict that the endpoints of all sub-functions of $g$ are rational points and that the slopes of all sub-functions of $g$ are rational numbers. More details for piecewise linear convex functions refer to \cite{DBLP:journals/tcs/BredereckFNST20}.

Let us recall the definition of the \textsc{Mixed Integer Programming with Simple Piecewise Linear Transformations} (MIP/SPLiT) problem \cite{DBLP:journals/tcs/BredereckFNST20}. The input is a set of integers $\{c_1,\ldots,c_m\}$, a set of piecewise linear concave functions $\{f_{i,j}: i\in [n+k], j\in [m]\}$, and a set of piecewise linear convex functions $\{g_{i,j}: i\in [n+k], j\in [m]\}$. The objective is to decide whether there exist $x_1,\ldots,x_{n+k}$ such that
\begin{align*}
 \sum_{i \in [n+k]} g_{i,j} (x_i) &\leq \sum_{i \in [n+k]} f_{i,j} (x_i) + c_j & \forall j \in [m]\\
  x_i &\in \mathbb{N}  & \forall i \in [n]\\
  x_i &\in \mathbb{R^+}  & \textrm{for} \: n+1 \leq i \leq n+k
\end{align*}
In addition, if functions $\{f_{i,j}: i\in [n+k], j\in [m]\}$ or functions $\{g_{i,j}: i\in [n+k], j\in [m]\}$ are non-convex functions, then we call the problem MIP with non-convex constraints.

\begin{theorem}[\cite{DBLP:journals/tcs/BredereckFNST20}]
\label{IPSPLIT is FPT}
The MIP/SPLiT problem can be decided in $n^{2.5n+o(n)}(I+P)^{O(1)}$ time, where $I$ is the size of the input, $P$ is the pieces number of the function that has the most number of pieces, and $n$ is the number of integer variables.
\end{theorem}

We will use this result as a black box in our proof. More details for MIP/SPLiT problem can be found in \cite{DBLP:journals/tcs/BredereckFNST20}. If all the variables of the instances are integers, then it is a special case of MIP/SPLiT problem where $k=0$.

\begin{theorem}\label{succint table theorem for bdd}
Let $\G$ be a graph class that admits a succinct solution table for \bdd.  Assume a minimum $\G$-modular partition is given. Then, \bdd~is FPT parameterized by $\G$-\gmparam.  
\end{theorem}
\begin{proof}
Let $(G,q,-\beta)$ be an input of the \bdd~problem, where $G=(V,E)$. Let $\P=\{M_1,\ldots,M_k\}$ be a $\G$-modular partition of $V$, where $k$ equals the $\G$-\gmparam~of $G$. Let $i\in [k]$. For each module $M_i$, suppose  $N(M_i)\subseteq \P$ includes all the modules that are adjacent to $M_i$. Let function $f^{G[M_i]}(x)$ be the maximum degree of $G[M_i]$ after an $x$-deletion in $G[M_i]$, for integer $x\in [0, |M_i|]$.

\begin{claim}\label{ilp with funciton bdd}
$(G,q,-\beta)$ is a yes instance of \bdd~problem if and only if the following instance is yes for \ilpwf.
\begin{align*}
&(1) \: \sum_{1\leq i \leq k} x_{i} \leq q &\\
&(2) \: \sum_{M_j \in N(M_i)} (|M_j|-x_j) +  f^{G[M_i]}(x_i) \leq -\beta & \forall i \in [k]\\
&(3) \: x_i \in \mathbb{N} & \forall i \in [k]\\
&(4) \: 0 \leq x_i \leq |M_i| & \forall i \in [k]
\end{align*}
\end{claim}

\begin{claimproof}
For one direction, suppose that there is $X \subseteq V$ with $|X|\leq q$ such that $G-X$ has maximum degree $-\beta$. Assume $X_i = M_i \cap X$. We claim that $x_i = |X_i|$ for all $i\in [k]$ is a feasible solution for the instance. According to the definition of $|X_i|$, the solutions fulfill constraints of the forms (3) and (4), and constraint (1). Consider each module $M_i$. Suppose $v_i$ is a vertex with the maximum degree $\Delta(G[M_i\setminus X_i])$ in $G[M_i\setminus X_i]$. Then, among all vertices in $M_i\setminus X_i$, $v_i$ has the most adjacent vertices in $G-X$, which is 
\begin{align*}
\sum_{M_j \in N(M_i)} (|M_j|-|X_j|) + \Delta(G[M_i\setminus X_i]) \leq -\beta
\end{align*}
since the degree of any vertex of $G-X$ is at most $-\beta$. Additionally, we have $f^{G[M_i]}(|X_i|) \leq \Delta(G[M_i\setminus X_i])$ according to the definition of $f^{G[M_i]}$. Therefore, the solutions fulfill the forms of constraint (2).

For the other direction, assume that $x_i=x'_i$ for all $i \in [k]$ is a feasible solution for the instance of \ilpwf. According to the forms of constraints (3) and (4), $x'_i$ is a natural number of size at most $|M_i|$. So we may assume that the solution $X$ is a subset of $V$ such that, for each $i$, $|X\cap M_i|=x'_i$ and $X\cap M_i$ is an $x'_i$-deletion set of $G[M_i]$. Based on the definition of $f^{G[M_i]}$, we have $\Delta(G[M_i\setminus X]) = f^{G[M_i]}(x'_i)$ for all $i \in [k]$. This means that, for every $i \in [k]$, any vertex in $M_i\setminus X$ has at most $f^{G[M_i]}(x'_i)$ neighbors in $G[M_i\setminus X]$, and thus has at most
\begin{align*}
\sum_{M_j \in N(M_i)} (|M_j|-x'_j) + f^{G[M_i]}(x'_i)
\end{align*}
neighbors in graph $G-X$. Therefore, we have $\Delta(G-X)\leq -\beta$ based on the forms of constraints (2). In addition, the cardinality of $X$ is at most $q$ according to constraint (1). Hence, $(G,q,-\beta)$ is a yes instance for \bdd~problem.
\end{claimproof}


 Since $\G$ admits a succinct solution table for \bdd,
 for every $i \in [k]$, there exists a function $g^{G[M_i]} : \mathbb{R} \rightarrow \mathbb{R}$ that fulfills one of the following conditions: 

\begin{enumerate}
    \item 
    $g^{G[M_i]}$ is a constant piecewise linear function  such that $f^{G[M_i]}(x) = g^{G[M_i]}(x)$ for every integer $x\in [0, |M_i|]$,
    
    
    \item 
    $g^{G[M_i]}$ is a piecewise convex linear function  such that $f^{G[M_i]}(x) = \left\lceil g^{G[M_i]}(x) \right\rceil$  for every integer $x\in [0, |M_i|]$,
\end{enumerate}
moreover, $g^{G[M_i]}$ can be constructed in polynomial time according to Definition.
Let $t\in \{1,2\}$. Suppose $\P_1, \P_2$ is a partition of $\P$, where $M_i\in \P_t$ has the property that $f^{G[M_i]}$ fulfills the condition $t$, for $t \in \{1,2\}$. 

\begin{claim}\label{ilp with nonconvex to convex}
The instance for \ilpwf~in Claim \ref{ilp with funciton bdd} is a yes instance if and only if the following instance is yes for \ilpwf.
\begin{align*}
&(1) \: \sum_{1\leq i \leq k} x_{i} \leq q &\\
&(2) \: g^{G[M_i]}(x_i) \leq -\beta -\sum_{M_j \in N(M_i)} (|M_j|-x_j) & \forall i \in [k]\\
&(3) \: x_i \in \mathbb{N} & \forall i \in [k]\\
&(4) \: 0 \leq x_i \leq |M_i| & \forall i \in [k]
\end{align*}
\end{claim}

\begin{claimproof}
For one direction, assume that $x_i=x'_i$ for all $i \in [k]$ is a feasible solution for the instance of \ilpwf~in Claim~\ref{ilp with funciton bdd}. Obviously, $x'_1,\ldots,x'_k$ are feasible solutions for the forms of constraints (3) and (4), and constraint (1) in Claim \ref{ilp with nonconvex to convex}. If $M_i\in \P_1$, then $f^{G[M_i]}(x'_i) = g^{G[M_i]}(x'_i)$. Clearly, $x'_1,\ldots,x'_k$ are feasible solutions for the forms of constraints (2) with $M_i \in \P_1$ in Claim~\ref{ilp with nonconvex to convex}. If $M_i\in \P_2$, then $f^{G[M_i]}(x'_i) = \left\lceil g^{G[M_i]}(x'_i)\right\rceil$. By rearranging constraint (2) of Claim~\ref{ilp with funciton bdd}, we have that 
\begin{align*}
\left\lceil  g^{G[M_i]}(x'_i) \right\rceil = f^{G[M_i]}(x'_i) \leq -\beta -\sum_{M_j \in N(M_i)} (|M_j|-x'_j).
\end{align*}
Moreover, for any real number $r$, we have $r\leq \left\lceil r \right\rceil$. Hence, $x'_1,\ldots,x'_k$ are feasible solutions for the forms of constraints (2) with $M_i\in \P_3$ in Claim \ref{ilp with nonconvex to convex}.

For the other direction, assume that $x_i=x'_i$ for all $i \in [k]$ is a feasible solution for the instance of \ilpwf~in Claim \ref{ilp with nonconvex to convex}. Obviously, $x'_1,\ldots,x'_k$ are feasible solutions for the forms of constraints (3) and (4), and constraint (1) in Claim \ref{ilp with funciton bdd}. If $M_i\in \P_1$, then $f^{G[M_i]}(x'_i) = g^{G[M_i]}(x'_i)$. Clearly, $x'_1,\ldots,x'_k$ are feasible solutions for the forms of constraints (2) with $M_i \in \P_1$ in Claim \ref{ilp with funciton bdd}. For each $M_i \in \P_2$, by constraint (2), we get
\begin{align*}
 g^{G[M_i]}(x'_i) \leq -\beta -\sum_{M_j \in N(M_i)} (|M_j|-x'_j)
\end{align*}
for each $i \in [k]$, 
where the left part and the right part of the inequality
are a real number and an integer, respectively. Moreover, for any real number $r$ and any integer $z$, if $r\leq z$ then $\left\lceil r \right\rceil \leq z$. Hence, for every $M_i\in \P_2$, we have
\begin{align*}
f^{G[M_i]}(x'_i) = \left\lceil g^{G[M_i]}(x'_i)\right\rceil  \leq d -\sum_{M_j \in N(M_i)} (|M_j|-x'_j).
\end{align*}
As a result, $x'_1,\ldots,x'_k$ are feasible solutions for the forms of constraints (2) with $M_i\in \P_2$ in Claim \ref{ilp with funciton bdd}.
\end{claimproof}

If a graph class $\G$ admits a succinct solution table for \bdd, it is not hard to verify that the instance for \ilpwf~in Claim \ref{ilp with nonconvex to convex} can be constructed in polynomial time when given an instance $(G,q,-\beta)$.
Now, let us consider the instance of \ilpwf~in Claim \ref{ilp with nonconvex to convex}, where $g^{G[M_i]}$ is either a constant piecewise linear function or a piecewise convex linear function. 

 Consider $g^{G[M_i]}$ in the instance of \ilpwf~in Claim \ref{ilp with nonconvex to convex}, where $M_i\in \P_1$. In this case, $g^{G[M_i]}$ is a constant piecewise linear function. Then there is a constant $h_i$ such that $g^{G[M_i]}$ consists of $h_i$ linear sub-functions in the domain $[0,|M_i|]$.  
 Let us denote $C = \max_{M_i \in \P_1} h_i$.
 Additionally, we call the domain for a linear sub-function of $g^{G[M_i]}$ a linear sub-domain of $g^{G[M_i]}$. More specifically, for each $M_i \in  \P_1$, assume the linear sub-domains of $g^{G[M_i]}$ are intervals $[m_i^0,m_i^1], [m_i^1,m_i^2] \ldots [m_i^{h_i-1},m_i^{h_i}]$, where $m_i^0 = 0$ and $m_i^{h_i} = |M_i|$.

 For the instance of \ilpwf~in Claim \ref{ilp with nonconvex to convex}, if we restrict each integer variable $x_i$, where $M_i\in \P_1$, to a linear sub-domain of $g^{G[M_i]}$, then the instance is divided into $\Pi_{M_i\in \P_1} h_i \leq C^k$ new instances for MIP/SPLiT, moreover, the original instance is feasible if and only if there exists one of the new instances, as follows, is feasible.
\begin{align*}
&(1) \: \sum_{1\leq i \leq k} x_{i} \leq q &\\
&(2) \: g^{G[M_i]}(x_i) \leq -\beta -\sum_{M_j \in N(M_i)} (|M_j|-x_j) & \forall i \in [k]\\
&(3) \: x_i \in \mathbb{N} & \forall i \in [k]\\
&(4) \: 0 \leq x_i \leq |M_i| & \forall  M_i \in \P_2\\
&(5) \: m_i^{s-1} \leq x_i \leq m_i^{s} & \forall M_i \in \P_1, \:  \textrm{and some} \: s \in [h_i]
\end{align*}
 
 In addition, there is an $O^*(k^{2.5k+o(k)})$ time algorithm that decides the instance of MIP/SPLiT based on Theorem \ref{IPSPLIT is FPT}. Overall, \ilpwf~in Claim \ref{ilp with nonconvex to convex} can be solved in $O^*(k^{2.5k+o(k)}\cdot C^k)$ time, and thus \bdd~is FPT parameterized by $\G$-\gmparam.
\end{proof}

\noindent
\textbf{Application:} For each $t \in \{1,2\}$, we provide a graph class that admits a succinct solution table of type $t$ for \bdd.  Let graph class BDT include all graphs with bounded degree and treewidth. We first demonstrate that 
 BDT admits a succinct solution table of type 1, thus \bdd~problem is FPT parameterized by $BDT$-\gmparam.  Then, we prove that the \clusters~admits a succinct solution table of type 2, hence \bdd~problem is FPT parameterized by \clusters-\gmparam. 

 Clearly, \textit{linear forest}, formed from the disjoint union of path graphs, and \textit{binary forest}, formed from the disjoint union of binary trees, have unbounded modular-width, but have bounded treewidth and bounded degree. Thus, linear forest and binary forest are subsets of BDT, moreover, $BDT$-\gmparam~is incomparable with modular-width. In addition, it is not hard to verify that  $BDT$-\gmparam~is bounded by a function of vertex cover number since there are no twin vertices in the quotient graph. 
Furthermore, \clusters-\gmparam~is a parameter with size at most neighborhood diversity and at least modular-width.


\begin{theorem}
\label{grid has a succinct table}
BDT admits a succinct solution table for \bdd.
Therefore, the \bdd~problem is FPT parameterized by $BDT$-\gmparam.
\end{theorem}
\begin{proof}
First, there is a linear time algorithm to decide if a graph in BDT \cite{DBLP:journals/siamcomp/Bodlaender96}. In addition, we have demonstrated that $(0,\beta)$-\textsc{Linear Degree Domination} with constant $\beta$ (or \bdd) admits a constant-length MSO$_1$ formula, thus admits a constant-length MSO$_2$ formula. Therefore, \bdd~on BDT graph can be solved in polynomial time for any $\beta$ according to the Courcelle's theorem \cite{DBLP:journals/algorithmica/BoriePT92, DBLP:journals/iandc/Courcelle90}. 

Let $G=(V,E)$ be a BDT graph with maximum degree $d$. Assume the maximum degree of $G$ after an $x$-deletion is $f^G(x)$, which is clearly a non-increasing function. 
Assume $S= \{ (x,f^G(x)): x\in [0,|V|]\}$. Suppose $Y = \{ y : (x,y) \in S\} = \{d_1,\ldots,d_t\}$, where $d_1 > d_2 \ldots >d_t$. Then, we have $Y \subseteq \{0,1,\ldots,d\}$ and $t \leq d+1$.  Clearly, we can use one horizontal segment (point) to cover all points of $(x,y) \in S$ with $y = d_i$ for each $i\in [t]$, where the endpoints of the horizontal segment (point) are in $S$. Then, connect all horizontal segments (points) from left to right using another $t-1$ segments, all of which make up the piecewise linear function $g^G: \mathbb{R} \rightarrow \mathbb{R}$ in the domain $[0,|V|]$ that contains at most $2t-1$ segments. To make $g^G$ well-defined, assume $g^G(x) = d_1$ for $x\leq 0$, and $g^G(x) = d_t$ for $x\geq |V|$.  
\end{proof}

We now move to an example of a type 2 succinct solution table on cluster graphs.
Let $H \in \clusters$. For any complete graph $K$ in $H$ and any $X \subseteq V(H)$, suppose $X'$ is obtained from $X$ by replacing $K\cap X$ with any $|K\cap X|$ vertices of $K$. Clearly, graphs $H-X$ and $H - X'$ are isomorphic, and $\Delta(H-X)$ equals $\Delta(H - X')$. Thus, when we consider the \bdd~problem on any complete graph of a cluster graph, we only need to consider the deleted vertex number of it and do not need to distinguish the difference between any two vertices of it.

\begin{lemma}\label{reduce b in cluster}
Assume cluster graph $H$ contains $b$ complete graphs $K_a$, where $K_a$ is the maximum size complete graph in $H$. Let $q\in [b,|V(H)|]$. Suppose $R$ is obtained from $H$ by deleting exactly one vertex from every $K_a$ of $H$. Then, $H$ has a $q$-deletion such that the maximum degree of the remaining graph is $h$ if and only if $R$ has a $(q- b)$-deletion such that the maximum degree of the remaining graph is $h$. 
\end{lemma}
\begin{proof}
The correctness of the lemma is evident if $a=1$. Next, assume $a\geq 2$. We claim that any $q$-deletion set of $H$ contains at least one vertex from every $K_a$. Otherwise, the maximum degree of $H$ after the $q$-deletion is $a-1$, however, we can delete one vertex from every $K_a$ and any other $q-b$ vertices from $H$ so that the maximum degree of the remaining graph is at most $a-2$, a contradiction. Next, we prove the lemma for the case that $a\geq 2$.

For one direction, assume $H$ has a $q$-deletion and the maximum degree of the remaining graph of $H$ is $h$. Then, the deleted vertices from $H$ contain at least one vertex from every $K_a$. Without loss of generality, we may assume the vertices of $V(H)-V(R)$ are included in the $q$-deletion of $H$. We have that $R$ has a $(q- b)$-deletion such that the maximum degree of the remaining graph is at most $h$. If $R$ has a $(q- b)$-deletion set $D_R$ such that the maximum degree of the remaining graph of $R$ is at most $h-1$, then $V(H)-V(R)$ together with $D_R$ is a $q$-deletion for $H$ such that the maximum degree of the remaining graph of $H$ is at most $h-1$, a contradiction.

For the other direction, assume $R$ has a $(q-b)$-deletion and the maximum degree of the remaining graph of $R$ is $h$. Clearly, $H$ has a $q$-deletion such that the maximum degree of the remaining graph is at most $h$. Moreover, if $H$ has a $q$-deletion such that the maximum degree of the remaining graph of $H$ is at most $h-1$, then, according to the result proved in the last paragraph, $R$ has a $(q-b)$-deletion and the maximum degree of the remaining graph of $R$ is at most $h-1$, a contradiction.
\end{proof}

\begin{theorem}
\label{k12-free has a succinct table}
The graph class \clusters~admits a succinct solution table for \bdd. Hence, the \bdd~problem is FPT parameterized by $cluster$-\gmparam.
\end{theorem}
\begin{proof}
Clearly, it is polynomial-time decidable whether a given graph is a cluster graph. Let $G=(V,E)$ be a cluster graph. Assume the maximum degree of $G$ after an $x$-deletion is $f^G(x)$. Suppose $G$ contains $r$ different types of complete graphs $K_{a_1+1},\ldots, K_{a_r+1}$ with vertex number at least two, where $a_1 > \ldots > a_r > 0$. Define $a_{r+1}=0$.  For each $i\in [r+1]$, assume $G$ contains $b_i$ complete graphs with exactly ${a_i+1}$ vertices. Clearly, we have $|V|= \sum_{1\leq i \leq r +1}(a_i+1)b_i$.

Let $g^G: \mathbb{R} \rightarrow \mathbb{R}$ be a piecewise linear function. For each $i\in [r]$, two endpoints of the $i$-th segment of function $g^G$ are $(x_i,a_i)$ and $(x_{i+1},a_{i+1})$, where $x_i= \sum_{1 \leq j\leq i} (a_{j}-a_{i})b_j$ and $x_{r+1}= \sum_{1 \leq j\leq r+1} a_jb_j$. It is not hard to calculate the formulas of the function for all $r$ segments, which is 
\begin{align*}
g^G(x) &= -\frac{1}{\sum_{1 \leq j\leq i} b_j}(x-x_i) + a_i & \textrm{for} \: x \in [x_i,x_{i+1}).
\end{align*}

To make $g^G$ well-defined, assume that $g^G(x)= -|V|\cdot x+a_1$ is the 0-th segment, where $x\in (-\infty, x_1)$; and $g^G(x)= 0$ is the $(r+1)$-th segment, where $x\in [x_{r+1}, +\infty)$. Consider the slope for each segment of the $g^G(x)$. For each $i\in [r-1]$, the slope of $i$-th segment $-(b_1+\ldots+b_i)^{-1}$ is less than the slope of $(i+1)$-th segment $-(b_1+\ldots+b_i+b_{i+1})^{-1}$. In addition, the slope of $0$-th segment $-|V|$ is less than the slope of $1$-th segment $-b_1^{-1}$; and the slope of $r$-th segment $-(b_1+\ldots+b_r)^{-1}$ is less than the slope of $(r+1)$-th segment 0. Consequently, $g^G$ is a convex function since the slope for the $i$-th segment of $g^G$ increases as $i$ increases for all $0\leq i\leq r+1$. In addition, it is clear that $g^G$ can be constructed in polynomial time when giving $G$ as an input.

Now, we prove that $f^G(x)=\left\lceil g^G(x) \right\rceil$ for every integer $x\in [0,|V|]$ to demonstrate that cluster graph admits a succinct solution table of type 2 for \bdd. 
Let vertex set $S$ be a subset of $V$ such that every complete graph of $G$ has exactly one vertex that is not included in $S$. Clearly, $\Delta(G - S)=0$ and $S$ contains exactly $x_{r+1} = \sum_{1 \leq j\leq r+1} a_jb_j$ vertices. Consider every integer $x\in [x_{r+1}, |V|]$. $S$ together with any $x-x_{r+1}$ vertices of $V\setminus S$ are an $x$-deletion set of $G$ and the maximum degree of the remaining graph is zero because $\Delta(G - S)$ is already zero and the maximum degree of $G$ after deleting any $x$ vertices is at least zero. Therefore, $f^G(x)=0$ for every integer $x\in [x_{r+1}, |V|]$.
In addition, based on the definition of $g^G$, $g^G(x)=0$ for all integer $x\geq x_{r+1}$. Thus, we have $f^G(x)=\left\lceil g^G(x) \right\rceil$  for every integer $x \in [x_{r+1}, |V|]$. 

Next, let us consider $x\in [0,x_{r+1})$. Since $x_1 = 0$, we only need to consider the 1st to $r$-th segments of function $g^G$. Consider $x\in [x_i,x_{i+1})$ for each $i\in [r]$, which is the domain of $i$-th segment of $g^G$. Based on Lemma \ref{reduce b in cluster}, an $x$-deletion set of $G$, denoted by $X$, can be obtained by carrying out the following two processes.

\begin{enumerate}
    \item 
    If $x$ is not smaller than the number of maximum complete graphs of $G$, we repeatedly do the process: select exactly one vertex from every maximum complete graph of $G$, add them to $X$, delete them from $G$, and $x = x - x'$, where $x'$ is the number of the deleted vertices in this process.
    
    \item 
    If $x$ is smaller than  the number of maximum complete graphs of $G$, then select arbitrary $x$ vertices from $G$, add them to $X$, delete them from $G$, and $x=0$.
\end{enumerate}

Since $x_{i+1}- x_i = (a_i - a_{i+1})(b_1+ \ldots +b_i)$, we may assume $x=x_i+ts+y$, for $s=b_1+ \ldots +b_i$, integer $t \in [0, a_i - a_{i+1})$, and integer $y \in [0, b_1+ \ldots +b_i)$. Since $x\geq x_i= \sum_{1 \leq j\leq i} (a_{j}-a_{i})b_j$, the first $x_i$ vertices added to $X$ in process (1) come from the complete graphs with vertex numbers larger than $a_i+1$, and the vertex numbers of all these complete graphs decrease to $a_i+1$ after the deletion of $x_i$ vertices  in process (1). 
After that, $G$ contains exactly $s$ complete graphs with $a_i+1$ vertices, and there are $ts+y$ vertices to be deleted. Then, process (1) will be carried out $t$ more times and the maximum degree of $G$ decreases to $a_i-t > a_{i+1}$. 
After that, there are $y$ vertices to be deleted. Since $y< b_1+ \ldots +b_i$ and the number of maximum complete graphs with $a_i-t+1$ vertices in $G$ is $s > y$, the $y$ arbitrary vertex deletion in process (2) does not change the maximum degree of $G$. As a result, the maximum degree of $G$ is $a_i-t$ after an $x$-deletion of $G$, this means that $f^G(x)=a_i-t$. On the other hand, we know 
\begin{align*}
\left\lceil g^G(x) \right\rceil &= \left\lceil -\frac{1}{\sum_{1 \leq j\leq i} b_j}(x-x_i) + a_i \right\rceil & \textrm{for} \: x \in [x_i,x_{i+1}).
\end{align*}

Replace $x$ with $x_i+st+y$,  we have
\begin{align*}
\left\lceil g^G(x) \right\rceil &= \left\lceil-\frac{1}{\sum_{1 \leq j\leq i} b_j}(x_i+st+y -x_i) + a_i \right\rceil\\
 &= \left\lceil-\frac{1}{s}(st+y) + a_i \right\rceil\\ 
 &= \left\lceil a_i -t -\frac{y}{s}  \right\rceil\\
 &= a_i -t.  
\end{align*}
Hence, $f^G(x) = \left\lceil g^G(x) \right\rceil$ for every integer $x\in [0,x_{r+1})$.
\end{proof}



In fact, for a modular partition $\M=\{M : M\subseteq V\}$ of $V$, if the factors $G[M]$ belong to different graph classes that admit succinct solution tables for \bdd, i.e. some factors are BDT and some factors are clusters, then the problem is still FPT parameterized by $|\M|$ according to Theorem~\ref{succint table theorem for bdd}. Therefore, based on Lemma \ref{grid has a succinct table},  Lemma \ref{k12-free has a succinct table}, and Theorem~\ref{succint table theorem for bdd}, we have the following corollary.

\begin{corollary}
The \bdd~problem is FPT parameterized by ($BDT  \cup \clusters)$-\gmparam.
\end{corollary}

Clearly, the parameter ($BDT \cup \clusters)$-\gmparam~is not greater than neighborhood diversity but is incomparable with modular-width.

\medskip

\section{Conclusions}

We conclude with some interesting problems.
\begin{itemize}
\item 
Can we characterize graph classes are easily mergeable?   
For instance, is the class of $H$-free graphs easily mergeable, for any fixed graph $H$?

\item 
 Is \bdd~fixed-parameter tractable in parameter ($K_{1, t}$-free)-\gmparam, where $t\geq 3$ is either fixed or a parameter?

\item 
Is \kd~FPT in parameter \clusters-\gmparam, where $\beta$ is related to the input size?
\item 
Is \ad~FPT in parameter \clusters-\gmparam?

\item 
The \textsc{Red-Blue Capacitated Dominating Set}~problem was shown to be W[1]-hard in $cw$ in~\cite{DBLP:journals/siamcomp/FominGLS14}.  It is not hard to prove it to be FPT in $mw$ using succinct solution tables.  Does the same hold for the \textsc{Red-Blue Exact Saturated Dominating Set}?

\item 
Are \textsc{Edge Dominating Set}, \textsc{Max-cut}, and \textsc{Partition into Triangles} FPT in parameter \cographs-\gmparam?
\end{itemize}

\section*{Acknowledgement}

We thank the anonymous reviewers for their valuable comments.

\newpage

\bibliographystyle{abbrv}
\bibliography{main}

\newpage

\section*{Appendix: proof for Theorem \ref{ld-main-theorem} }

For convenience, we recall the problem definition and some of the definitions that appear in the main text.  

\medskip 

\noindent 
The~\ld~problem\\
\textbf{Input:} a graph $G = (V, E)$ and a non-negative integer $q$;\\
\textbf{Question:} does there exist a subset $X\subseteq V$ of size at most $q$ such that for every $v \in V \setminus X$,  $|N(v)\cap X| \geq \alpha |N(v)| + \beta$?

\medskip

If $X$ is a $(\alpha, \beta)$-linear degree dominating set of $G = (V,E)$, then we say $V\setminus X$ and every vertex of $V\setminus X$ satisfies the linear inequality of the problem. In addition, for convenience, we also call $X$ the deletion part of $G$.

For a vertex $v\in V \setminus X$ and $N(v) \neq \emptyset$ with respect to graph $G$, we define that the dominating coefficient of $v$ with respect to $(G,X)$ is $\lambda(v,G,X) = \frac{|N(v)\cap X| - \beta}{|N(v)|}$.
Moreover, if $V\setminus X$ is not empty, then, for a vertex set $W \subseteq V\setminus X$, we define that the dominating coefficient of $W$ with respect to $(G,X)$ is $\Lambda(W, G ,X) = \min\{\lambda(v,G,X) : v\in W \}$, otherwise, we define that the dominating coefficient of $V\setminus X$ with respect to $(G,X)$  is one. Clearly, $X$ is a $(\alpha,\beta)$-linear dominating set of $G$ if and only if the dominating coefficient of $V\setminus X$ with respect to $(G,X)$  is at least $\alpha$, where $G$ does not have an isolated vertex. 

For any integer $x\in [0, |V|]$, we call $f(x) = \min \{ \Delta (G - X) : |X| = x \}$ the \textit{degree deletion function} of $G$, where $X$ is a subset of $V$. Then, we have the following observation, the proof of which is trivial. 

\begin{observation} \label{obervation-for-del-table}
The degree deletion function of a graph $G$ is a non-increasing function.
\end{observation}

Observation \ref{obervation-for-del-table} and Definition \ref{def:kd-degdeletion-bdd-121}  imply that, for any $a_i \in I$, there exists $X \subseteq V$ of size at least $c(a_i)$ such that $\Delta(G - X) \leq a_i$.

In a similar fashion, for any integer $x\in [0, |V|-1]$, we call $f(x) = \max \{ \delta (G - X,X) : |X|= x\}$ the \textit{degree retention function} of $G$, where $X$ is a subset of $V$. Then, we have the following observation, the proof of which is trivial.

\begin{observation} \label{obervation-for-ren-table}
The degree retention function of a graph $G$ is a non-decreasing function.
\end{observation}

Observation \ref{obervation-for-ren-table} and Definition \ref{def:kd-degdeletion} imply that, for any $a_i \in \{a_0\} \cup I$, there exists $X \subset V$ of size at least $p + c(a_i)$ such that $\delta (G - X,X) \geq a_i$.

In the following, we provide the reduction for the hardness for arbitrary deletion and retention tables.
Our reduction is from a specific variant of the \textsc{Multicolored Clique} problem, which we call~\mcc.
A symmetric multicolored graph $G = (V^1\cup V^2 \cup \ldots \cup V^k , E)$ is a connected 
graph such that, for all distinct $i,j \in [k]$,
\begin{enumerate}
    \item
    $V^i = \{v_1^i, \ldots, v_n^i\}$, where $n\geq k$;
    
    \item 
    all the vertices of $V^i$ are colored by color $i$ and form an independent set;

    \item
    if $v_r^i v_s^j \in E(G)$, then $v_s^j v_r^i \in E(G)$ as well.
\end{enumerate}
Then, for the \mcc~problem, the input is a symmetric multicolored graph $G$ and an integer $k$, and the objective is to decide whether $G$ contains a $k$-clique with vertices of all $k$ colors. We also call $v_r^i v_s^j$ and $v_s^i v_r^j$ \emph{symmetry edges}.

We claim that \mcc~problem is W[1]-hard, which is immediate by using the same proof of Lemma 1 in \cite{DBLP:journals/tcs/FellowsHRV09}. For convenience, we sketch the proof here, which is via a reduction from $k$-\textsc{Clique} problem. Let $(G,k)$ be an instance of $k$-\textsc{Clique}. Without loss of generality, assume $G$ is connected. Create $G'$ by replacing every $v\in V(G)$ with vertices $v_1,\ldots,v_k$, one in each color class. If $u$ and $v$ are adjacent in $G$, then add the edge $u_i v_j$, for every distinct $i, j \in [k]$. Clearly, $(G',k)$ is a valid instance of \mcc~problem, and $(G,k)$ is a yes instance of $k$-\textsc{Clique} problem if and only if $(G',k)$ is a yes instance of \mcc~problem. 

Suppose $S_k$ denotes a star with $k$ leaves. Here, we provide two lemmas about stars, which will be used later.

\begin{lemma}\label{the detail of degree deletion graph}
Let $G$ be a $(a_0, I, c)$-degree deletion star graph. Apart from the stars with leaves less than $a_{|I|-1}$, $G$ consists of exactly $c(a_{i}) - c(a_{i-1})$ stars with $a_{i-1}$ leaves for all  $i\in [|I|]$.
\end{lemma}
\begin{proof}
Let us consider every $i \in  [|I|]$. According to condition 2 of Definition \ref{def:kd-degdeletion-bdd-121}, there are at most $c(a_i)$ stars, each of which contains at least $a_i+1$ leaves. Based on condition 3 of Definition \ref{def:kd-degdeletion-bdd-121} there are at least $c(a_i)$ stars, each of which contains at least $a_{i-1}$ leaves. In addition, we have $a_{i-1} \geq a_i +1$, therefore there are exactly $c(a_i)$ stars, each of which contains at least $a_{i-1}$ leaves. Suppose $G_i$ consists of the $c(a_i)$ stars with at least $a_{i-1}$ leaves. Thus, together with condition 2 of Definition \ref{def:kd-degdeletion-bdd-121}, exactly one vertex for each star in $G_i$ has to be added to $X$ to make $G-X$ have maximum degree $a_i$ since $a_i\leq a_{i-1} - 1$. Since $G_i$ contains exactly $c(a_i)$ stars, apart from $G_i$, there is no vertices in $G$ are added to $X$. This means that, apart from $G_i$, the internal vertex of every star of $G$ contains at most $a_i$ leaves. Moreover, for every $i\in [|I|-1]$, we also know that there are exactly $c(a_{i+1})$ stars, each of which contains at least $a_{i}$ leaves. Therefore, for every $i\in [|I|-1]$, there are exactly $c(a_{i+1}) - c(a_i)$ stars, each of which contains exactly $a_{i}$ leaves. In addition, there are exactly $c(a_1)$ stars, each of which contains at least $a_{0}$ leaves, together with the fact that $G$ has maximum degree $a_{0}$ based on condition 1 of Definition \ref{def:kd-degdeletion-bdd-121}. Therefore, there are exactly $c(a_{1}) = c(a_{1}) - c(a_0)$ stars, each of which contains exactly $a_{0}$ leaves.
\end{proof}

\begin{lemma}\label{greedy lemma}
Suppose $M$ is a stars-module of a graph $G$ such that the degrees of internal vertices $v_1,\ldots,v_k$ of all the stars satisfy that $deg(v_1)\geq \ldots \geq deg(v_k)\geq 2$. Assume $X$ is a subset of $V(G)$ and $|X\cap M|=r$. Then, for any $r\in [k-1]$, to obtain the maximum value of the dominating coefficient of $M\setminus X$ with respect to $(G,X)$, we may choose $X\cap M=\{v_1,\ldots,v_r\}$.
\end{lemma}
\begin{proof}
Suppose the vertex number of $N(M)$ is $y$ and the vertex number of $N(M)\cap X$ is $x$, clearly we have $y\geq x\geq 0$. Assume that $\gamma = \frac{x}{y+deg(v_{r+1})}$. We have $\gamma \leq \frac{x}{y+2}$ since $deg(v_{r+1}) \geq deg(v_{k}) \geq 2$.

First, we claim that, for any $X\cap M$ with $r$ vertices, the dominating coefficient of $M\setminus X$ with respect to $(G,X)$ is at most $\gamma$. We give proof by contradiction for the claim as follows. Assume that there exists an $X\cap M$ with $r$ vertices such that the dominating coefficient of $M\setminus X$ with respect to $(G,X)$ is larger than $\gamma$. Consider all the stars with internal vertices $v_1,\ldots,v_{r+1}$. Clearly, for every $i\in [r+1]$, we have $\lambda(v_i,G,X) \leq \gamma$. Moreover, each star is a connected component in $G[M]$, so we have to add at least one vertex of each star to $X$ to change the value of $\lambda(v_i,G,X)$ for each $i$. Thus, $X\cap M$ contains at least $r+1$ vertices, a contradiction.

Second, we demonstrate that the dominating coefficient of $M\setminus X$  with respect to $(G,X)$ is at least $\gamma$ if we choose $X\cap M=\{v_1,\ldots,v_r\}$. For $i\in [r]$, the dominating coefficient of each adjacent vertex of $v_i$ with respect to $(G,X)$  is  $\frac{x+1}{y+1}> \gamma$. For $r +1 \leq  i \leq k$, the dominating coefficient of each adjacent vertex of $v_i$ with respect to $(G,X)$  is  $\frac{x}{y+1} \geq \gamma$, and the dominating coefficient of $v_i$ with respect to $(G,X)$  is $\frac{x}{y+deg(v_i)} \geq \gamma$.
\end{proof}

The rest of the section is dedicated to proving the following. Suppose $\I$ is the class of all edgeless graphs.  
We provide a reduction from the \mcc~problem to each above-mentioned case of the \textsc{$(\alpha,\beta)$-Linear Degree Domination} problem.  
Let $(G,k)$ be an instance of \mcc, where  $G = (V^1\cup V^2 \cup \ldots \cup V^k , E)$. Without loss of generality, suppose $k\geq 100$, otherwise, the problem can be solved in polynomial time.
We next describe a graph $H$ and integer $q$, which are the input to our corresponding instance of the \textsc{$(\alpha,\beta)$-Linear Degree Domination}.  

Note that Theorem~\ref{ld-main-theorem} has three cases and would require three reductions.  In cases 1 and 3, $\alpha$ is fixed (to either $0$ or $1$), but we must also specify $\beta$ as an input parameter, whereas in case 2, $\alpha$ and $\beta$ are constants that we have no control over.
In all three cases of the theorem, the general modular structure of $H$ is the same, and only the structure inside the modules differ (in addition to $\beta$ also being different).  We therefore present the proof of the theorem in one single reduction, and separate the proof into cases only when needed.


Before proceeding, we will make use of a $2$-sumfree-set, which is a set of positive integers in which every pair of elements has a distinct sum. That is, $I$ is a $2$-sumfree set if, for any $(a, b), (a', b') \in I \times I$, $a + b = a' + b'$ if and only if $\{a, b\} = \{a', b'\}$ (note that $a = b$ is possible).
Clearly, we can construct in time $O(n^4 \log n)$ a $2$-sumfree-set $I = \{x_1, \ldots, x_n\}$, in which the maximum value is at most $n^4$, using a greedy procedure. The reason is  that, for any positive integers $i, j, r \leq h < n$, $x_i+x_j-x_r$ has at most $h^3$ different values and we can always find the next element $x_{h+1}$ in the interval $(h^4, (h+1)^4]$ since $(h+1)^4-h^4 > h^3$. Let us also mention that the more general notion of $h$-sumfree sets was recently used to show hardness results for \textsc{Bin Packing} \cite{DBLP:journals/jcss/JansenKMS13}.

We assume that we have built a $2$-sumfree set $I = \{a_1, \ldots, a_n\}$,
where $a_1 > \ldots > a_n \geq 1$.
Without loss of generality, we may assume that for any integer $r \in \mathbb{N}$,  each element $a_i \in I$ is a multiple of $r$.  Indeed, if $I$ does not satisfy this property, we can update $I$ by multiplying each $a_i$ by $r$, which preserves the $2$-sumfree property. 
The value of $r$ we want to use depends on whether we are in case 1, 2, or 3 of the theorem, and the possibilities are listed in Table~\ref{table:reduction-values}.  We henceforth assume that every element of $I$ is a multiple of $r$, defined according to the table.
Since the initial $I$ had maximum value $n^4$, we thus have 
$r n^4 \geq a_1 > \ldots > a_n \geq r$.
Moreover, for any distinct pair $(a,b),(a',b')\in I \times I$, we have that the absolute value of $a+b-a'-b'$ is at least $r$.

For each 
color class $V^i = \{v_1^i, \ldots, v_n^i\}$ of $G$, 
let $f_i$ be the bijection from $V^i$ to $I$, such that $f_i(v^i_s) = a_s$ for every $s\in [n]$.
Clearly, we can use $f_i^{-1}(a_s)$ to denote the unique vertex $v^i_s$ of $V^i$ associated with $a_s \in I$. 
Since $I$ has the 2-sumfree property, 
we have that for any distinct $i, j \in [k]$ and any $u \in V^i$ and any $w \in V^j$, if $f_i(u) + f_j(w) = a_s + a_t$, then $u, w$ is either $v^i_s, v^j_t$ or $v^i_t, v^j_s$.  Moreover, for any distinct color classes $V^i, V^j$, edges $v_r^i v_s^j$ and $v_s^i  v_r^j$ are either both in $E$, or both not in $E$. Hence, by looking at a sum in $I$, we can decide whether it is corresponding to a pair of symmetry edges or non-edges.





Suppose positive integers $p, p_{ij} < n^{210}$ for all distinct $i,j \in [k]$. Assume $p'=n^{210}$. We define  $x= (10a_1)^5$, $y= (nka_1)^{10000}$ and $m=(nka_1)^{20000}$ for any $\alpha \in [0,1]$. Additionally, for each $i\in [k]$, we define integers $\beta$, $s$, $l$, $a_0$, and $t$ as listed in Table~\ref{table:reduction-values} and integers $q, q_1$, and $q_2$ as listed in Table \ref{table:reduction-values-q1-q2}, whose values depend on which case of the theorem we are in.
Table \ref{table:reduction-values-relationship} demonstrates some relationships between these values, which will be used later. We will explain Table \ref{table:reduction-values-relationship} as follows. 
Consider $\alpha = 0$. We have that $a_1 \leq n^4r = n^4 2(k-1)^2k^3 < n^{10}$, that $a_0 = a_1 +1 \leq n^{10}$, that $q_1 < ky + k^2\beta = ky + k^2(nk)^{50}< ky + n^{102}$, that $q_2 < kp + k^2p'  + k n^{10} < kn^{210} + k^2n^{210} + k n^{10} < n^{213}$, that  $q = q_1 + q_2 < ky +  n^{214}$. 
Consider $\alpha \in (0,1)$. We have that $q_1 < 2k^2y = 2k^2(nka_1)^{10000}$, that $q_2 < (k + \binom{k}{2}) (10a_1)^5 + 2ka_1 < 2k^2(10a_1)^5$, that $q < 3k^2y = 3k^2(nka_1)^{10000}$, that $a_0 = (10a_1)^{50} > 2k^2(10a_1)^5 > q_2$.
Consider $\alpha = 1$. We have that $q_1  = 0$, that $q_2 < k^2 n^{10}$, that $q < k^2 n^{10}$, that $a_1 \leq n^4r  < (nk)^5 \leq n^{10}$, that $a_0 < k^2n^{10} \leq n^{12}$.

\begin{table}[]
\begin{center} 
\begin{tabular}{|l|l|l|l|}
\hline
         & Case 1                                                                                & Case 2                                                                                                           & Case 3                                                                                \\ \hline
Variable & \begin{tabular}[c]{@{}l@{}}$\alpha = 0$, \\ $\beta \in \mathbb{Z}$ input\end{tabular} & \begin{tabular}[c]{@{}l@{}}$\alpha \in (0, 1)$ any constant, \\ $\beta \in \mathbb{Z}$ any constant\end{tabular} & \begin{tabular}[c]{@{}l@{}}$\alpha = 1$, \\ $\beta \in \mathbb{Z}$ input\end{tabular} \\ \hline
$r$      & $2(k - 1)^2 k^3$                                                                     & $\left\lceil\frac{10k(|\beta|+10)}{\alpha(1-\alpha)}\right\rceil^{10}$                                                                                                    &     $2(k-1)^2k^3$             \\ \hline
$\beta$  & $(nk)^{50}$                                                                           & fixed                                                                                                            & $-(nk)^{10000}$                                                                       \\ \hline
$s$      & $n^{10}$                                                                              & $0$                                                                                                              & $n^{10}$                                                                              \\ \hline
$l \in \mathbb{N}^0$      &  $< r$ any number                                                                     & $0$                                                                                                              & $0$                                                                                   \\ \hline
$a_0$    & $a_1 + 1$                                                                             & $x^{10} = (10a_1)^{50}$                                                                                          & $kn^{10} + \binom{k}{2} n^{10} + 1$                                                   \\ \hline
$t$      & $s$                                                                                   & $\left\lfloor (1-\alpha)^2\alpha x\right\rfloor + s$                                                             & $s$                                                                                   \\ \hline
\end{tabular}
\caption{The list of variable values to use in the reduction, depending on which case of the theorem we are proving.}
\label{table:reduction-values}
\end{center}
\end{table}

\begin{table}[]
\begin{center} 
\begin{tabular}{|l|l|l|}
\hline
       
\begin{tabular}[c]{@{}l@{}} Case 1: \\ $(\alpha =0)$ \end{tabular} 
& \begin{tabular}[c]{@{}l@{}} $q_1  = ky +  (k+ \binom{k}{2})\beta - (k + 2\binom{k}{2} -3)(s + l)$ \\ $q_2 = kp + \binom{k}{2}p'  + k (s + l -   a_n )$ \\ $q = q_1+q_2$ \end{tabular}   \\ \hline
\begin{tabular}[c]{@{}l@{}} Case 2: \\ $(\alpha \in (0,1))$ \end{tabular} 
& \begin{tabular}[c]{@{}l@{}} $q_1 = \binom{k}{2}(\left\lfloor (1-\alpha)y \right\rfloor +\left\lceil 2\alpha^2(1-\alpha)x \right\rceil)$ \\ $ \; \;\;\;\;\;\;\;\;\;\;\;\;\;\;\;\;\;\;\;\; + k\left( \left\lceil \alpha^2(1-\alpha)x \right\rceil + 2\left\lfloor (1-\alpha)y \right\rfloor \right)$ \\ $q_2 = (k + \binom{k}{2}) t + 2k|\beta|$ \\ $q = q_1+q_2$ \end{tabular}   \\ \hline
\begin{tabular}[c]{@{}l@{}} Case 3: \\ $(\alpha = 1)$ \end{tabular} 
& \begin{tabular}[c]{@{}l@{}} $q_1  = 0$ \\ $q_2 = (k + \binom{k}{2} )s$ \\ $q = q_1+q_2$ \end{tabular}   \\ \hline
 
\end{tabular}
\caption{The values of variables $q, q_1$, and $q_2$ to use in the reduction, depending on which case of the theorem we are proving.}
\label{table:reduction-values-q1-q2}
\end{center}
\end{table}

\begin{table}[]
\begin{center} 
\begin{tabular}{|l|l|l|l|}
\hline
& Case 1    & Case 2       & Case 3              \\ \hline
 
Variable & \begin{tabular}[c]{@{}l@{}}$\alpha = 0$, \\ $\beta \in \mathbb{Z}$ input\end{tabular} & \begin{tabular}[c]{@{}l@{}}$\alpha \in (0, 1)$ any constant, \\ $\beta \in \mathbb{Z}$ any constant\end{tabular} & \begin{tabular}[c]{@{}l@{}}$\alpha = 1$, \\ $\beta \in \mathbb{Z}$ input\end{tabular} \\ \hline
 
$q$   &  $< ky +  n^{214} $  &  $ <3k^2(nka_1)^{10000}$     & $< k^2n^{10}$     \\ \hline
 
$q_1$ & $< ky + n^{102}$ & $<2k^2(nka_1)^{10000}$   &  $= 0$    \\ \hline
 
$q_2$ & $< n^{213}$   & $< 2k^2(10a_1)^5$   &       $< k^2n^{10}$                     \\ \hline
 
$a_1$ & $<  n^{10}$    &    & $<n^{10}$  \\ \hline
 
$a_0$ &  $\leq n^{10}$   & $> q_2$   & $< n^{12}$          \\ \hline
\end{tabular}
\caption{Some relationships between the values of the variables to use in the reduction, depending on which case of the theorem we are proving.}
\label{table:reduction-values-relationship}
\end{center}
\end{table}





We can now start describing the graph $H$, which is illustrated in Figure \ref{fig:a-b-domi-reduction-graph}.
First, add to $H$ two factors $N$ and $K$, where:
\begin{itemize}
    \item
    $N$ is an edgeless graph of size $\left\lfloor \alpha(1-\alpha)m + 2(1-\alpha)(s+l) \right\rfloor$;
    \item
    $K$ is an edgeless graph of size $\left\lfloor \alpha(1-\alpha)m + (1-\alpha)(s+l) \right\rfloor$.
\end{itemize}

Secondly, for each color class $V^i$ of $G$, add six factors $A_i,$ $ B_i,$ $ D_i,$ $ R_i,$ $ S_i,$ $ T_i$ to each $H$, where:
\begin{itemize}
    \item
    $A_i$ and $D_i$ are edgeless graphs of size $\left\lfloor (1-\alpha)y \right\rfloor$;
    \item
    $B_i$ is an edgeless graph of size $\left\lceil \alpha^2(1-\alpha)x \right\rceil$;
    \item 
    $R_i$ is an edgeless graph of size $\left\lfloor |\beta| -s - (1-\alpha)l \right\rfloor$;

    \item 
    Now, consider $S_i$. Recall that $l < r$.
    
    First, $S_i$ is a $(a_0, I \cup \{a_n - 1\}, c)$-degree deletion graph if $\alpha = 1$, and is a $(a_0, I \cup \{a_n - 1\}, c)$-degree deletion star graph if $\alpha \in (0, 1)$, where we put  $c(a_j) = t - \left\lceil\frac{\alpha}{2}a_j\right\rceil$ for $a_j \in I$ and $c(a_n - 1) = a_0$.   
    
    Secondly, if $\alpha = 0$, then $S_i$ is a $(a_0, I, c)$-degree retention graph above $(p,l)$\footnote{We will later prove that $p < n^{210}$ is well-defined here.  Also note that $l < 100$ by definition, and therefore $l < r$ as desired.}, where $c(a_i) = \frac{1}{2}(a_i - a_n)$ and $c(a_0)= ks$.

    \item 
    $T_i$ is an edgeless graph of size $t$.
    
\end{itemize}
We then make $A_i$,  $B_i$, and $D_i$ adjacent with  $K$; make  $R_i$ adjacent with  $A_i$; make $S_i$ adjacent with  $R_i$, $B_i$, and $T_i$; make  $T_i$ adjacent with  $D_i$.

Thirdly, for each pair of color classes $V^i, V^j$ with $i < j$, we add another four factors $U_{ij},  R_{ij},  A_{ij},  B_{ij}$, where:
\begin{itemize}
    \item 
    $R_{ij}$ is an edgeless graph of size $\left\lfloor |\beta| - 2s - 2(1-\alpha)l \right\rfloor$;
    
    \item 
    $A_{ij}$ is an edgeless graph of size $\left\lfloor (1-\alpha)y \right\rfloor$;
    
    \item 
    $B_{ij}$ is an edgeless graph of size $\left\lceil 2\alpha^2 (1-\alpha) x \right\rceil$;

    \item 
    Now, we consider $U_{ij}$. Recall that $l < r$.
    Define $I_{ij} = \{a + b : f_i^{-1}(a) f_j^{-1}(b) \in E(G)\}$, i.e. the sums of pairs that correspond to an edge between $V^i$ and $V^j$.
    Let $\ell_{ij}$ and $\hbar_{ij}$ be the smallest element and the largest element of $I_{ij}$, respectively.

    First, $U_{ij}$ is a $(a_0, I_{ij} \cup$ $\{\ell_{ij} - 1\}, c_{ij})$-degree deletion graph if $\alpha = 1$, and is a $(a_0, I_{ij} \cup$ $\{\ell_{ij} - 1\}, c_{ij})$-degree deletion star graph if $0<\alpha <1$, where we put $c_{ij}(a + b) = t -  \frac{1}{k-1} (\left\lceil \frac{\alpha a}{2} \right\rceil + \left\lceil \frac{\alpha b}{2} \right\rceil)$, and we put $c_{ij}(\ell_{ij} - 1) = a_0$.

     Secondly, if $\alpha = 0$, then $U_{ij}$ consists of an edgeless graph with $p'_{ij}$ vertices together with  a $(2a_1 + 1, I_{ij}, c_{ij})$-degree retention graph above $(p_{ij},l)$\footnote{We will also later prove that $p_{ij} < n^{210}$ is well-defined here.},  
     where $c_{ij}(a+b) = \frac{1}{2(k-1)}(a + b - l_{ij})$, and $p'_{ij}=  p'  - p_{ij} + \frac{1}{2(k-1)}(l_{ij} - 2a_n)$. In addition, we define $c_{ij}(2a_1+1)= ks$.
\end{itemize}
We then make $U_{ij}$ adjacent with $T_i$, $T_j$, $B_{ij}$, and $R_{ij}$; make $A_{ij}$ adjacent with $R_{ij}$ and $N$; make $B_{ij}$ adjacent with $K$.

\begin{figure}
\centering
\begin{tikzpicture}

\filldraw[color=blue, fill=blue!5] (-2,0) ellipse (0.618 and 0.382); 
\filldraw[color=blue, fill=blue!5] (8,0) ellipse (0.618 and 0.382); 
\filldraw[color=blue, fill=blue!5] (-2.8,1.5) rectangle (0.8-2,2.5); 
\filldraw[color=blue, fill=blue!5]  (7.2,1.5) rectangle (8.8,2.5); 
\filldraw[color=blue, fill=blue!5] (-2,3.7) ellipse (0.618 and 0.382); 
\filldraw[color=blue, fill=blue!5] (0.5,3.7) ellipse (0.618 and 0.382); 
\filldraw[color=blue, fill=blue!5] (3,3.7) ellipse (0.618 and 0.382); 
\filldraw[color=blue, fill=blue!5] (5.5,3.7) ellipse (0.618 and 0.382); 
\filldraw[color=blue, fill=blue!5] (3,2) ellipse (0.618 and 0.382); 
\filldraw[color=blue, fill=blue!5] (0.5,2) ellipse (0.618 and 0.382); 
\filldraw[color=blue, fill=blue!5] (5.5,2) ellipse (0.618 and 0.382); 
\filldraw[color=blue, fill=blue!5] (0.5,0) ellipse (0.618 and 0.382); 
\filldraw[color=blue, fill=blue!5] (5.5,0) ellipse (0.618 and 0.382); 
\filldraw[color=blue, fill=blue!5] (8,3.7) ellipse (0.618 and 0.382); 
\filldraw[color=blue, fill=blue!5]  (2.1,5.05) rectangle (3.9,6.15); 
\filldraw[color=blue, fill=blue!5]   (3,7.5) ellipse (0.618 and 0.382); 
\filldraw[color=blue, fill=blue!5]   (3,9) ellipse (0.618 and 0.382); 
\filldraw[color=blue, fill=blue!5]   (3,10.5) ellipse (0.618 and 0.382); 

\draw[color=black] (-3.2,0) node {$R_i$};
\draw[color=black] (9.2,0) node {$R_j$};
\draw[color=black] (-3.5,2) node {$S_i$};
\draw[color=black] (9.5,2) node {$S_j$};
\draw[color=black] (-3.2,3.7) node {$T_i$};
\draw[color=black] (9.2,3.7) node {$T_j$};
\draw[color=black] (2,4.7) node {$U_{ij}$};
\draw[color=black] (2,3.7) node {$B_{ij}$};
\draw[color=black] (2.3,1.5) node {$K$};
\draw[color=black] (1.7,7.5) node {$R_{ij}$};
\draw[color=black] (1.7,7.5+1.5) node {$A_{ij}$};
\draw[color=black] (1.7,7.5+3) node {$N$};

\draw[color=black] (-0.3,4.1) node {$D_{i}$};
\draw[color=black] (-0.3,4.1-1.7) node {$B_{i}$};
\draw[color=black] (-0.3,0.4) node {$A_{i}$};

\draw[color=black] (6+0.3,4.1) node {$D_{j}$};
\draw[color=black] (6+0.3,4.1-1.7) node {$B_{j}$};
\draw[color=black] (6+0.3,0.4) node {$A_{j}$};


\draw [line width=0.5mm] (3,9.382) -- (3,10.118);
\draw [line width=0.5mm] (3,9.382-1.5) -- (3,10.118-1.5);
\draw [line width=0.5mm] (3,6.15) -- (3,10.118-3);
\draw [line width=0.5mm] (3,3.7+0.382) -- (3,5.05);
\draw [line width=0.5mm] (3,2.382) -- (3,10.118-8.5+1.7);

\draw [line width=0.5mm] (0.618-2,3.7) -- (0.5-0.618,3.7);
\draw [line width=0.5mm] (0.618-2+7.5,3.7) -- (0.5-0.618+7.5,3.7);

\draw [line width=0.5mm] (0.8-2,2) -- (0.5-0.618,2);
\draw [line width=0.5mm] (0.618-2+7.5,2) -- (7.2,2);

\draw [line width=0.5mm] (0.618-2,0) -- (0.5-0.618,0);
\draw [line width=0.5mm] (0.618-2+7.5,0) -- (0.5-0.618+7.5,0);

\draw [line width=0.5mm] (0.618-2+2.5,3.7) -- (3-0.618,2);
\draw [line width=0.5mm] (0.618-2+2.5,2) -- (3-0.618,2);
\draw [line width=0.5mm] (0.618-2+2.5,0) -- (3-0.618,2);

\draw [line width=0.5mm] (0.618-2+5,2) -- (0.5-0.618+5,3.7);
\draw [line width=0.5mm] (0.618-2+5,2) -- (0.5-0.618+5,2);
\draw [line width=0.5mm] (0.618-2+5,2) -- (0.5-0.618+5,0);

\draw [line width=0.5mm] (-2,3.7+0.382) -- (2.1,5.6);
\draw [line width=0.5mm] (8,3.7+0.382) -- (3.9,5.6);

\draw [line width=0.5mm] (-2,3.7-0.382) -- (-2,2.5);
\draw [line width=0.5mm] (-2,0.382) -- (-2,1.5);

\draw [line width=0.5mm] (8,3.7-0.382) -- (8,2.5);
\draw [line width=0.5mm] (8,0.382) -- (8,1.5);

\draw [line width=0.5mm] (-2,3.7+0.382) -- (-2,3.7+0.382+0.5);
\draw [line width=0.5mm] (-2,3.7+0.382) -- (-2.3,3.7+0.382+0.5);
\draw [line width=0.5mm] (-2,3.7+0.382) -- (-2.6,3.7+0.382+0.5);

\draw [line width=0.5mm] (8,3.7+0.382) -- (6+2,3.7+0.382+0.5);
\draw [line width=0.5mm] (8,3.7+0.382) -- (6+2.3,3.7+0.382+0.5);
\draw [line width=0.5mm] (8,3.7+0.382) -- (6+2.6,3.7+0.382+0.5);

\draw [line width=0.5mm] (3,2-0.382) -- (3-0.3,2-0.382-0.5);
\draw [line width=0.5mm] (3,2-0.382) -- (3,2-0.382-0.5);
\draw [line width=0.5mm] (3,2-0.382) -- (3+0.3,2-0.382-0.5);

\draw [line width=0.5mm] (2.5,10.3) -- (2.5,9.8);
\draw [line width=0.5mm] (2.5,10.3) -- (2.5-0.3,9.8);
\draw [line width=0.5mm] (2.5,10.3) -- (2.5-0.6,9.8);

\draw [line width=0.5mm] (3.5,10.3) -- (3.5,9.8);
\draw [line width=0.5mm] (3.5,10.3) -- (3.5+0.3,9.8);
\draw [line width=0.5mm] (3.5,10.3) -- (3.5+0.6,9.8);

\fill [color=black] (-2.4,0) circle (1pt);
\fill [color=black] (-2.2,0) circle (1pt);
\fill [color=black] (-2,0) circle (0.4pt);
\fill [color=black] (0.1-2,0) circle (0.4pt);
\fill [color=black] (0.2-2,0) circle (0.4pt);
\fill [color=black] (0.4-2,0) circle (1pt);

\fill [color=black] (6-0.4+2,0) circle (1pt);
\fill [color=black] (6-0.2+2,0) circle (1pt);
\fill [color=black] (6+2,0) circle (0.4pt);
\fill [color=black] (6.1+2,0) circle (0.4pt);
\fill [color=black] (8.2,0) circle (0.4pt);
\fill [color=black] (8.4,0) circle (1pt);

\fill [color=black] (-0.4-2,2.4) circle (1pt);
\fill [color=black] (-0.7-2,2.1) circle (1pt);
\draw (-0.4-2,2.4) -- (-0.7-2,2.1);
\fill [color=black] (-0.56-2,2.1) circle (1pt);
\draw (-0.4-2,2.4) -- (-0.56-2,2.1);
\fill [color=black] (-0.43-2,2.1) circle (0.4pt);
\fill [color=black] (-0.32-2,2.1) circle (0.4pt);
\fill [color=black] (-0.21-2,2.1) circle (0.4pt);
\fill [color=black] (-0.1-2,2.1) circle (1pt);
\draw (-0.4-2,2.4) -- (-0.1-2,2.1);

\fill [color=black] (0.4-2,2.4) circle (1pt);
\fill [color=black] (0.1-2,2.1) circle (1pt);
\draw (0.4-2,2.4) -- (0.1-2,2.1);
\fill [color=black] (0.26-2,2.1) circle (1pt);
\draw (0.4-2,2.4) -- (0.26-2,2.1);
\fill [color=black] (0.37-2,2.1) circle (0.4pt);
\fill [color=black] (0.48-2,2.1) circle (0.4pt);
\fill [color=black] (0.59-2,2.1) circle (0.4pt);
\fill [color=black] (0.7-2,2.1) circle (1pt);
\draw (0.4-2,2.4) -- (0.7-2,2.1);

\fill [color=black] (-0.29-2,1.7) circle (0.4pt);
\fill [color=black] (-0.4-2,1.7) circle (0.4pt);
\fill [color=black] (-0.51-2,1.7) circle (0.4pt);

\fill [color=black] (0.4-2,1.9) circle (1pt);
\fill [color=black] (0.1-2,1.6) circle (1pt);
\draw (0.4-2,1.9) -- (0.1-2,1.6);
\fill [color=black] (0.26-2,1.6) circle (1pt);
\draw (0.4-2,1.9) -- (0.26-2,1.6);
\fill [color=black] (0.37-2,1.6) circle (0.4pt);
\fill [color=black] (0.48-2,1.6) circle (0.4pt);
\fill [color=black] (0.59-2,1.6) circle (0.4pt);
\fill [color=black] (0.7-2,1.6) circle (1pt);
\draw (0.4-2,1.9) -- (0.7-2,1.6);

\fill [color=black] (5.6+2,2.4) circle (1pt);
\fill [color=black] (5.3+2,2.1) circle (1pt);
\draw (5.6+2,2.4) -- (5.3+2,2.1);
\fill [color=black] (6-0.56+2,2.1) circle (1pt);
\draw (5.6+2,2.4) -- (6-0.56+2,2.1);
\fill [color=black] (6-0.43+2,2.1) circle (0.4pt);
\fill [color=black] (6-0.32+2,2.1) circle (0.4pt);
\fill [color=black] (6-0.21+2,2.1) circle (0.4pt);
\fill [color=black] (6-0.1+2,2.1) circle (1pt);
\draw (5.6+2,2.4) -- (5.9+2,2.1);

\fill [color=black] (6.4+2,2.4) circle (1pt);
\fill [color=black] (6.1+2,2.1) circle (1pt);
\draw (6.4+2,2.4) -- (6.1+2,2.1);
\fill [color=black] (6.26+2,2.1) circle (1pt);
\draw (6.4+2,2.4) -- (6.26+2,2.1);
\fill [color=black] (6.37+2,2.1) circle (0.4pt);
\fill [color=black] (6.48+2,2.1) circle (0.4pt);
\fill [color=black] (6.59+2,2.1) circle (0.4pt);
\fill [color=black] (6.7+2,2.1) circle (1pt);
\draw (6.4+2,2.4) -- (6.7+2,2.1);

\fill [color=black] (6-0.29+2,1.7) circle (0.4pt);
\fill [color=black] (6-0.4+2,1.7) circle (0.4pt);
\fill [color=black] (6-0.51+2,1.7) circle (0.4pt);

\fill [color=black] (6.4+2,1.9) circle (1pt);
\fill [color=black] (6.1+2,1.6) circle (1pt);
\draw (6.4+2,1.9) -- (6.1+2,1.6);
\fill [color=black] (6.26+2,1.6) circle (1pt);
\draw (6.4+2,1.9) -- (6.26+2,1.6);
\fill [color=black] (6.37+2,1.6) circle (0.4pt);
\fill [color=black] (6.48+2,1.6) circle (0.4pt);
\fill [color=black] (6.59+2,1.6) circle (0.4pt);
\fill [color=black] (6.7+2,1.6) circle (1pt);
\draw (6.4+2,1.9) -- (6.7+2,1.6);

\fill [color=black] (-0.4-2,3.7) circle (1pt);
\fill [color=black] (-0.2-2,3.7) circle (1pt);
\fill [color=black] (0-2,3.7) circle (0.4pt);
\fill [color=black] (0.1-2,3.7) circle (0.4pt);
\fill [color=black] (0.2-2,3.7) circle (0.4pt);
\fill [color=black] (0.4-2,3.7) circle (1pt);

\fill [color=black] (-0.4+0.5,3.7) circle (1pt);
\fill [color=black] (-0.2+0.5,3.7) circle (1pt);
\fill [color=black] (0+0.5,3.7) circle (0.4pt);
\fill [color=black] (0.1+0.5,3.7) circle (0.4pt);
\fill [color=black] (0.2+0.5,3.7) circle (0.4pt);
\fill [color=black] (0.4+0.5,3.7) circle (1pt);

\fill [color=black] (-0.4+0.5,2) circle (1pt);
\fill [color=black] (-0.2+0.5,2) circle (1pt);
\fill [color=black] (0+0.5,2) circle (0.4pt);
\fill [color=black] (0.1+0.5,2) circle (0.4pt);
\fill [color=black] (0.2+0.5,2) circle (0.4pt);
\fill [color=black] (0.4+0.5,2) circle (1pt);

\fill [color=black] (-0.4+0.5,0) circle (1pt);
\fill [color=black] (-0.2+0.5,0) circle (1pt);
\fill [color=black] (0+0.5,0) circle (0.4pt);
\fill [color=black] (0.1+0.5,0) circle (0.4pt);
\fill [color=black] (0.2+0.5,0) circle (0.4pt);
\fill [color=black] (0.4+0.5,0) circle (1pt);

\fill [color=black] (-0.4+3,3.7) circle (1pt);
\fill [color=black] (-0.2+3,3.7) circle (1pt);
\fill [color=black] (0+3,3.7) circle (0.4pt);
\fill [color=black] (0.1+3,3.7) circle (0.4pt);
\fill [color=black] (0.2+3,3.7) circle (0.4pt);
\fill [color=black] (0.4+3,3.7) circle (1pt);

\fill [color=black] (-0.4+3,9) circle (1pt);
\fill [color=black] (-0.2+3,9) circle (1pt);
\fill [color=black] (0+3,9) circle (0.4pt);
\fill [color=black] (0.1+3,9) circle (0.4pt);
\fill [color=black] (0.2+3,9) circle (0.4pt);
\fill [color=black] (0.4+3,9) circle (1pt);

\fill [color=black] (-0.4+3,10.5) circle (1pt);
\fill [color=black] (-0.2+3,10.5) circle (1pt);
\fill [color=black] (0+3,10.5) circle (0.4pt);
\fill [color=black] (0.1+3,10.5) circle (0.4pt);
\fill [color=black] (0.2+3,10.5) circle (0.4pt);
\fill [color=black] (0.4+3,10.5) circle (1pt);

\fill [color=black] (-0.4+3,2) circle (1pt);
\fill [color=black] (-0.2+3,2) circle (1pt);
\fill [color=black] (0+3,2) circle (0.4pt);
\fill [color=black] (0.1+3,2) circle (0.4pt);
\fill [color=black] (0.2+3,2) circle (0.4pt);
\fill [color=black] (0.4+3,2) circle (1pt);

\fill [color=black] (-0.4+5.5,3.7) circle (1pt);
\fill [color=black] (-0.2+5.5,3.7) circle (1pt);
\fill [color=black] (0+5.5,3.7) circle (0.4pt);
\fill [color=black] (0.1+5.5,3.7) circle (0.4pt);
\fill [color=black] (0.2+5.5,3.7) circle (0.4pt);
\fill [color=black] (0.4+5.5,3.7) circle (1pt);

\fill [color=black] (-0.4+5.5,2) circle (1pt);
\fill [color=black] (-0.2+5.5,2) circle (1pt);
\fill [color=black] (0+5.5,2) circle (0.4pt);
\fill [color=black] (0.1+5.5,2) circle (0.4pt);
\fill [color=black] (0.2+5.5,2) circle (0.4pt);
\fill [color=black] (0.4+5.5,2) circle (1pt);

\fill [color=black] (-0.4+5.5,0) circle (1pt);
\fill [color=black] (-0.2+5.5,0) circle (1pt);
\fill [color=black] (0+5.5,0) circle (0.4pt);
\fill [color=black] (0.1+5.5,0) circle (0.4pt);
\fill [color=black] (0.2+5.5,0) circle (0.4pt);
\fill [color=black] (0.4+5.5,0) circle (1pt);

\fill [color=black] (6-0.4+2,3.7) circle (1pt);
\fill [color=black] (6-0.2+2,3.7) circle (1pt);
\fill [color=black] (6+2,3.7) circle (0.4pt);
\fill [color=black] (6.1+2,3.7) circle (0.4pt);
\fill [color=black] (6.2+2,3.7) circle (0.4pt);
\fill [color=black] (6.4+2,3.7) circle (1pt);

\fill [color=black] (3-0.4,2.6+2.4+1) circle (1pt);
\fill [color=black] (3-0.7,2.6+2.1+1) circle (1pt);
\draw (3-0.4,2.6+2.4+1) -- (3-0.7,2.6+2.1+1);
\fill [color=black] (3-0.56,2.6+2.1+1) circle (1pt);
\draw (3-0.4,2.6+2.4+1) -- (3-0.56,2.6+2.1+1);
\fill [color=black] (3-0.43,2.6+2.1+1) circle (0.4pt);
\fill [color=black] (3-0.32,2.6+2.1+1) circle (0.4pt);
\fill [color=black] (3-0.21,2.6+2.1+1) circle (0.4pt);
\fill [color=black] (3-0.1,2.6+2.1+1) circle (1pt);
\draw (3-0.4,2.6+2.4+1) -- (3-0.1,2.6+2.1+1);

\fill [color=black] (3.4,2.6+2.4+1) circle (1pt);
\fill [color=black] (3.1,2.6+2.1+1) circle (1pt);
\draw (3.4,2.6+2.4+1) -- (3.1,2.6+2.1+1);
\fill [color=black] (3.26,2.6+2.1+1) circle (1pt);
\draw (3.4,2.6+2.4+1) -- (3.26,2.6+2.1+1);
\fill [color=black] (3.37,2.6+2.1+1) circle (0.4pt);
\fill [color=black] (3.48,2.6+2.1+1) circle (0.4pt);
\fill [color=black] (3.59,2.6+2.1+1) circle (0.4pt);
\fill [color=black] (3.7,2.6+2.1+1) circle (1pt);
\draw (3.4,2.6+2.4+1) -- (3.7,2.6+2.1+1);

\fill [color=black] (3-0.29,2.6+1.7+1) circle (0.4pt);
\fill [color=black] (3-0.4,2.6+1.7+1) circle (0.4pt);
\fill [color=black] (3-0.51,2.6+1.7+1) circle (0.4pt);

\fill [color=black] (3.4,2.6+1.9+1) circle (1pt);
\fill [color=black] (3.1,2.6+1.6+1) circle (1pt);
\draw (3.4,2.6+1.9+1) -- (3.1,2.6+1.6+1);
\fill [color=black] (3.26,2.6+1.6+1) circle (1pt);
\draw (3.4,2.6+1.9+1) -- (3.26,2.6+1.6+1);
\fill [color=black] (3.37,2.6+1.6+1) circle (0.4pt);
\fill [color=black] (3.48,2.6+1.6+1) circle (0.4pt);
\fill [color=black] (3.59,2.6+1.6+1) circle (0.4pt);
\fill [color=black] (3.7,2.6+1.6+1) circle (1pt);
\draw (3.4,2.6+1.9+1) -- (3.7,2.6+1.6+1);

\fill [color=black] (3-0.4,6.5+1) circle (1pt);
\fill [color=black] (3-0.2,6.5+1) circle (1pt);
\fill [color=black] (3,6.5+1) circle (0.4pt);
\fill [color=black] (3.1,6.5+1) circle (0.4pt);
\fill [color=black] (3.2,6.5+1) circle (0.4pt);
\fill [color=black] (3.4,6.5+1) circle (1pt);

\end{tikzpicture}
\caption{Illustration of the construction of $H$. Factors $S_i$, $S_j$, and $U_{ij}$ are degree retention graphs if $\alpha = 0$, and are degree deletion graphs if $\alpha \in (0,1]$. All other factors are edgeless graphs (including the null graphs). Each thick line between two factors means that the two linked factors are adjacent. The thick lines on factors $T_i$, $T_j$, $K$, and $N$ mean that they are adjacent to some other factors of $H$ that are not illustrated in this figure.}
\label{fig:a-b-domi-reduction-graph}
\end{figure}
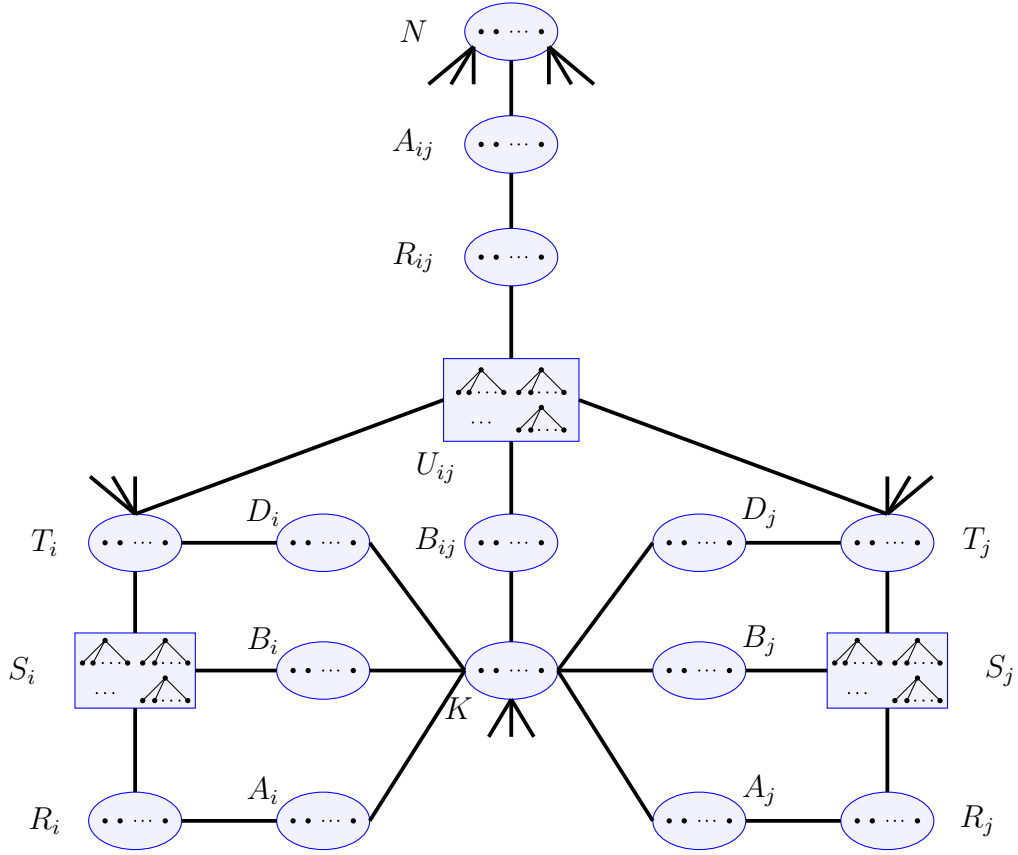

To avoid cumbersome notation, for each $(\alpha, \beta)$, we define $U_{ij} = U_{ji}$, $R_{ij} = R_{ji}$, $c_{ij} = c_{ji}$,  $B_{ij} = B_{ji}$, and $A_{ij} = A_{ji}$. This concludes the construction of $H$.


\begin{table}[]
\begin{center} 
\begin{tabular}{|l|l|l|l|}
\hline
& Case 1    & Case 2       & Case 3              \\ \hline

Variable   &  $\alpha = 0$  &  $\alpha \in (0, 1)$     &   $\alpha = 1$    \\ \hline
 
$|V(S_i)| \leq $   &  $ n^{210} $  &  $(10a_1)^{1000}$     & $ n^{240}$     \\ \hline
 
$|V(U_{ij})| \leq $ & $ n^{211} $ & $(10a_1)^{1000}$  & $ n^{240}$      \\ \hline
\end{tabular}
\caption{The bounds of the vertex numbers of $S_i$ and $U_{ij}$ for distinct $i,j \in [k]$.}
\label{table: size of si and uij}
\end{center}
\end{table}



The $S_i$ and $U_{ij}$ factors form the main components of the reduction, and we first provide bounds on their size (vertex number).

Consider $i \in [k]$ and the $S_i$ factor. According to definition \ref{def:kd-degdeletion-bdd-121}, we have that the size of $S_i$ is at most $(a_0c(a_n - 1))^{10}= (a_0)^{20}$ if $\alpha \in (0,1]$. More specifically, if $\alpha = 1$, then $(a_0)^{20} < (n^{12})^{20} =  n^{240}$ based on Table \ref{table:reduction-values-relationship}. If $\alpha \in (0,1)$, then $(a_0)^{20} = ((10a_1)^{50})^{20} = (10a_1)^{1000}$  based on Table \ref{table:reduction-values}.  
Suppose that $\alpha =0$. According to definition \ref{def:kd-degdeletion}, the size of $S_i$ is at most $(a_0c(a_0))^{10}$, which is $(a_0ks)^{10}\leq (n^{10}kn^{10})^{10} \leq n^{210}$ based on Table \ref{table:reduction-values} and Table \ref{table:reduction-values-relationship}. In addition, according to Definition \ref{def:kd-degdeletion}, we have $p< (a_0c(a_0))^{10}$, which is at most $n^{210}$. So, the restriction that $p< n^{210}$ does not affect the definition of $S_i$.

We next bound the size of each $U_{ij}$ as follows.   According to definition \ref{def:kd-degdeletion-bdd-121}, we have that the size of $U_{ij}$ is at most $(a_0c_{ij}(\ell_{ij} - 1))^{10}= (a_0)^{20}$ if $\alpha \in (0,1]$. Moreover, we have already known that $(a_0)^{20} < n^{240}$ if $\alpha = 1$, and that $(a_0)^{20} = (10a_1)^{1000}$ if $\alpha \in (0,1)$. Assume that $\alpha =0$. According to definition \ref{def:kd-degdeletion},  the size of the $(2a_1 + 1, I_{ij}, c_{ij})$-degree retention graph above $(p_{ij},l)$ is at most $(a_0c_{ij}(2a_1+1))^{10} = (a_0ks)^{10}$, which has been proved to be at most $n^{210}$. Additionally, according to Definition \ref{def:kd-degdeletion}, we have $p< (a_0c(a_0))^{10}$, which is at most $n^{210}$. So, the restriction that $p< n^{210}$ does not affect the definition of $U_{ij}$. Furthermore, since $p'= n^{210}$ and $p_{ij}$ is a positive integer according to Definition \ref{def:kd-degdeletion}, we have
\begin{align*}
p'_{ij} & =  p'  - p_{ij} + \frac{1}{2(k-1)}(l_{ij} - 2a_n) < p'  + \frac{1}{2(k-1)}(l_{ij} - 2a_n)  \\
&< p'  + \frac{1}{2(k-1)}2a_1 < n^{210}  + a_1 < n^{210}  + n^{10} < 2n^{210}.
\end{align*}
Therefore, the size of $U_{ij}$ is at most $p'_{ij} + (a_0c(2a_1+1))^{10} <  2n^{210} + n^{210} < n^{211}$. For convenience, we summarize these bounds in Table \ref{table: size of si and uij}.



As we argued, the set $I$ can be constructed in polynomial time with maximum value $O(rn^4)$ and, by the definition of admitting arbitrary degree deletion or retention tables, the $S_i$ and $U_{ij}$ factors can also be built in polynomial time. Moreover, all other factors are edgeless graphs with sizes polynomial in $n+k$. Thus $H$ can be constructed in polynomial time.
All that remains is to argue that the objective instance is equivalent to the original one. Lemma \ref{ldlem:first-direction} and \ref{lem:ld-second-dir}  will give the forward and reverse directions, respectively.

\begin{lemma}\label{ldlem:first-direction}
Suppose that $G$ contains a multicolored clique $C$ with $k$ vertices. Then there is $X \subseteq V(H)$ of size at most $q$ such that $|N(v)\cap X| \geq \alpha |N(v)| + \beta$ for every $v \in V(H) \setminus X$.
\end{lemma}

\begin{proof}
For $i \in [k]$, define $\ha_i$ as the unique element of $I$ such that $f^{-1}_{i}(\ha_i)$ is the vertex of $V^i$ that belongs to the multicolored clique $C$. Then for any distinct $i, j \in [k]$, $f^{-1}_{i}(\ha_i)  f^{-1}_{j}(\ha_j) \in E(G)$.  
This implies that $\ha_i + \ha_j$ is in the $I_{ij}$ list that was used to construct $U_{ij}$.  
We divide the proof into two cases, which is $\alpha = 0$ and $0< \alpha \leq  1$.

\textit{Case 1:} consider $\alpha = 0$. Recall that $\beta = (nk)^{50}$ in this case.  Define $w_i = s + l - \ha_i$.  Note that for any $a_i \in I$, $0\leq l < r \leq a_i \leq a_1 < n^{10} = s$ based on Table \ref{table:reduction-values} and Table \ref{table:reduction-values-relationship}.  Therefore, we have $0 < w_i < s$.
First, we show how to construct the vertex set $X$ for $H$.
For all $i\neq j$, the sets $B_i$, $B_{ij}$ are empty sets, so we do not need to consider them.
Add all vertices of $K$ and $N$ to $X$.
Let $i \in [k]$. Add all vertices of $R_i$ and $D_i$ to $X$, and the intersection of $A_i$ and $X$ is empty. 
In $S_i$, we add $p + c(\ha_i) = p+ \frac{1}{2}(\ha_i- a_n)$ vertices to $X$.  This can be done so that $S_i - X$ has minimum degree $\ha_i$ in $X\cap S_i$ according to Condition 2 of Definition \ref{def:kd-degdeletion}. For each $T_i$, add any $w_i$ of its vertices to $X$, which is feasible since $0 < w_i < s$ and $|T_i| = s$. 
Next consider distinct $i, j \in [k]$.  Add all vertices of $R_{ij}$ to $X$, and make the intersection of $A_{ij}$ and $X$ empty.
For $U_{ij}$, we first add all $p'_{ij}$ vertices of its edgeless graph to $X$. Then, we add $p_{ij} + c_{ij}(\ha_i + \ha_j)$ vertices to $X$ so that  $U_{ij} - X$ has minimum degree $\ha_i + \ha_j$ in $X\cap U_{ij}$, again according to Condition 2 of Definition \ref{def:kd-degdeletion}. 
We claim that $X$ is a solution to the instance $H$.

Secondly, we bound the size of $X$.
The number of elements of $X$ from the $S_i$'s is 
\begin{align*}
    \sum_{i \in [k]} |S_i \cap X| = \sum_{i \in [k]}\left(p+ \frac{1}{2}(\ha_i-a_n)\right) = kp -\frac{1}{2} k a_n  + \frac{1}{2} \sum_{i \in [k]} \ha_i.
\end{align*}

Next consider distinct $i, j \in [k]$.
Recall that we defined $p'_{ij}=p' - p_{ij} + \frac{1}{2(k-1)}(l_{ij} - 2a_n)$. Thus
\begin{align*}
    |U_{ij} \cap X| &= p'_{ij}+ p_{ij} + c_{ij}(\ha_i + \ha_j)\\
    &= p' - p_{ij} + \frac{1}{2(k-1)}(l_{ij} - 2a_n) + p_{ij} +  \frac{1}{2(k-1)}(\ha_i + \ha_j - l_{ij})\\
    &= p' + \frac{1}{2(k-1)}(\ha_i + \ha_j - 2a_n). 
\end{align*}
The number of elements of $X$ from the $U_{ij}$'s is
\begin{align*}
\sum_{1 \leq i < j \leq k} |U_{ij} \cap X| &= \sum_{1 \leq i < j \leq k}\left(p' + \frac{1}{2(k-1)}(\ha_i + \ha_j - 2a_n)\right)\\
&= \binom{k}{2}p' - \binom{k}{2} \frac{1}{2(k-1)} 2a_n + \sum_{i \in [k]}(k - 1) \frac{1}{2(k-1)} \ha_i  \\
&= \binom{k}{2}p' -\frac{1}{2} k a_n + \frac{1}{2} \sum_{i \in [k]} \ha_i,
\end{align*}
where the equality on the second line is because each element of $\ha_1, \ldots, \ha_k$ appears exactly $k - 1$ times in the summation.
In addition, the size of the intersection of $X$ and each $T_{i}$ is $w_i = s+l-a_i$. So the size of the intersection of $X$ and all $U_{ij}$, $S_{i}$, and $T_{i}$ factors is
\begin{align*}
&~kp -\frac{1}{2} k a_n  + \frac{1}{2} \sum_{i \in [k]} \ha_i + \binom{k}{2}p' -\frac{1}{2} k a_n + \frac{1}{2} \sum_{i \in [k]} \ha_i +  \sum_{i \in [k]} w_i\\
=&~kp + \binom{k}{2}p' - k a_n  + \sum_{i \in [k]} \ha_i + \sum_{i \in [k]} \left(s+l-\ha_i\right)\\
=&~kp + \binom{k}{2}p' - k a_n  + \sum_{i \in [k]} \left(s+l \right)\\
=&~kp + \binom{k}{2}p'  + k \left(s + l - a_n  \right)= q_2.
\end{align*}

In addition, the number of vertices in $X$, other than those from the $S_i, T_i$, and $U_{ij}$ factors, is
\begin{align*}
&\sum_{1 \leq i < j \leq k}|R_{ij}| + \sum_{i \in [k]}\left( |D_{i}| + |R_{i}|\right) + |N| + |K|,
\end{align*}
which is equal to $q_1$ after simple calculation. Overall, the size of $X$ is $q_1 + q_2 = q$.

Thirdly, we verify that $X$ is a $(\alpha,\beta)$-linear degree dominating set of $H$. Since $\alpha =0$, we only need to demonstrate that $|N(v)\cap X| \geq \beta$ for every vertex $v \in V(H) \setminus X$.
The inequality holds for every vertex of $A_i$ of each $i \in [k]$ since $|R_i| + |K| = |\beta| = \beta$, which holds because the floor functions can be ignored when $\alpha = 1$. 
In addition, the inequality holds for every vertex of $A_{ij}$ for every distinct $i, j \in [k]$, since $|R_{ij}| + |K| = |\beta| = \beta$. 
Consider a vertex $v$ in some $T_i$. We have $|N(v)\cap X| \geq |D_i| = y$, where $y$ is equal to $ (nka_1)^{10000} > (nk)^{50} = \beta$. 
Consider a vertex in some $S_i$. We know $S_i - X$ has minimum degree $\ha_i$ in $X\cap S_i$. In addition, the size of the intersection of $T_i$ and $X$ is $s+l-\ha_i$ and $R_i \subseteq X$. Thus, for any $v\in S_i - X$, we have that $|N(v)\cap X|$ is at least
\begin{align*}
\ha_i + (s+l-\ha_i) + (|\beta|- s - l) = \beta.
\end{align*}
Consider a vertex in some $U_{ij}$. We know that $U_{ij}$ has minimum degree $\ha_i + \ha_j$ in $X\cap U_{ij}$. Moreover, $R_{ij} \subseteq X$ and the intersection of $T_i \cup T_j$ and $X$ contains $2s+2l-\ha_i-\ha_j$ vertices. Thus, for any $v\in U_{ij} - X$, we have that $|N(v)\cap X|$ is at least
\begin{align*}
(\ha_i+\ha_j) + (|\beta|- 2s - 2l) + (2s+2l-\ha_i-\ha_j) = \beta.
\end{align*}
We have handled every non-empty factor of $H - X$, which completes this part of the proof.

\textit{Case 2:} consider $0< \alpha \leq  1$. Recall that $\beta$ is some arbitrary constant if $\alpha \in (0,1)$, and $\beta = -(nk)^{10000} $ if $\alpha = 1$.
We define $w_i$ as follows:
\begin{itemize}
    \item 
    if $\alpha \in (0,1)$, then  $w_i = 2(|\beta|+ \left\lceil \frac{\alpha}{2} \ha_i\right\rceil)$;
    \item 
    if $\alpha = 1$, then $w_i = 2\left\lceil \frac{\alpha}{2} \ha_i\right\rceil$.
\end{itemize}
In addition, when $\alpha = 1$, $w_i$ is also equal to $\ha_i$ since $\ha_i$ is an even number in this case (we leave $w_i$ in the above form to relate it more easily to the $c(\ha_i)$'s later on).

First, we show how to construct the vertex set $X$.
The intersection of $K \cup N$ and $X$ is an empty set. 
Let $i \in [k]$. Add all vertices of $A_i$, $B_i$ and $D_i$ to $X$. The intersection of $R_i$ and $X$ is empty. 
In $T_i$, add any of its $w_i$ vertices to $X$.
In $S_i$, we add $c(\ha_i) = t - \left\lceil\frac{\alpha}{2}\ha_i\right\rceil$ vertices to $X$.  This can be done so that $S_i - X$ has maximum degree $\ha_i$ according to Definition \ref{def:kd-degdeletion-bdd-121}. 
Consider distinct $i, j \in [k]$. Add all vertices of $A_{ij}$ and $B_{ij}$ to $X$ and the intersection of $R_{ij}$ and $X$ is empty.
For $U_{ij}$, we add
\begin{align*}
    c_{ij}(\ha_i + \ha_j) =  t -  \frac{1}{k-1} \left( \left\lceil \frac{\alpha}{2} \ha_i \right\rceil + \left\lceil \frac{\alpha}{2} \ha_j \right\rceil \right)
\end{align*}
vertices to $X$ so that  $U_{ij}- X$ has maximum degree $\ha_i + \ha_j$ according to Definition \ref{def:kd-degdeletion-bdd-121}.

Secondly, we bound the size of $X$. The size of the intersection of $X$ and all $S_{i}$ factors is 
\begin{align*}
    \sum_{i \in [k]} |S_i \cap X| &= \sum_{i \in [k]}\left(t - \left\lceil\frac{\alpha}{2}\ha_i\right\rceil\right) = kt -  \sum_{i \in [k]} \left\lceil\frac{\alpha}{2}\ha_i\right\rceil.
\end{align*}
The size of the intersection of $X$ and all $U_{ij}$ factors is 
\begin{align}
\sum_{1 \leq i < j \leq k} |U_{ij} \cap X| &= \sum_{1 \leq i < j \leq k}\left(t -  \frac{1}{k-1} \left(\left\lceil \frac{\alpha}{2} \ha_i \right\rceil + \left\lceil \frac{\alpha}{2} \ha_j \right\rceil \right)\right) \label{use-in-last-uij-size} \\
&= \binom{k}{2}t - \frac{1}{k-1} \sum_{i \in [k]} (k-1) \left\lceil \frac{\alpha}{2} \ha_i \right\rceil \nonumber \\
&= \binom{k}{2}t - \sum_{i \in [k]} \left\lceil \frac{\alpha}{2} \ha_i \right\rceil. \nonumber
\end{align}
Recall that the size of the intersection of $X$ and each $T_{i}$ is $w_i$. So the size of the intersection of $X$ and all $U_{ij}$, $S_{i}$, and $T_{i}$ factors is
\begin{align*}
&~kt -  \sum_{i \in [k]} \left\lceil\frac{\alpha}{2}\ha_i\right\rceil + \binom{k}{2}t - \sum_{i \in [k]} \left\lceil \frac{\alpha}{2} \ha_i \right\rceil  +  \sum_{i \in [k]} w_i\\
=&~\left(k + \binom{k}{2}\right) t -  2 \sum_{i \in [k]} \left\lceil\frac{\alpha}{2}\ha_i\right\rceil +  \sum_{i \in [k]} w_i.
\end{align*}
When $\alpha \in (0, 1)$, this equals
\begin{align*}
&\left(k + \binom{k}{2}\right) t -  2 \sum_{i \in [k]} \left\lceil\frac{\alpha}{2}\ha_i\right\rceil +  \sum_{i \in [k]}  2\left( |\beta|+ \left\lceil \frac{\alpha}{2} \ha_i\right\rceil\right)\\
= &\left(k + \binom{k}{2}\right) t + 2k|\beta| = q_2
\end{align*}
and when $\alpha = 1$, this equals
\begin{align*}
&\left(k + \binom{k}{2}\right) t -  2 \sum_{i \in [k]} \left\lceil\frac{\alpha}{2}\ha_i\right\rceil +  \sum_{i \in [k]} 2 \left\lceil \frac{\alpha}{2} \ha_i\right\rceil\\
= & \left(k + \binom{k}{2}\right) t = \left(k + \binom{k}{2}\right) s = q_2
\end{align*}
In addition, apart from all $U_{ij}$, $S_{i}$, and $T_{i}$ factors, the size of the $X$ is
\begin{align*}
&\sum_{1 \leq i < j \leq k}(|A_{ij}| +|B_{ij}|) + \sum_{i \in [k]}\left( |A_{i}| + |B_{i}| + |D_{i}|\right)\\
=& \binom{k}{2}(\left\lfloor (1-\alpha)y \right\rfloor +\left\lceil 2\alpha^2(1-\alpha)x \right\rceil) +    k\left( \left\lceil \alpha^2(1-\alpha)x \right\rceil + 2\left\lfloor (1-\alpha)y \right\rfloor \right),
\end{align*}
which equals $q_1$ if $ 0 < \alpha < 1$, and equals $0 = q_1$ if $\alpha = 1$.
Overall, the size of $X$ is $q_1 + q_2 = q$ for any $\alpha \in (0,1]$.

Thirdly, we verify that $X$ is a $(\alpha,\beta)$-linear degree dominating set of $H$.  We need to demonstrate that every vertex $v \in V(H) \setminus X$ satisfies the inequality that $|N(v)\cap X| \geq \alpha |N(v)| + \beta$. We divide it into two cases, which are $\alpha =1$ and $\alpha \in (0, 1)$.

Assume $\alpha =1$. Recall that $\beta= - (nk)^{10000}$. In addition, the size of $S_i$ and $U_{ij}$ are less than $n^{240}$ based on Table \ref{table: size of si and uij}.
Thus, we need to prove that every $v \in V(H) \setminus X$ satisfies that $|N(v)| - |N(v)\cap X| \leq -\beta = (nk)^{10000}$.
Clearly, $N$, $K$, all $A_i$, $B_i$, $D_i$, $A_{ij}$, and $B_{ij}$ are empty sets, so we do not need to consider them.
For any vertex $v$ of $T_{i}$, we have $N(v) \leq |S_i| + \sum_{j\neq i} |U_{ij}|$, which is at most $k^2 n^{240} < -\beta$. For any vertex $v\in R_{i}$, we have $N(v) = |S_i| \leq  n^{240} < -\beta$. For any vertex $v\in R_{ij}$, we have $N(v) = |U_{ij}|  \leq  n^{240} < -\beta$.
Consider any vertex $v$ in some $U_{ij} - X$. Suppose $a$ and $b$ denote the number of vertices in $N(v)\cap (U_{ij} - X)$ and  $N(v)\cap U_{ij} \cap X$, respectively. Since the maximum degree of $U_{ij} - X$ is $\ha_i + \ha_j$, we have $a\leq \ha_i + \ha_j$. Furthermore, recall that $t=s$ in this case, so $N(v)$ equals
\begin{align*}
 & |R_{ij}| + |T_i| + |T_j| +a +b\\
=&~|\beta| - 2s + 2t + a + b\\
=&~|\beta| +a +b.
\end{align*}
In addition, we have that $|N(v)\cap X|$ equals
\begin{align*}
&~|T_i\cap X| + |T_j\cap X| + |N(v) \cap U_{ij}\cap X|\\
=&~w_i +  w_j + b\\
=&~\ha_i +  \ha_j + b.
\end{align*}
Overall, we have that
\begin{align*}
&~|N(v)| - |N(v)\cap X|\\
=&~|\beta| +a +b - (\ha_i +  \ha_j + b)\\
=&~|\beta| + (a- \ha_i -  \ha_j)\\ 
\leq &~|\beta| = -\beta.
\end{align*}
Consider any vertex $v$ of some $S_{i} - X$. Suppose $a$ and $b$ denote the number of vertices in $N(v)\cap (S_i - X)$ and  $N(v)\cap S_i  \cap X$, respectively. Additionally, $S_i - X$ has maximum degree $\ha_i$, so $a\leq \ha_i$. First, we have $N(v) = |R_{i}| + |T_i| +a +b = |\beta| +a +b$. In addition, we have that $|N(v)\cap X|$ equals $w_i + b = \ha_i + b$. Thus, we have $|N(v)| - |N(v)\cap X| = |\beta| + (a- \ha_i) \leq |\beta| = -\beta$.
This concludes the case $\alpha = 1$.

Assume $\alpha \in (0, 1)$. 
Recall that $a_1 > r$. In addition, the number of vertices in any $S_i$ or $U_{ij}$ is at most $(10a_1)^{1000}$ based on Table \ref{table: size of si and uij}. 
Consider any vertex $v$ of some $T_{i} - X$. Since all vertices of $D_{i}$ are in $X$, we have that $|N(v)\cap X| -\alpha |N(v)|$ is at least
\begin{align}
&~|D_i| - \alpha(|D_i|+|S_i| + \sum_{j \in [k]\setminus \{i\} } |U_{ij}|) \nonumber\\
\geq &~(1-\alpha) \left\lfloor (1-\alpha) y \right\rfloor - \alpha k (10a_1)^{1000} \label{eq:y-is-big}\\
\geq &~(1-\alpha) \left\lfloor (1-\alpha) (nka_1)^{10000} \right\rfloor - \alpha k (10a_1)^{1000}  \nonumber \\
> &~(1-\alpha)^2 (nka_1)^{10000}  - \alpha k (10a_1)^{1000} - 1 +\alpha \nonumber \\
> &~(1-\alpha)^2 (nka_1)^{1000}a_1  -  (nka_1)^{1000} - 1  \nonumber \\
> &~\left((1-\alpha)^2 r  -1\right) (nka_1)^{1000} - 1  \nonumber \\
= &~\left((1-\alpha)^2 \left\lceil\frac{10k(|\beta|+10)}{\alpha(1-\alpha)}\right\rceil^{10}  -1\right) (nka_1)^{1000} - 1  \nonumber \\
> &~(1-\alpha)^2 \left\lceil\frac{10k(|\beta|+10)}{\alpha(1-\alpha)}\right\rceil^{10}  -2 >  \beta  \nonumber,
\end{align}
which implies that $|N(v) \cap X| \geq \alpha |N(v)| + \beta$. 

Consider any vertex $v$ of some $R_i$. Since all vertices of $A_i$ are in $X$,  we have
\begin{align*}
|N(v)\cap X| -\alpha |N(v)| = &~|A_i| - \alpha(|A_i|+|S_i|)\\
\geq&~(1-\alpha) \left\lfloor (1-\alpha) y \right\rfloor - \alpha (10a_1)^{1000},
\end{align*}
which is clearly larger than formula  (\ref{eq:y-is-big}), thus lager than $\beta$.
Consider any vertex $v$ of some $R_{ij}$. Since all vertices of $A_{ij}$ are in $X$,  we have
\begin{align*}
|N(v)\cap X| -\alpha |N(v)| \geq &~|A_{ij}| - \alpha(|A_{ij}|+|U_{ij}|)\\
\geq &~ (1-\alpha) \left\lfloor (1-\alpha) y \right\rfloor - \alpha (10a_1)^{1000},
\end{align*}
which is clearly larger than formula  (\ref{eq:y-is-big}), thus lager than $\beta$.
Consider any vertex $v$ of $K$ and $N$. Since all neighbor vertices of $v$ are in $X$, we have
\begin{align*}
|N(v)\cap X| -\alpha |N(v)| = (1-\alpha)|N(v)|
\geq  (1-\alpha) \left\lfloor (1-\alpha)y \right\rfloor,
\end{align*}
which is clearly larger than formula  (\ref{eq:y-is-big}), thus lager than $\beta$.
Consider any vertex $v$ of some $S_{i} - X$. Suppose $a$ and $b$ denote the number of vertices of $N(v)\cap (S_i - X)$ and  $N(v)\cap X \cap S_i$, respectively. In addition, we know $S_i - X$ has maximum degree $\ha_i$.  Clearly, we have $0\leq a\leq \ha_i$ and $b\geq 0$. Thus, we have that the dominating coefficient $\lambda(v,H,X) = \frac{|N(v)\cap X| - \beta}{|N(v)|}$ of $v$ with respect to $(H,X)$ is
\begin{align*}
&~\frac{|B_i| + |T_i \cap X| + b - \beta}{|T_i| + |B_i| + |R_i| + a + b}\\
=&~\frac{\left\lceil\alpha^2(1-\alpha)x \right\rceil +  2|\beta| +  2\left\lceil \frac{1}{2}\alpha \ha_i\right\rceil + b - \beta}{\left\lfloor \alpha(1-\alpha)^2x \right\rfloor+ s + \left\lceil \alpha^2 (1- \alpha)x \right\rceil + \left\lfloor |\beta| - s - (1-\alpha)l \right\rfloor + a + b}\\
\geq &~\frac{\left\lceil\alpha^2(1-\alpha)x \right\rceil  + \alpha |\beta|+ \alpha \ha_i + \alpha b}{ \alpha(1-\alpha)^2x +  \left\lceil \alpha^2 (1- \alpha)x \right\rceil + |\beta|  + \ha_i + b}\\
\geq &~\frac{\alpha^2(1-\alpha)x  + \alpha |\beta|+ \alpha \ha_i + \alpha b}{ \alpha(1-\alpha)^2x +  \alpha^2 (1- \alpha)x  + |\beta|  + \ha_i + b}\\
= &~\frac{\alpha (\alpha (1- \alpha)x  + |\beta|  + \ha_i + b)}{ \alpha (1- \alpha)x  + |\beta|  + \ha_i + b}= \alpha,
\end{align*}
where we used the facts that $l=0$, and 
where the equality on the penultimate line is because $ \frac{\left\lceil c\right\rceil+d}{\left\lceil c\right\rceil+e} \geq \frac{ c+d}{ c+e} $ for any real number $e\geq d \geq 0$ and $c > 0$. 
Consider any vertex $v$ of some $U_{ij} -  X$. Suppose $a$ and $b$ denote the number of vertices of $N(v)\cap (U_{ij} - X)$ and  $N(v)\cap X \cap U_{ij}$, respectively. Additionally, we know that $U_{ij}- X$ has maximum degree $\ha_i + \ha_j$. Clearly, we have $0\leq a\leq \ha_i+\ha_j$ and $b\geq 0$. Thus, we have the dominating coefficient $\lambda(v,H,X)$ of $v$  with respect to $(H,X)$ is
\begin{align*}
&~\frac{|B_{ij}| + |T_i \cap X| + |T_j \cap X| + b -\beta}{|T_i| + |T_j| + |B_{ij}| + |R_{ij}| + a + b}\\
=&~\frac{\left\lceil 2\alpha^2(1-\alpha)x \right\rceil + 2|\beta| + 2\left\lceil \frac{1}{2}\alpha \ha_i\right\rceil +  2|\beta| + 2\left\lceil \frac{1}{2}\alpha \ha_j\right\rceil + b - \beta}{2t + \left\lceil 2\alpha^2 (1- \alpha)x \right\rceil + \left\lfloor |\beta| - 2s - 2(1-\alpha)l \right\rfloor + a + b}\\
\geq &~\frac{\left\lceil 2\alpha^2(1-\alpha)x \right\rceil + 3|\beta| + \alpha \ha_i+ \alpha \ha_j + b}{2\left\lfloor \alpha(1-\alpha)^2x \right\rfloor + \left\lceil 2\alpha^2 (1- \alpha)x \right\rceil + |\beta|  + a + b}\\
\geq &~\frac{\left\lceil 2\alpha^2(1-\alpha)x \right\rceil + \alpha |\beta| + \alpha \ha_i+ \alpha \ha_j + \alpha b}{2 \alpha(1-\alpha)^2x  + \left\lceil 2\alpha^2 (1- \alpha)x \right\rceil  + |\beta|  +  \ha_i+  \ha_j + b}\\
\geq &~\frac{ 2\alpha^2(1-\alpha)x + \alpha |\beta| + \alpha \ha_i+ \alpha \ha_j + \alpha b}{2 \alpha(1-\alpha)^2x  + 2\alpha^2 (1- \alpha)x  + |\beta|  +  \ha_i+  \ha_j + b}\\
= &~\frac{  \alpha (2 \alpha(1-\alpha) x  + |\beta|  +  \ha_i+  \ha_j + b)}{2 \alpha(1-\alpha) x  + |\beta|  +  \ha_i+  \ha_j + b} = \alpha,
\end{align*}
where, again, the equality on the penultimate line is because $ \frac{\left\lceil c\right\rceil+d}{\left\lceil c\right\rceil+e} \geq \frac{ c+d}{ c+e} $ for any real number $e\geq d \geq 0$ and $c > 0$. 

Overall, we proved that $X$ is a $(\alpha,\beta)$-linear degree dominating set of $H$, where the size of $X$ at most $q$. 
\end{proof}

We now consider the reverse direction, which will be demonstrated in the following lemma.

\begin{lemma}\label{lem:ld-second-dir}
Suppose that there is $X \subseteq V(H)$ with $|X| \leq q$ such that $|N(v)\cap X| \geq \alpha |N(v)| + \beta$ for every $v \in V(H) \setminus X$. Then $G$ contains a multicolored clique with $k$ vertices.
\end{lemma}
\begin{proof}
Assume $i,j\in [k]$. To alleviate notation slightly, for a factor $M$ of $H$, we will write $X(M) := X \cap V(M)$ for the vertices of $M$ added to $X$, and $\chi(M) := |X(M)|$ for the number of vertices of $M$ added to $X$. 
In the proof, $M$ will be one of the 
factors in $H$.

Define $V_2$ as the set of vertices consisting of the union of all the vertices of $T_i$, $S_i$, and $U_{ij}$ for all distinct $i, j\in [n]$. In addition, define $V_1 = V(H) \setminus V_2$.
Let us first consider the size of $T_i$ of $H$. If $\alpha \in \{0,1\}$, then $|T_i| = s = n^{10}$. If  $\alpha \in (0,1)$, then  $|T_i| = \left \lfloor \alpha (1-\alpha)^2 x \right \rfloor \leq x = (10a_1)^5$. Thus, we have 
\begin{align*}
|V_2| = &\sum_{i\in [k]} (|S_i| + |T_i|) + \sum_{1\leq i<j \leq k} |U_{ij}|.
\end{align*}
If $\alpha =1$, then, according to Table \ref{table: size of si and uij}, we have
\begin{align*}
|V_2|< k (n^{240} +  n^{10}) + k^2 n^{240} < n^{243}.
\end{align*}
If $\alpha \in (0,1)$, then, according to Table \ref{table: size of si and uij}, we have
\begin{align*}
|V_2|< k ((10a_1)^{1000} +  (10a_1)^{5}) + k^2 (10a_1)^{1000} < (10a_1)^{1001}.
\end{align*}
If $\alpha =0$, then, according to Table \ref{table: size of si and uij}, we have
\begin{align*}
|V_2|< k (n^{210} +  n^{10}) + k^2 n^{211} < n^{214}.
\end{align*}
The proof is divided into a series of claims.

\begin{claim}\label{cl-nk-modules}
If $\alpha \in (0,1)$, then all factors adjacent to $N$ or $K$ are in $X$, moreover, the number of vertices of all these factors is $q_1$.
\end{claim}
\begin{proof}
First note that, according to Table \ref{table:reduction-values-relationship}, the size of $X$ is at most $q  <  3k^2(nka_1)^{10000}$.  Also notice that $|\beta| < r < a_1$ (see Table~\ref{table:reduction-values}).
Consider any vertex $v$ which is adjacent to a vertex of $N$. Assume $v \not \in X$. Recall that $k\geq 100$. Then, we have a contradiction as follows.
\begin{align*}
|X| & \geq |N(v)\cap X| \geq \alpha |N(v)| + \beta \geq \alpha |N| - r\\  
&\geq \alpha \lfloor \alpha (1-\alpha) m \rfloor - (a_1 - 1)  \\
&> \alpha^2(1-\alpha) m - a_1  \\
&= \alpha^2(1-\alpha) (nka_1)^{20000} - a_1\\
&> \alpha^2(1-\alpha) a_1(nka_1)^{19999} - a_1\\
&> \alpha^2(1-\alpha) \left\lceil\frac{10k(|\beta|+10)}{\alpha(1-\alpha)}\right\rceil^{10} (nka_1)^{19999} - a_1\\
&> 2 (nka_1)^{19999} - a_1 > (nka_1)^{19999} > k^3 (nka_1)^{10000} \\
& > 3k^2(nka_1)^{10000} > q \geq |X|.
\end{align*}
Thus, all vertices adjacent to $N$ are in $X$. Additionally, the proof that all factors adjacent to $K$ are in $X$ is identical, since $|K| = |N|$. Finally, it is easy to calculate that the vertex number of all factors adjacent to $N$ or $K$ is $q_1$.
\end{proof}

\begin{claim}\label{cl:ld-module-2}
Consider factors $B_i$, $B_{ij}$, and $D_i$ for distinct $i,j \in [k]$. We have that $\chi(B_i) = |B_i|$, $\chi(B_{ij}) = |B_{ij}|$, and $\chi(D_i)$ $ = |D_i|$.
\end{claim}
\begin{proof}
Assume $\alpha = 0$. First, $B_i$ and $B_{ij}$ are empty sets. Secondly, for any vertex $v\in D_i$, we have $|N(v)| = |T_i| + |K|$, which equals
\begin{align*}
2s + l
< 2s + r < 3s = 3n^{10} < (nk)^{50} = \beta.
\end{align*}
Thus, $v$ must be in $X$ and $D_i \subseteq X$ (if $v \notin X$, it does not have enough neighbors, since it needs at least $\beta$ neighbors in $X$).
In addition, according to Claim \ref{cl-nk-modules}, $B_i$, $B_{ij}$, and $D_i$ are subsets of $X$ if $\alpha \in (0,1)$, moreover, $B_i$, $B_{ij}$, and $D_i$ are empty set if $\alpha =1$.
\end{proof}

\begin{claim}\label{cl:ld-module-n-k}
Consider factors $N$ and $K$. We may assume that 
\begin{enumerate}
    \item 
    $\chi(N) = |N|$ and  $\chi(K) = |K|$ if $\alpha = 0$;
    
    \item 
    $\chi(N) = 0$ and  $\chi(K) = 0$ if $\alpha \in (0,1]$.
\end{enumerate}
\end{claim}
\begin{proof}
\textit{Case 1: $\alpha = 0$.} Consider $A_i$ for some $i \in [k]$. Suppose that there exists a vertex of $R_i \cup K$ that is not in $X$. For any vertex $v\in A_i$, we would have $|N(v)\cap X| < |R_i| + |K| = |\beta|$. Thus, each vertex of $A_i$ is in $X$. Moreover, we know that every $D_i$ is a subset of $X$ according to Claim \ref{cl:ld-module-2}, which means that $|X|$ is at least
\begin{align*}
|A_i|+ \sum_{i\in [k]} |D_i| = ky +y =ky + (nka_1)^{10000}.
\end{align*}
However, we know that  $|X| = q$ is at most $ky + n^{214}$ according to Table \ref{table:reduction-values-relationship}, a contradiction. Thus, all vertices of $R_i \cup K$ are in $X$. The proof for that $R_{ij} \cup N$ is a subset of $X$ goes the same way.

\textit{Case 2: $\alpha \in (0, 1]$.} First, $N$ and  $K$ are empty sets if  $\alpha =1$. Secondly,  all factors adjacent to $N$ or $K$ are in $X$ according to Claim \ref{cl-nk-modules} if $\alpha \in (0,1)$, furthermore, the vertices in $N$ or $K$ are independent set. Thus, for any $v \in N \cup K$, we have $N(v) \subseteq X$. For any $v\in N \cup K$, we claim that $|N(v)\cap X| \geq \alpha |N(v)| + \beta$. Since $N(v) \subseteq X$, it is enough to argue that $|N(v)| \geq \frac{\beta}{1-\alpha}$. In addition, since $A_i \subseteq N(K)$ and $A_{ij} \subseteq N(N)$, we have  
\begin{align*}
|N(v)| &\geq \min \{ |A_i|, |A_{ij}|\} = \left\lfloor (1-\alpha) y \right\rfloor \\  
&>   (1-\alpha) y  - 1 = (1-\alpha) (nka_1)^{10000}  - 1\\
&>  (1-\alpha) a_1  - 1 > (1-\alpha) r  - 1 \\
&>  (1-\alpha) \left\lceil\frac{10k(|\beta|+10)}{\alpha(1-\alpha)}\right\rceil^{10}  - 1,
\end{align*}
which is clearly larger than $\frac{\beta}{1-\alpha}$.
\end{proof}
Note that, in this Claim \ref{cl:ld-module-n-k}, we have proved the fact that  $R_i \cup K \subseteq X$ and $R_{ij} \cup N \subseteq X$ for all distinct $i,j\in [k]$ if $\alpha = 0$. Moreover, this fact will be used in Claim \ref{cl:ld-module-3}, Claim \ref{cl:ld-module-4}, and Claim \ref{cl:ld-module-5}.

\begin{claim}\label{cl:ld-module-3}
Consider factors $A_i$ and $A_{ij}$ for all distinct $i, j\in [k]$. We may assume that 
\begin{enumerate}
    \item 
    $\chi(A_i) = 0$ and  $\chi(A_{ij}) = 0$ if $\alpha = 0$;
    
    \item 
    $\chi(A_i) = |A_i|$ and  $\chi(A_{ij}) = |A_{ij}|$ if $\alpha \in (0,1]$.
\end{enumerate}
\end{claim}
\begin{proof}
\textit{Case 1: $\alpha = 0$.} By Claim  \ref{cl:ld-module-n-k}, $R_i \cup K$ and  $R_{ij} \cup N$ are subsets of $X$ for all distinct $i,j\in [k]$. Thus, we may assume that $\chi(A_i) = 0$ and  $\chi(A_{ij}) = 0$ since $|R_i| + |K| = \beta$ and $|R_{ij}| + |N| = \beta$.

\textit{Case 2: $\alpha \in (0, 1]$.} By Claim \ref{cl-nk-modules},
every $A_i$ and $A_{ij}$ are subsets of $X$ if $\alpha \in (0,1)$, furthermore, every $A_i$ and $A_{ij}$ are empty sets if $\alpha = 1$. Thus, we may assume that $\chi(A_i) = |A_i|$ and  $\chi(A_{ij}) = |A_{ij}|$ if $\alpha \in (0,1]$.
\end{proof}

\begin{claim}\label{r-t are safe}
For any vertex $v$ of $T_i$, $R_i$, or $R_{ij}$ such that $v \notin X$, we have $|N(v) \cap X| \geq \alpha |N(v)| + \beta$.
\end{claim}
\begin{proof}
Consider a vertex $v$ of some $R_i$ or $R_{ij}$. If $\alpha =0$, then $R_i$ and $R_{ij}$ are subsets of $X$ based on the proof of Claim \ref{cl:ld-module-n-k}. So we do not need to consider this case. If $\alpha =1$, then, according to Table \ref{table:reduction-values} and Table \ref{table: size of si and uij}, $|N(v)| = |S_i| < n^{240} = -\beta$ for $v\in R_i$, and $|N(v)| = |U_{ij}| < n^{240} < -\beta$ for $v\in R_{ij}$.
Let us consider $\alpha \in (0, 1)$. Assume $v\in R_i$. According to Claim \ref{cl:ld-module-3}, we have $A_i \subseteq X$.
Even if $S_i \cap X$ is an empty set, the dominating coefficient of any vertex $v$ in $R_i$ with respect to $(H, X)$ is 
\begin{align*}
\lambda(v, H, X) &= \frac{|A_i| - \beta}{|A_i| + |S_i|} \geq \frac{\left\lfloor (1-\alpha)y \right\rfloor - (r - 1)}{\left\lfloor (1-\alpha)y \right\rfloor + (10a_1)^{1000}}\\
&>   \frac{ (1-\alpha)y  - r}{(1-\alpha)y + (10a_1)^{1000}} >   \frac{ (1-\alpha)y  - r - (10a_1)^{1000}}{(1-\alpha)y}\\
&>  1 - \frac{ 2(10a_1)^{1000}}{(1-\alpha)(nka_1)^{10000}}
>  1 -  \frac{1}{(1-\alpha)a_1}\\
&> 1 -  \frac{1}{(1-\alpha)r}  =  1 -  \frac{1}{(1-\alpha)\left\lceil\frac{10k(|\beta|+10)}{\alpha(1-\alpha)}\right\rceil^{10}}\\
&>  1 -  \frac{1}{(1-\alpha)\frac{1}{\alpha^{10}(1-\alpha)^{10}}}
>  1 -  (1-\alpha) = \alpha.
\end{align*}
The second line is because $\frac{b}{a}\geq \frac{b-c}{a-c}$ for any $0<c<b<a$. Suppose $v\in R_{ij}$. 
According to Claim \ref{cl:ld-module-3}, we have $A_{ij} \subseteq X$. Even if $U_{ij} \cap X$ is an empty set, the dominating coefficient of any vertex $v$ in $R_{ij}$ with respect to $(H, X)$ is 
\begin{align*}
\lambda(v, H, X) = \frac{|A_{ij}| - \beta}{|A_{ij}| + |U_{ij}|} > \frac{\left\lfloor (1-\alpha)y \right\rfloor - (r - 1)}{\left\lfloor (1-\alpha)y \right\rfloor + (10a_1)^{1000}},
\end{align*}
which is larger than $\alpha$ according to the inequality above.

Consider a vertex $v$ of $T_i$. $D_i$ is a subset of $X$ by Claim \ref{cl:ld-module-2}. So $|N(v) \cap X| \geq |D_i| = \left \lfloor (1-\alpha)y \right\rfloor$. If $\alpha =0$, then it is trivial that $|N(v) \cap X| = y > (nk)^{50} = \beta$. Assume $\alpha \in (0, 1)$. Even if, apart from $D_i$, the intersection of $X$ and the other neighbors of $v$ is empty, the dominating coefficient of $v$ with respect to $(H,X)$ is  
\begin{align*}
\lambda(v, H, X) =& \frac{|D_{ij}| - \beta}{|D_{ij}| + |S_{i}| + \sum_{j\in [k]\setminus \{i\}}|U_{ij}|} \\
> &\frac{\left\lfloor (1-\alpha)y \right\rfloor - (r - 1)}{\left\lfloor (1-\alpha)y \right\rfloor + k(10a_1)^{1000}},
\end{align*}
which is larger than $\alpha$ using a similar procedure as that used above. Now, let us consider $\alpha =1$. Clearly, we have 
\begin{align*}
|N(v)| = |S_{i}| + \sum_{j\in [k]\setminus \{i\}}|U_{ij}|
 <  kn^{240} < (nk)^{10000} = -\beta.
\end{align*}
Thus, we have $|N(v)| + \beta < 0 \leq |N(v)\cap X|.$
\end{proof}

\begin{claim}\label{cl:ld-module-4}
Consider factor $R_i$ for all $i$. We may assume that
\begin{enumerate}
    \item 
    $\chi(R_i) = |R_i|$ if $\alpha = 0$;
    
    \item 
    $\chi(R_i) = 0$ if $\alpha \in (0,1]$.
\end{enumerate}
\end{claim}
\begin{proof}
According to Claim \ref{r-t are safe}, we only need to consider the vertices outside $R_i$ that are affected by the intersection of $R_i$ and $X$.

\textit{Case 1:} that $\chi(R_i) = |R_i|$ is demonstrated in the proof of Claim \ref{cl:ld-module-n-k} if $\alpha = 0$.

\textit{Case 2:} first, let us consider $\alpha =1$, and recall that in this case $G - X$ has maximum degree at most $|\beta|$. If $\chi(S_i) < c(a_1)$, then $\Delta(S_i - X) \geq a_0$ by Definition~\ref{def:kd-degdeletion-bdd-121}. 
Let $v \in S_i - X$ be a vertex of maximum degree in the graph $S_i - X$.  In $H$, $v$ has $|\beta|$ neighbors in $R_i \cup T_i$, which means that at least $a_0 = q +1$ vertices of $R_i\cup T_i$ are added to $X$, a contradiction. Thus,  $\chi(S_i) \geq c(a_1)$. 
We may assume that the vertices of $S_i \cap X$ are chosen to minimize $\Delta(S_i - X)$, since the choice of vertices of $S_i$ to add to $X$ does not affect other factors, only the number.  Thus, we may assume that $\Delta(S_i - X) \leq a_1$.

Assume that there exists some $v \in X(R_i)$. Suppose that there exists at least one vertex $u \in T_i - X$. Recall that we also call $X$ the deletion part of $H$. Consider $X' = (X \setminus \{v\}) \cup \{u\}$.  The subgraph $H - X'$ can be seen as taking $H - X$, deleting $u$, and reinserting $v$.
Deleting $u$ decreases the degree of the vertices in the neighbors of $T_i$ by $1$, including the $S_i$ vertices.  Then by reinserting $v$, we only increase the degrees of the $S_i$ vertices by $1$.  Thus, the maximum degree of $H - X'$ is not more than the maximum degree of $H - X$ and we may use $X'$ instead. We may repeat this argument until no element of $R_i$ is in $X$ (in which case we are done), or until $T_i - X$ is empty, i.e. every element of $T_i$ is in $X$.
So suppose that we reach a point where all vertices of $T_i$ are in $X$. Since $\Delta(S_i - X) \leq a_1$, if we assume that $X(R_i)$ is empty, each remaining vertex of $S_i$ in $H - X$ has maximum degree at most $|\beta| - s + a_1 < |\beta|$.  Thus, we may assume that $\chi(R_i) = 0$.


Secondly, let us consider $\alpha \in (0,1)$. 
If $\chi(S_i) < c(a_1)$, then $S_i - X$ has a vertex $v$ with degree $a_0$ by Definition~\ref{def:kd-degdeletion-bdd-121}. Consider the degree of $v$ in $H$. The size of $X\cap N(v)$ is at least
\begin{align*}
\alpha |N(v)| + \beta  
\geq \alpha( |R_i| + |T_i| + |B_i| + a_0) + \beta   
> \alpha a_0 + \beta
\end{align*}
Furthermore, based on Claim \ref{cl-nk-modules}, we have $X\cap (H - N(v)) \geq q_1 - |B_i|$. Thus, the size of $X$ is larger than
\begin{align}
&~\alpha a_0 + \beta + q_1 - |B_i|
=  \alpha x^{10} + \beta +  q_1 -\left \lceil \alpha^2 (1-\alpha)x \right \rceil  \nonumber\\
> &~ q_1 + \alpha x^{10} + \beta - 2\alpha x -1  
=   q_1 + \alpha x (x^9 -2) + \beta  -1  \label{use-in-uij}\\
= &~ q_1 + \alpha (10a_1)^5 ( (10a_1)^{45} -2) + \beta  -1  
>    q_1 + \alpha ((10a_1)^{44}r-2)r + \beta  -1  \nonumber\\
= &~   q_1 + \alpha ((10a_1)^{44}r-2)\left\lceil\frac{10k(|\beta|+10)}{\alpha(1-\alpha)}\right\rceil^{10} + \beta  -1  \nonumber\\
>  &~ q_1 + (10a_1)^{44}r + \beta  -3  = q_1 + (10a_1)^{44}\left\lceil\frac{10k(|\beta|+10)}{\alpha(1-\alpha)}\right\rceil^{10} + \beta  -3\nonumber\\
> &~ q_1 +  (10a_1)^{44}  
>   q_1 + 2k^2(10a_1)^{5} > q_1 + q_2 = q, \nonumber
\end{align}
which is a contradiction.  Thus,  $\chi(S_i) \geq c(a_1)$. 
According to Lemma \ref{the detail of degree deletion graph}, $S_i$ contains $c(a_1) -c(a_0) = c(a_1)$ stars with $a_0$ leaves. Based on Lemma \ref{greedy lemma}, we may assume that the internal vertices of all $c(a_1)$ stars with $a_0$ leaves are added to $X$. Thus, $\Delta(S_i - X) \leq a_1$ and the largest star in $S_i - X$ contains at most $a_1$ leaves. 

Assume that there exists some $v \in X(R_i)$. Suppose that at least one vertex $u$ of $T_i$ is \emph{not} deleted by $X$, i.e. there exists $u \in T_i - X$.  Consider $X' = (X \setminus \{v\}) \cup \{u\}$. Consider every neighbor factor of $T_i$ that is not a subset of $X$. It is not hard to verify that $\Lambda(V(S_i -X),H,X) = \Lambda(V(S_i - X'),H,X)$, and that $\Lambda(V(U_{ij} -X),H,X) \leq  \Lambda(V(U_{ij} - X'),H,X)$ for every $j\in [k] \setminus
\{i\}$. Furthermore, $S_i$ is the only neighbor factor of $R_i$ that is not a subset of $X$. Thus, we may use $X'$ instead.
So suppose instead that all vertices of $T_i$ are deleted. Since the largest star in $S_i - X$ contains at most $a_1$ leaves, we have that, even if $\chi(R_i) = 0$, $\Lambda(V(S_i -X),H,X)$ is at least
\begin{align}
&~\frac{|T_i| + \chi(B_i) + \chi(R_i) - \beta}{|T_i| + |B_i| + |R_i| + a_1} \nonumber\\
= & ~\frac{\left\lfloor \alpha (1-\alpha)^2
x \right\rfloor +\left\lceil\alpha^2(1-\alpha)x \right\rceil - \beta}{\left\lfloor \alpha (1-\alpha)^2
x \right\rfloor  +\left\lceil\alpha^2(1-\alpha)x \right\rceil + |\beta| + a_1} \nonumber\\
> &~\frac{ \alpha(1-\alpha)x - 1 - |\beta|}{\alpha(1-\alpha)x + 1 + |\beta| + a_1} \nonumber\\
> &~\frac{ \alpha(1-\alpha)x - 2 - 2|\beta| - a_1}{\alpha(1-\alpha)x } \nonumber\\
> &~1 - \frac{  2a_1}{\alpha(1-\alpha)x } =  1 - \frac{  2a_1}{\alpha(1-\alpha)(10a_1)^5} \label{use-in-uij-eqal}\\
> &~1 - \frac{ 1}{\alpha(1-\alpha)r } =  1 - \frac{1}{\alpha(1-\alpha) \left\lceil\frac{10k(|\beta|+10)}{\alpha(1-\alpha)}\right\rceil^{10}} \nonumber \\
> &~1 - (1-\alpha) = \alpha. \nonumber
\end{align}
The fifth line is becasue $a_1 > r > 2|\beta|+2$. Thus, we may assume  $\chi(R_i) = 0$.
\end{proof}

\begin{claim}\label{cl:ld-module-5}
Consider factor $R_{ij}$ for all $i\neq j$. We may assume that 
\begin{enumerate}
    \item 
    $\chi(R_{ij}) = |R_{ij}|$ if $\alpha = 0$;
    
    \item 
    $\chi(R_{ij}) = 0$ if $\alpha \in (0,1]$.
\end{enumerate}
\end{claim}
\begin{proof}
\textit{Case 1:} that $\chi(R_{ij}) = |R_{ij}|$ is demonstrated in the proof Claim \ref{cl:ld-module-n-k} if $\alpha = 0$.

\textit{Case 2:} first, let us consider $\alpha =1$. The idea is the same as that of the proof for Claim \ref{cl:ld-module-4}.
If $\chi(U_{ij}) < c_{ij}(\hbar_{ij})$, then by Definition~\ref{def:kd-degdeletion-bdd-121}, we have $\Delta(U_{ij} - X) = a_0$.  This implies $a_0$ deletions in $R_{ij}\cup T_i \cup T_j$. Then we have $|X| \geq a_0 = q +1$, a contradiction. Thus, $\chi(U_{ij}) \geq c_{ij}(\hbar_{ij})$ and we may assume that $\Delta(U_{ij} - X) \leq \hbar_{ij}  < 2s$ (recall that $s > a_1$).

If there exists some $u \in (V(T_i) \cup V(T_j)) \setminus X$ and some $v \in X(R_{ij})$, then $X' = (X \setminus \{v\}) \cup \{u\}$ does not alter the degrees in $V(U_{ij}) \setminus X$ and can only reduce the degrees of other relevant factors.
So assume that all elements of $T_i$ and $T_j$ are deleted.  
We know $\Delta(U_{ij} - X) < 2s$, thus, in $H - X$, the vertices of $V(U_{ij}) \setminus X$ have maximum degree at most $\hbar_{ij} + d - 2s < d$ even if $X(R_{ij})$ is empty. Thus, we may assume that $\chi(R_{ij}) = 0$.

Secondly, let us consider $\alpha \in (0,1)$. 
If $\chi(U_{ij}) < c_{ij}(\hbar_{ij})$, then $U_{ij} - X$ has a vertex $v$ with degree $a_0$ by Definition~\ref{def:kd-degdeletion-bdd-121}. Consider the degree of $v$ in $H$. The size of $X\cap N(v)$ is at least
\begin{align*}
\alpha |N(v)| + \beta  
\geq \alpha( |R_{ij}| + |T_i| + |T_j| + |B_{ij}| + a_0) + \beta   
> \alpha a_0 + \beta
\end{align*}
Furthermore, based on Claim \ref{cl-nk-modules}, we have $X\cap (H - N(v)) \geq q_1 - |B_{ij}|$. Thus, the size of $X$ is larger than
\begin{align*}
&~\alpha a_0 + \beta + q_1 - |B_{ij}|
=  \alpha x^{10} + \beta +  q_1 -\left \lceil 2 \alpha^2 (1-\alpha)x \right \rceil \\
> &~ q_1 + \alpha x^{10} + \beta - 2\alpha x -1  
=  q_1 + \alpha x (x^9 -2) + \beta  -1, 
\end{align*}
which is larger than $q$ based on formula (\ref{use-in-uij}), a contradiction.  Thus,  $\chi(U_{ij}) \geq c_{ij}(\hbar_{ij})$. 
According to Lemma \ref{the detail of degree deletion graph}, $U_{ij}$ contains $c_{ij}(\hbar_{ij}) - c_{ij}(a_{0}) = c_{ij}(\hbar_{ij})$ stars with $a_0$ leaves. Based on Lemma \ref{greedy lemma}, we may assume that the internal vertices of all $c_{ij}(\hbar_{ij})$ stars with $a_0$ leaves are added to $X$. Thus, $\Delta(U_{ij} - X) \leq \hbar_{ij}$ and the largest star in $U_{ij} - X$ contains at most $\hbar_{ij}$ leaves. 

Assume that there exists some $v \in X(R_{ij})$. Suppose that at least one vertex $u$ of $T_i \cup T_j$ is \emph{not} deleted by $X$, i.e. there exists $u \in V(T_i\cup T_j) \setminus X$.  Consider $X' = (X \setminus \{v\}) \cup \{u\}$, and every neighbor factor $F$ of $T_i$ or $T_j$ that is not a subset of $X$. (Clearly, it includes $U_{ij}$.) It is not hard to verify that $\Lambda(V(F -X), H, X) \leq  \Lambda(V(F - X'), H, X)$. Furthermore, $U_{ij}$ is the only neighbor factor of $R_{ij}$ that is not a subset of $X$. Thus, we may use $X'$ instead.
So suppose instead that all vertices of $T_i\cup T_j$ are deleted. Since the largest star in $U_{ij} - X$ contains at most $\hbar_{ij}$ leaves. We have that, even if $\chi(R_{ij}) = 0$, $\Lambda(V(U_{ij} -X), H, X)$ is at least
\begin{align*}
&~\frac{|T_i| + |T_j|+ \chi(B_{ij}) + \chi(R_{ij}) - \beta}{|T_i| + |T_j|+ |B_{ij}| + |R_{ij}| + \hbar_{ij}}\\
= &~\frac{2\left\lfloor \alpha (1-\alpha)^2
x \right\rfloor +\left\lceil2\alpha^2(1-\alpha)x \right\rceil - \beta}{2\left\lfloor \alpha (1-\alpha)^2
x \right\rfloor  +\left\lceil2\alpha^2(1-\alpha)x \right\rceil + |\beta| + \hbar_{ij}}\\
> &~\frac{ 2\alpha(1-\alpha)x - 2 - |\beta|}{2\alpha(1-\alpha)x + 1 + |\beta| + 2a_1}\\
> &~\frac{ 2\alpha(1-\alpha)x - 3 - 2|\beta| - 2a_1}{2\alpha(1-\alpha)x }\\
> &~1 - \frac{  4a_1}{2\alpha(1-\alpha)x } =  1 - \frac{  2a_1}{\alpha(1-\alpha)(10a_1)^5},
\end{align*}
which is larger than $\alpha$ based on formula (\ref{use-in-uij-eqal}). The fifth line is because $3+2|\beta| < r < a_1 < 2a_1$. Thus, we may assume that $\chi(R_{ij}) = 0$.
\end{proof}

Thus in all cases, we may assume that the size of the intersection of $X$ and $V_1$ is $q_1$. This means that the size of the intersection of $X$ and $V_2$ is at most $q - q_1 = q_2$. In the next, we consider the vertices of $V_2$.

\begin{claim}\label{cl:ld-module-66}
If $\alpha =0$, then $\chi(S_i)\geq p$  for every $i\in [k]$, and $\chi(U_{ij}) \geq p_{ij}+p'_{ij} \geq p'$ for all distinct $i,j \in [k]$.
\end{claim}
\begin{proof}
Let $\alpha = 0$. Consider the neighbor modules $R_i, T_i$ of $S_i$. We have $|R_i| = |\beta|  - s -l$ and $|T_i| = s$. Thus, for any vertex $v \in S_i$, the number of neighbors of $v$ outside $S_i$ is $|\beta|  - s -l + s = |\beta|  -l$. In addition, according to the definition, $S_i$ has $p$ vertices of degree less than $l$ in $S_i$. Thus, all the $p$ vertices must be in $X$, otherwise, there exists some vertex $u\in S_i - X$ such that $|N(u)\cap X| \leq |N(u)| < |\beta| -l + l = \beta$, a contradiction. Hence,  $\chi(S_i)\geq p$.


Now, consider the neighbor modules $R_{ij}, T_i, T_j$ of $U_{ij}$. We have $|R_{ij}| = |\beta| - 2s - 2l$ and $|T_i| = |T_j| = s$. Thus, for any vertex $v\in U_{ij}$, the number of adjacent vertices of $v$ outside $U_{ij}$ is at most $|R_{ij}| + |T_i| + |T_j| = |\beta| -2l$. In addition, according to the definition, $U_{ij}$ has $p'_{ij}$ vertices with degree zero and $p_{ij}$ vertices of degree less than $l$. Therefore, the $p_{ij}+p'_{ij}$ vertices must be in $X$, otherwise, there exists some vertex $u\in U_{ij} - X$ such that $|N(u)\cap X| \leq |N(u)| < |\beta| -2l + l = |\beta| - l \leq \beta$, a contradiction. Recall that $p'_{ij}=  p'  - p_{ij} + \frac{1}{2(k-1)}(l_{ij} - 2a_n)$.
Since $l_{ij} \geq 2a_n$, we have
\begin{align*}
p_{ij}+p'_{ij} = p_{ij} +  p'  - p_{ij} + \frac{1}{2(k-1)}(l_{ij} - 2a_n) \geq p'.
\end{align*}
Hence, we have $\chi(U_{ij}) \geq p_{ij}+p'_{ij} \geq p'$.
\end{proof}

Recall that the minimum degree of $S_i - X$ in $X$ is denoted by $\delta (S_i - X, X)$, and that the maximum degree of $S_i - X$ is denoted by $\Delta(S_i - X)$.

\begin{claim}\label{cl:ld-module-6}
Consider factor $S_i$ for all $i$. We may assume that, 
\begin{enumerate}
    \item 
     if $\alpha = 0$, then $\delta (S_i - X, X)$ is an element of $I$, and $\chi(S_i) = p + c(\delta (S_i$ $ - X, X) )$;

     \item    
     if $\alpha \in (0,1]$, then $\Delta(S_i - X)$ is an element of $I$, and $\chi(S_i)$ $ = c(\Delta(S_i - X))$, moreover, $S_i \cap X$ consists of the first $\chi(S_i)$ largest degree internal vertices of the stars graph $S_i$ if $\alpha \in (0,1)$.
\end{enumerate}
\end{claim}
\begin{proof}
\textit{Case 1:} since $\alpha =0$, $S_i$ is a $(a_0, I, c)$-degree retention graph above $(p,l)$, where $c(a_i) = \frac{1}{2}(a_i - a_n)$ and $c(a_0)$ $= ks$. 
Since only the number of the vertices added to $X$ in $S_i$ affects the vertices outside of $S_i$, for any fixed $\chi(S_i)$, we may assume that the deletion of $\chi(S_i)$ vertices maximize $\delta (S_i - X, X)$. Based on Claim \ref{cl:ld-module-66}, $\chi(S_i)\geq p$. Moreover, since $p+ c(a_n) = p$, we have $\delta (S_i - X, X)\geq a_n$ according to the definition of $S_i$ and Definition~\ref{def:kd-degdeletion}.
Assume that there exists some $i \in [k]$ such that $\chi(S_{i}) \geq p + c(a_0) = p +ks$. We know  $\chi(U_{ij})\geq p'$ according to Claim \ref{cl:ld-module-66}. Therefore, we have $q_2 = |X \cap V_2|$, which equals 
\begin{align*}
&~\sum_{i\in [k]} (\chi(S_i) + \chi(T_i)) + \sum_{1\leq i<j \leq k} \chi(U_{ij})\\
\geq &~ \sum_{i\in [k]} \chi(S_i) + \sum_{1\leq i<j \leq k} \chi(U_{ij})\\
\geq &~  (k-1)p + (p + ks) + \binom{k}{2} p'\\
= &~  kp + \binom{k}{2} p' + ks \\
> &~  kp + \binom{k}{2}p'  + k (s + l - a_n) = q_2,
\end{align*}
a contradiction. The last line is because  $l <  r \leq a_n$. Hence, for every $i \in [k]$, we have $\chi(S_{i}) < p + c(a_0)$. This means that $\delta (S_i - X, X)\leq a_1$ according to the definition of $S_i$ and Definition~\ref{def:kd-degdeletion}. 
Now, we have that $a_n \leq \delta(S_i - X, X) \leq a_1$, and that $p + c(a_n) \leq \chi(S_{i}) < p + c(a_0)$.

Let $j \in [n]$ be the minimum index such that $p + c(a_j) \leq \chi(S_i)$. Then, we have $\chi(S_i) < p+ c(a_{j-1})$, which implies that  $\delta(S_i - X, X) \leq a_{j}$ based on Definition~\ref{def:kd-degdeletion}. Moreover, it is possible to make $\delta(S_i - X, X) \geq a_j$ by adding $\chi(S_i)$ vertices to $X$ according to Definition~\ref{def:kd-degdeletion} and Observation \ref{obervation-for-ren-table}. Thus, we have $\delta(S_i - X, X) \in I$. 
Now assume that $\chi(S_i) = p+  c(a_j) + h$ for some $h > 0$.  
Consider $X'$ obtained from $X$, but adding exactly $p+c(a_j)$ vertices to $X$ from $S_i$ instead. Then $\delta(S_i - X', X')$ can still be $a_j$. Only the vertices in $T_i$ and $R_i$ might be affected by this change. However, $R_i \subseteq X'$ based on Claim \ref{cl:ld-module-4}, and $|N(v)\cap X'| \geq |D_i| > \beta$ for any $v\in T_i$ according to Claim \ref{cl:ld-module-2}. So this change does not create any problems.  Hence, we may assume that $\delta(S_i - X, X) = a_j$ and $\chi(S_i) = p+c(a_j)$.

\textit{Case 2:} first consider $0< \alpha <1$. $S_i$ is a $(a_0, I \cup \{a_n - 1\}, c)$-degree deletion star graph. According to Lemma \ref{the detail of degree deletion graph}, we have that, apart from the stars with leaves less than $a_n$, $S_i$ consists of exact $c(a_{j}) - c(a_{j-1})$ stars with $a_{j-1}$ leaves for all  $j\in [n]$ as well as exact $c(a_{n}-1) - c(a_{n})$ stars with $a_{n}$ leaves. Clearly, the number of stars with at least $a_{n}$ leaves equals
\begin{align*}
&~ c(a_{n}-1) - c(a_{n}) + \sum_{j \in [n]} (c(a_{j}) - c(a_{j-1}))\\
=&~ c(a_{n}-1) - c(a_{0}) = c(a_{n}-1) = a_{0}.
\end{align*}
Since $a_{0} > q_2$ according to Table \ref{table:reduction-values-relationship}, the number of stars with at least $a_{n}$ leaves is larger than $\chi(S_i)$.  It means that $\chi(S_i) < c(a_{n}-1)$.
Furthermore, since only the number of the vertices added to $X$ affects the dominating coefficients of the vertices outside of $S_i$ with respect to $(H,X)$, for any fixed $\chi(S_i)$, we may assume that the deletion of $\chi(S_i)$ vertices maximize the value of the dominating coefficient of $V(S_i)\setminus X$ with respect to $(H,X)$. 
Thus, based on Lemma \ref{greedy lemma}, we may assume that $X(S_i)$ consists of the first $\chi(S_i)$ largest degree internal vertices of the stars graph $S_i$.
Hence, $\Delta(S_i - X)\in I \cup \{a_0\}$. 
Moreover, according to case 2 of the proof of Claim \ref{cl:ld-module-4}, we have $\Delta(S_i - X) \leq a_1$, which also means that $\chi(S_i) \geq c(a_{1})$ based on the definition of $S_i$ and Definition \ref{def:kd-degdeletion-bdd-121}. 
Overall, we have $\Delta(S_i - X) \in I$ and $c(a_{1})  \leq \chi(S_i) < c(a_{n}-1)$.

Let $j \in [n]$ be the maximum index such that $c(a_j) \leq \chi(S_i)$. Assume $\chi(S_i) =  c(a_j) + h$ for some $h >0$. Consider $X'$ obtained from $X$, but adding the largest $c(a_j)$ internal vertices from $S_i$ instead.  Then $S_i - X'$ still have maximum degree $a_j$, and thus $\Lambda(S_i - X, H,X) = \Lambda(S_i - X', H, X')$.
In addition, according to the proof of Lemma \ref{r-t are safe}, the dominating coefficient of $T_i$ and $R_i$ with respect to $(H, X)$ is always larger than $\alpha$ no matter what is the value of $\chi(S_i)$.
Hence, we may assume that $S_i - X$ has degree $a_j \in I$ and $\chi(S_i) = c(a_j)$.

Secondly, consider $\alpha  = 1$. The degree of every vertex of $H - X$ is at most $-\beta = (nk)^{10000}$.
As before, we may assume that the $\chi(S_i)$ deletions in $S_i$ minimize $\Delta(S_i - X)$.  
According to case 2 in the proof of Claim~\ref{cl:ld-module-4}, we have $\chi(S_i) \geq c(a_1)$, and $\Delta(S_i - X) \leq a_1$.  
Moreover, $S_i - X$ cannot have maximum degree $a_n - 1$ or less, since this requires that $\chi(S_i)$ is at least $c(a_n - 1) = a_0 = q+1$, a contradiction.
Hence, we have $a_n \leq \Delta(S_i - X) \leq a_1$ and $c(a_1) \leq \chi(S_i) < c(a_n -1)$.

Let $j \in [n]$ be the maximum index such that $c(a_j) \leq \chi(S_i)$.  
Note that $j$ is well-defined since $\chi(S_i) \geq c(a_1)$.
If $j = n$, then we made at least $c(a_n)$ deletions and we have $\Delta(S_i - X) = a_n$.
If instead $j < n$, we have $\chi(S_i) < c(a_{j+1})$ and, by Defintion~\ref{def:kd-degdeletion-bdd-121}, we have $\Delta(S_i - X) \geq a_j$. Furthermore, according to Defintion~\ref{def:kd-degdeletion-bdd-121} and Observation \ref{obervation-for-del-table}, it is possible to make $\Delta(S_i - X) \leq a_j$ by adding $\chi(S_i)$ vertices to $X$. Thus, we may assume that $\Delta(S_i - X) = a_j$.
In either case, $\Delta(S_i - X) \in I$, as desired.

Now assume that $\chi(S_i) = c(a_j) + h$ for some $h > 0$.  
Consider $X'$ obtained from $X$, but adding exactly $c(a_j)$ vertices from $S_i$ instead.  Then $S_i - X'$ can still have maximum degree $a_j$. Although the degrees of vertices of $T_i$ and $R_i$ in $H - X'$ have increased by $h$, any vertex of $R_i$ and $T_i$ has a total degree in $H$ at most $|V_2| < n^{243}$, which is much less than $- \beta =(nk)^{10000}$, so this increase cannot create any problem.  Hence, we may assume that $S_i - X$ has degree $a_j$ and $\chi(S_i) = c(a_j)$.
\end{proof}

Recall that integer set $I_{ij}$ consists of all $a + b$ such that $a$, $b \in$ $I$ and edge  $(f_{i}^{-1}(a), f_{j}^{-1}(b)) \in E(G)$, and that $\ell_{ij}$ and $\hbar_{ij}$ are the smallest element and the largest element of $I_{ij}$, respectively.

\begin{claim}\label{cl:ld-module-7}
Consider factor $U_{ij}$ for all distinct $ i, j \in [k]$. We may assume that 
\begin{enumerate}
    \item 
     if $\alpha = 0$, then  $\delta (U_{ij} - X, X)$ is an element of $I_{ij}$, and $\chi(U_{ij}) = p'_{ij} + p_{ij}  + c_{ij}(\delta (U_{ij} - X, X) )$;

     \item 
     if $\alpha \in (0,1]$, then $\Delta(U_{ij} - X)$ is an element of $I_{ij}$, and  $\chi(U_{ij}) = c_{ij}(\Delta(U_{ij} - X))$, moreover, $U_{ij} \cap X$ consists of the first $\chi(U_{ij})$ largest degree internal vertices of the stars graph $U_{ij}$ if  $\alpha \in (0,1)$.
\end{enumerate}
\end{claim}
\begin{proof}
For any two vertex sets $V^i$ and $V^j$ of graph $G$, suppose there are $m_{ij}$ pairs of symmetry edges between them. Then, in $H$, each $I_{ij}$ has exactly $m_{ij}$ elements. 
Suppose  $I_{ij} = \{a_1^{ij}, \ldots, a_{m_{ij}}^{ij}\}$, where $\hbar_{ij} = a_1^{ij} > \ldots  >  a_{m_{ij}}^{ij} = \ell_{ij}$. Assume $\jmath \in [m_{ij}]$.

\textit{Case 1:} since $\alpha =0$, $U_{ij}$ consists of an edgeless graph with $p'_{ij}$ vertices together with  a $(2a_1 + 1, I_{ij}, c_{ij})$-degree retention graph above $(p_{ij},l)$, where $c_{ij}(a+b) = \frac{1}{2(k-1)}(a + b - l_{ij})$, and $p'_{ij}=  p'  - p_{ij} + \frac{1}{2(k-1)}(l_{ij} - 2a_n)$. In addition, we have $c_{ij}(2a_1+1)= ks$. Assume that $a_0^{ij} = 2a_1 + 1$.
Since only the number of the vertices added to $X$ affects the vertices outside of $U_{ij}$, for any $\chi(U_{ij})$, we may assume that the selection of $\chi(U_{ij})$ vertices maximize $\delta (U_{ij} - X, X)$. Based on the proof of Claim \ref{cl:ld-module-66}, we have $\chi(U_{ij}) \geq  p_{ij} + p'_{ij}$ and the $p'_{ij}$ vertices of the edgeless graph in $U_{ij}$ must be in $X$. Moreover, since $p'_{ij}+p_{ij}+ c_{ij}(\ell_{ij}) = p_{ij} + p'_{ij}$, we have $\delta (U_{ij} - X, X)\geq \ell_{ij}$ according to the definition of $U_{ij}$ and Definition \ref{def:kd-degdeletion}.
Assume there exists some $U_{ij}$ such that $\chi(U_{ij}) \geq p'_{ij} + p_{ij} + c_{ij}(2a_1+1)$. In addition, we know $c_{ij}(2a_1+1)= ks$ and $\chi(S_{i})\geq p$ according to Claim \ref{cl:ld-module-66}. Therefore, $q_2 = |X \cap V_2|$ is equal to 
\begin{align*}
&~ \sum_{i\in [k]} (\chi(S_{i}) + \chi(T_i)) + \sum_{1\leq i<j \leq k} \chi(U_{ij})\\
\geq &~ \sum_{i\in [k]} \chi(S_{i}) + \sum_{1\leq i<j \leq k} \chi(U_{ij})\\
\geq &~  kp + \binom{k}{2} (p'_{ij} + p_{ij} ) + k s\\
\geq &~  kp + \binom{k}{2} p' + k s \\
> &~ kp + \binom{k}{2}p'  + k (s + l -   a_n) = q_2,
\end{align*}
a contradiction. The penultimate line is because $p'_{ij} + p_{ij} \geq p'$ based on  Claim \ref{cl:ld-module-66}. The last line is because $l < r \leq  a_n$. Hence, for all $U_{ij}$, we have $\chi(U_{ij}) < p'_{ij} + p_{ij}  + c_{ij}(2a_1+1)$. This means that $\delta (U_{ij} - X, X)\leq \hbar_{ij}$ according to the definition of $U_{ij}$ and Definition \ref{def:kd-degdeletion}.
Now, we have that $\ell_{ij} \leq \delta(U_{ij} - X, X) \leq \hbar_{ij}$, and that $p'_{ij} + p_{ij}  + c_{ij}(\ell_{ij}) \leq \chi(U_{ij}) < p'_{ij} + p_{ij}  + c_{ij}(2a_1+1)$.

Recall that the $p'_{ij}$ vertices of the edgeless graph in $U_{ij}$ must be in $X(U_{ij})$, so the number of the vertex deletion of the $(2a_1 + 1, I_{ij}, c_{ij})$-degree retention graph above $(p_{ij},l)$ is $\chi(U_{ij}) - p'_{ij}$. 
Let $\jmath \in [m_{ij}]$ be the minimum index such that $ p'_{ij} + p_{ij}  + c_{ij}(a_{\jmath}^{ij}) \leq \chi(U_{ij})$. Then, we have $\chi(U_{ij}) <  p'_{ij} + p_{ij} + c_{ij}(a_{\jmath-1}^{ij})$, which implies that  $\delta(U_{ij} - X, X) \leq a_{\jmath}^{ij}$ based on the definition of $U_{ij}$ and Definition~\ref{def:kd-degdeletion}. Moreover, it is possible to make $\delta(U_{ij} - X, X) \geq a_{\jmath}^{ij}$ by adding $\chi(U_{ij})$ vertices to $X$ according to Definition~\ref{def:kd-degdeletion} and Observation \ref{obervation-for-ren-table}. Thus, we have $\delta(U_{ij} - X, X)=a_{\jmath}^{ij}$, which is an element of $I_{ij}$.
Now assume that $\chi(U_{ij}) = p'_{ij} + p_{ij}  + c_{ij}(a_{\jmath}^{ij}) + h$ for some $h > 0$.  
Consider $X'$ obtained from $X$, but adding exactly $p'_{ij} + p_{ij}  +c_{ij}(a_{\jmath}^{ij})$ vertices to $X$ from $U_{ij}$ instead. Clearly, $\delta(U_{ij} - X', X')$ can still be $a_{\jmath}^{ij}$, and only the vertices in $T_i$,  $T_j$, and $R_{ij}$ might be affected by this change. However, $R_{ij} \subseteq X'$ based on Claim \ref{cl:ld-module-5}, and $|N(v)\cap X'| \geq \min \{|D_i|, |D_j|\} = y > \beta$ for any $v\in T_i \cup T_j$ according to Claim \ref{cl:ld-module-2}. So this change does not create any problems.  Hence, we may assume $\delta(U_{ij} - X, X)$ is $a_{\jmath}^{ij}$ and $\chi(U_{ij}) =p'_{ij} + p_{ij}  + c_{ij}(a_{\jmath}^{ij})$.

\textit{Case 2:} first consider $0< \alpha <1$. $U_{ij}$ is a ${(a_0, I_{ij} \cup \{\ell_{ij} - 1\}, c_{ij})}$-degree deletion star graph. Assume that $a_0^{ij} = a_0$.
According to Lemma \ref{the detail of degree deletion graph}, apart from the stars with leaves less than $\ell_{ij}$, $U_{ij}$ consists of exact $c_{ij}(a_{\jmath}^{ij}) - c_{ij}(a_{\jmath-1}^{ij})$ stars with $a_{\jmath-1}^{ij}$ leaves for all  $\jmath \in [m_{ij}]$ as well as exact $c_{ij}(\ell_{ij} - 1) - c_{ij}(a_{m_{ij}}^{ij})$ stars with $a_{m_{ij}}^{ij}$ leaves. It is not hard to calculate that the number of stars with at least $a^{m_{ij}}_{ij}= \ell_{ij}$ leaves is $c_{ij}(\ell_{ij} - 1)= a_{0}$.
Since $a_{0} > q_2$ according to Table \ref{table:reduction-values-relationship}, the number of stars with at least $\ell_{ij}$ leaves is larger than $\chi(U_{ij})$. It means that $\chi(U_{ij}) < c_{ij}(\ell_{ij}-1)$. 
Furthermore, since only the number of the vertices added to $X$ affects the dominating coefficients of the vertices outside of $U_{ij}$ with respect to $(H,X)$, for any $\chi(U_{ij})$, we may assume that the selection of $\chi(U_{ij})$ vertices maximize the value of the dominating coefficient of $V(U_{ij})\setminus X$ with respect to $(H,X)$. Thus, based on Lemma \ref{greedy lemma}, we may assume $V(U_{ij}) \cap X$ consists of the first $\chi(U_{ij})$ largest degree internal vertices of the stars graph $U_{ij}$.
Therefore, $\Delta(U_{ij} - X)$ is an element of set $I_{ij} \cup \{a_0\}$. 
Moreover, based on case 2 in the proof of Claim \ref{cl:ld-module-5},  $\Delta(U_{ij} - X)$ is at most $\hbar_{ij}$, which also means that $\chi(U_{ij}) \geq c_{ij}(\hbar_{ij})$ based on Definition \ref{def:kd-degdeletion-bdd-121}.
Overall, we have $\Delta(U_{ij} - X) \in I_{ij}$ and $c_{ij}(\hbar_{ij})  \leq \chi(U_{ij}) < c_{ij}(\ell_{ij}-1)$.

Let $\jmath \in [n]$ be the maximum index such that $c_{ij}(a_{\jmath}^{ij}) \leq \chi(U_{ij})$. Assume $\chi(U_{ij}) =  c_{ij}(a_{\jmath}^{ij}) + h$ for some $h >0$. Consider $X'$ obtained from $X$, but adding the largest $c_{ij}(a_{\jmath}^{ij})$ internal vertices from $U_{ij}$ instead.  Then $U_{ij} - X'$ still can have maximum degree $a_{\jmath}^{ij}$, and $\Lambda(U_{ij} - X, H, X) = \Lambda(U_{ij} - X', H, X')$.
In addition, according to the proof of Claim \ref{r-t are safe}, any vertex of $T_i$, $T_j$ and $R_{ij}$ always satisfies the linear inequality of the problem no matter what is the value of $\chi(U_{ij})$.
Hence, we may assume that $U_{ij} - X$ has degree $a_{\jmath}^{ij}$ and $\chi(U_{ij}) = c_{ij}(a_{\jmath}^{ij})$.

Secondly, consider $\alpha  = 1$. Every vertex $v\in V(H) - X$ satisfies that $|N(v)| - |N(v) \cap X| \leq - \beta$, where where $-\beta = (nk)^{10000}$. This implies that the degree of $v$ in $H$ is at most $|\beta|$ after deleting all vertices of $X$. 
As before, we may assume that the $\chi(U_{ij})$ deletions in $U_{ij}$ minimize $\Delta(U_{ij} - X)$.  According to case 2 in the proof of Claim \ref{cl:ld-module-5}, we have that $\chi(U_{ij})$ is at least $c_{ij}(\hbar_{ij} )$, and that  $U_{ij} - X$ has maximum degree $\hbar_{ij} $ or less.
Moreover, $U_{ij} - X$ cannot have maximum degree $\ell_{ij} - 1$ or less, since this requires that $\chi(U_{ij})$ is at least $c_{ij}(\ell_{ij} - 1) = a_0 = q_2+1$, a contradiction.
Overall, we may assume that $\ell_{ij} \leq \Delta(U_{ij} - X) \leq \hbar_{ij}$ and $c_{ij}(\hbar_{ij} ) \leq \chi(U_{ij}) < c_{ij}(\ell_{ij} - 1)$.

Recall that $\hbar_{ij} = a_{1}^{ij}$ and $\ell_{ij} = a_{m_{ij}}^{ij}$. Let $\jmath \in [m_{ij}]$ be the maximum index such that $c_{ij}(a_{\jmath}^{ij}) \leq \chi(U_{ij})$. Note that $\jmath$ is well-defined since $\chi(U_{ij}) \geq c_{ij}(\hbar_{ij} )$.
If $\jmath = m_{ij}$, then we made at least $c_{ij}(\ell_{ij})$ deletions and we have $\Delta(U_{ij} - X) = \ell_{ij}$.
If instead $\jmath < m_{ij}$, we have $\chi(U_{ij}) < c_{ij}(a_{\jmath+1}^{ij})$ and, by Defintion~\ref{def:kd-degdeletion-bdd-121}, we have $\Delta(U_{ij} - X) \geq a_{\jmath}^{ij}$. According to Defintion~\ref{def:kd-degdeletion-bdd-121} and Obersavation \ref{obervation-for-del-table}, it is possible to make $\Delta(U_{ij} - X) \leq a_{\jmath}^{ij}$ by adding $\chi(U_{ij})$ vertices to $X$. Thus, we may assume that $\Delta(U_{ij} - X) = a_{\jmath}^{ij}$.
In either case, $\Delta(U_{ij} - X) \in I_{ij}$, as desired.
Now assume that $\chi(U_{ij}) = c_{ij}(a_{\jmath}^{ij}) + h$ for some $h > 0$.  
Consider $X'$ obtained from $X$, but adding exactly $c_{ij}(a_{\jmath}^{ij})$ vertices from $U_{ij}$ instead.  Then $U_{ij} - X'$ can still have maximum degree $a_{\jmath}^{ij}$. 
In addition, according to the proof of Claim \ref{r-t are safe}, any vertex of $T_i$, $T_j$ and $R_{ij}$ has degree in $H$ less than $-\beta$ no matter what is the value of $\chi(U_{ij})$. 
Hence, we may assume that $U_{ij} - X$ has degree $a_{\jmath}^{ij}$ and $\chi(U_{ij}) = c_{ij}(a_{\jmath}^{ij})$.
\end{proof}

We sometimes abuse terminology by saying that $S_i$ choose $a_j \in I$ to indicate the following points: (1) if $\alpha = 0$, then the minimum degree of $S_i - X$ in $X$ is $a_j$, and $\chi(S_i) = p + c(a_j)$; (2) if $\alpha \in (0, 1)$, then $V(S_i) \cap X$ consists of the first $\chi(S_i)$ largest degree internal vertices of $S_i$, the maximum degree of $S_i - X$ is $a_j$, and $\chi(S_i) = c(a_j)$; (3) if $\alpha =1$, then the maximum degree of $S_i - X$ is $a_j$ and $\chi(S_i) = c(a_j)$.

\begin{claim}\label{cl:ld-module-8}
Suppose that $S_i$ chose $a_j \in I$.  Then, $\chi(T_i)$ is at least 
\begin{enumerate}
    \item 
     $s +l - a_j$  if $\alpha = 0$;

    \item 
    $\alpha a_j + \beta - 1$   if $\alpha \in (0,1)$; 
    
    \item 
    $a_j$  if $\alpha =1$. 
\end{enumerate}
\end{claim}
\begin{proof}
Consider $\alpha = 0$. For any vertex $v\in S_i - X$ in $H$, we have $|N(v)\cap X| \geq \beta$. Since the minimum degree of $S_i - X$ in $X$ is $a_j$, and $R_i \subseteq X$ based on Claim \ref{cl:ld-module-4}, we have $|R_i|  +\chi(T_i)+a_j  \geq   \beta$, which means that $\chi(T_i)$ is at least
 \begin{align*}
\beta - |R_i| - a_j = \beta - (|\beta| - s -l) - a_j,
\end{align*}
which is equal to $s +l - a_j$ since $\beta$ is a positive number here.

Consider $\alpha \in (0,1)$. Recall $s=0$ in this case. For any vertex $v\in S_i - X$ in $H$, we have $|N(v)\cap X| \geq \alpha |N(v)| + \beta$. Based on Claim \ref{cl:ld-module-6}, it is enough to make the internal vertex of the largest star, with $a_j$ leaves, in $S_i - X$ satisfied the inequality of the problem. Thus, we have
\begin{align*}
\chi{(B_i)} + \chi{(R_i)} + \chi{(T_i)} \geq \alpha (|B_i|+|R_i|+|T_i|+a_j) + \beta.
\end{align*}
Moreover, $\chi{(B_i)} = |B_i|$ and $\chi{(R_i)} = 0$ according to Claim \ref{cl:ld-module-2} and Claim \ref{cl:ld-module-4}. Then,  $\chi{(T_i)}$ is at least
\begin{align*}
&~ \alpha (|R_i|+|T_i|+a_j) + (\alpha -1)|B_i| + \beta\\
\geq &~ \alpha (|T_i|+a_j) + (\alpha -1)|B_i| + \beta\\
= &~ \alpha ( \left\lfloor \alpha (1-\alpha)^2x  \right\rfloor   + s + a_j ) -(1 - \alpha)  \left\lceil \alpha^2(1-\alpha)x  \right\rceil + \beta\\
= &~ \alpha a_j + \beta + \alpha \left\lfloor \alpha(1-\alpha)^2x  \right\rfloor - (1- \alpha)\left\lceil \alpha^2(1-\alpha)x  \right\rceil\\
> &~  \alpha a_j + \beta - 1.
\end{align*}

Consider  $\alpha =1$. For any vertex $v\in S_i - X$, the formula $|N(v)| - |N(v)\cap X|  \leq  - \beta$ should be satisfied. Suppose there are $h$ vertices in $X \cap V(S_i)$ that are adjacent to $v$. Then, we have 
\begin{align*}
|B_i| +|R_i|  +|T_i|+a_j + h  - (\chi{(B_i)} + \chi{(R_i)} + \chi{(T_i)} +h ) \leq   -\beta.
\end{align*}
Moreover, we have that $\chi(B_i) = |B_i| =0$, $ \chi{(R_i)} = 0$ according to Claim \ref{cl:ld-module-2} and Claim \ref{cl:ld-module-4}. Then, we have $\chi{(T_i)}$ is at least $\beta + |R_i|+|T_i|+a_j$, which is equal to
$ \beta + (|\beta| - s )+ s + a_j =  a_j$ since $\beta$ is a negative number. 
\end{proof}

We sometimes abuse terminology by saying that $U_{ij}$ chose $a + b \in I_{ij}$ to indicate the following points: (1) if $\alpha = 0$, then the minimum degree of $U_{ij} - X$ in $X$ is $a+b$, and $\chi(U_{ij}) = p_{ij} + p'_{ij} + c_{ij}(a + b)$; (2) if $\alpha \in (0, 1)$, then $V(U_{ij}) \cap X$ consists of the first $\chi(U_{ij})$ largest degree internal vertices of $U_{ij}$, the maximum degree of $U_{ij} - X$ is $a + b$, and $\chi(U_{ij}) = c_{ij}(a + b)$; (3) if $\alpha =1$, then the maximum degree of $U_{ij} - X$ is $a+b$ and $\chi(U_{ij}) = c_{ij}(a+b)$.

\begin{claim}\label{cl:ld-module-912}
Suppose that $U_{ij}$ chose $a + b \in I_{ij}$.  Then, $\chi(T_i)  + \chi(T_j)$ is at least 
\begin{enumerate}
    \item 
    $2s +2l - a-b$ if $\alpha = 0$;

    \item 
    $\alpha a +  \alpha b + \beta - 2$  if $\alpha \in (0,1)$; 
    
    \item 
    $a + b$ if $\alpha =1$. 
\end{enumerate}
\end{claim}
\begin{proof}
Consider $\alpha = 0$. For any vertex $v\in U_{ij} - X$, the formula $|N(v)\cap X| \geq \beta$ should be satisfied. Since the minimum degree of $U_{ij} - X$ in $X$ is $a+b$, and $R_{ij} \subseteq X$ based on Claim \ref{cl:ld-module-5}, we have $|R_{ij}|  +\chi(T_i)+\chi(T_j)+a +b  \geq   \beta$. Thus, $\chi(T_i)+\chi(T_j)$ is at least
\begin{align*}
 \beta - |R_{ij}| - a -b  
=  \beta - (|\beta| - 2s -2l) - a -b 
=  2s +2l - a-b
\end{align*}
since $\beta >0$ in this case.

Consider  $\alpha \in (0,1)$. Recall $s=0$ in this case. For any vertex $v\in U_{ij} - X$, the formula $|N(v)\cap X| \geq \alpha |N(v)| + \beta$ should be satisfied. Based on Claim \ref{cl:ld-module-7}, it is enough to make the internal vertex of the largest star, with $a+b$ leaves, in $U_{ij} - X$ satisfied the inequality of the problem. Thus, we have
\begin{align*}
\chi{(B_{ij})} + \chi{(R_{ij})} + \chi{(T_i)} +\chi{(T_j)} \geq \alpha (|B_{ij}|+|R_{ij}|+|T_i| + |T_i| + a+b) + \beta.
\end{align*}
Moreover, we have that $\chi{(B_{ij})} = |B_{ij}|$ and $\chi(R_{ij}) = 0$ according to Claim \ref{cl:ld-module-2} and Claim \ref{cl:ld-module-5}. Then,  $\chi{(T_i)} + \chi{(T_j)}$ is at least
\begin{align*}
&~ \alpha (|R_{ij}|+|T_i| + |T_j| + a+b) - (1-\alpha)|B_{ij}| + \beta\\
\geq &~ \alpha (|T_i| + |T_j| + a+b) - (1-\alpha)|B_{ij}| + \beta\\
= &~ \alpha ( 2\left\lfloor \alpha (1-\alpha)^2 x  \right\rfloor  + 2s + a+b ) -(1 - \alpha)  \left\lceil 2\alpha^2(1-\alpha)x  \right\rceil + \beta\\
= &~ \alpha a + \alpha b + \beta  + 2\alpha \left\lfloor \alpha(1-\alpha)^2x  \right\rfloor - (1- \alpha)\left\lceil 2\alpha^2(1-\alpha)x  \right\rceil\\
> &~ \alpha a + \alpha b + \beta  - 2.
\end{align*}

Consider  $\alpha =1$. For any vertex $v\in U_{ij} - X$, the formula $|N(v)| - |N(v)\cap X|  \leq  - \beta$ should be satisfied. Suppose there are $h$ vertices in $X \cap V(U_{ij})$ that are adjacent to $v$. Then, we have 
\begin{multline*}
|B_{ij}| +|R_{ij}|  +|T_i|+|T_j|+a+b + h \\ - (\chi{(B_{ij})} + \chi{(R_{ij})} + \chi{(T_i)} +\chi{(T_j)} +h ) \leq   -\beta.
\end{multline*}
Moreover, we have that $\chi(B_{ij}) = |B_{ij}| =0$ and that $ \chi{(R_{ij})} = 0$ based on  Claim \ref{cl:ld-module-2} and Claim \ref{cl:ld-module-5}. Then, we have
$\chi{(T_i)} +\chi{(T_j)} $ is at least  $\beta + |R_{ij}|+|T_i|+|T_j|+a+b$, which is equal to $\beta + (|\beta| - 2s )+ 2s + a+b = a+b$ since $\beta$ is a negative number here.
\end{proof}

\begin{claim}\label{cl:ld-module-10}
For each distinct $i, j \in [k]$, if $S_i$ chose $a \in I$ and $S_j$ chose $b \in I$, then $U_{ij}$ chose $a + b$.
\end{claim}
\begin{proof}
For each $i \in [k]$, according to Claim \ref{cl:ld-module-6}, we will denote by $\ha_i$ the element of $I$ that $S_i$ chose. For each distinct $i,j \in [k]$, we define $\hu_i$, $\he_i$, and  $p_{ij}$ as follows:
\begin{itemize}
    \item 
    if $\alpha = 0$, then $\hu_i = s+l-\ha_i$, $\he_i = p + s +l - \frac{1}{2} (\ha_i + a_n)$, and $\wp_{ij} = p_{ij} + p'_{ij}$;
    \item 
    if $\alpha \in (0,1)$, then $\hu_i = \alpha \ha_i - |\beta| -1$, $\he_i = t + \frac{\alpha}{2} \ha_i - |\beta| - 2$, and $\wp_{ij} = 0$;
    \item 
    if $\alpha = 1$, then $\hu_i = \ha_i$, $\he_i = s + \frac{1}{2} \ha_i$, and $\wp_{ij} = 0$.
\end{itemize}
Moreover, we claim that
\begin{align}
\chi(S_i) + \chi(T_i) \geq   \chi(S_i) + \hu_i \geq \he_i.
    \label{two-inequality-t-u-e}
\end{align}
The first inequality is true since   $\chi(T_i) \geq \hu_i$ according to Claim \ref{cl:ld-module-8}. For the second inequality, let us consider $\chi(S_i) + \hu_i$, which equals
\begin{itemize}
    \item 
    $p+ \frac{1}{2} (\ha_i - a_n) + s+l-\ha_i = \he_i $ if $\alpha = 0$;
    \item 
    $t - \left\lceil\frac{\alpha}{2}\ha_i\right\rceil$ + $\alpha \ha_i - |\beta| -1 > \he_i$ if $\alpha \in (0,1)$;
    \item 
    $s - \left\lceil\frac{1}{2} \ha_i \right\rceil +  \ha_i = \he_i$ if $\alpha = 1$.
\end{itemize}

Assume $\Psi_{ij} = \wp_{ij} + c_{ij}(\ha_i + \ha_j)$ for each distinct $i,j\in [k]$.
Let us define functions $t_{ij}: I_{ij}  \rightarrow \mathbb{N}$ for all distinct $i, j \in [k]$ as follows. Assume $U_{ij}$ chose $a' + b' \in I_{ij}$. Suppose that $T_{ij} \subseteq T_i \cup T_j$ is the minimum possible vertex set added to $X$ such that every vertex $v\in U_{ij} - X$ satisfies that $|N(v)\cap X| \geq \alpha |N(v)| + \beta$. Then, we use $t_{ij}(a' + b')$ to denote the vertex number of $T_{ij}$.
According to Claim \ref{cl:ld-module-912}, for any $a'+ b' \in I_{ij}$, we have $t_{ij}(a' + b')$ is at least
\begin{itemize}
    \item 
    $ 2s+2l-a'-b'$ if $\alpha = 0$;
    
    \item 
    $\alpha a' + \alpha b' -|\beta| -2$ if $\alpha \in (0,1)$;
    
    \item
    $a'+b'$ if $\alpha = 1$.
\end{itemize}


Assume $U_{ij}$ chose $a' + b'$. Then, according to Claim \ref{cl:ld-module-7}, we have $a' + b' \in I_{ij}$. We divide all $U_{ij}$ into three groups, denoted by $U_<,  U_=$, and $ U_>$, as follows.

\begin{itemize}
    \item 
    $U_<$ consists all $U_{ij}$ such that $a'+b' < \ha_i + \ha_j$ if $\alpha \in (0,1]$, and that $a'+b' > \ha_i + \ha_j$ if $\alpha = 0$;
    
    \item 
    $U_=$ consists all $U_{ij}$ such that $a'+b' = \ha_i + \ha_j$;
    
    \item
    $U_>$ consists all $U_{ij}$ such that $a'+b' > \ha_i + \ha_j$ if $\alpha \in (0,1]$, and that $a'+b' < \ha_i + \ha_j$ if $\alpha = 0$.
\end{itemize}
Furthermore, $U$ denotes the union of $U_<$, $U_=$, and $U_>$. 
$U_{\geq}$ denotes the union of $U_=$ and $U_>$.
$U_{\leq}$ denotes the union of $U_<$ and $U_=$.
In addition, since $|a'+b' - \ha_i - \ha_j| \geq r $, it is not hard to verify that $\Psi_{ij} > \chi(U_{ij})$ if $U_{ij} \in U_>$, that $\Psi_{ij} = \chi(U_{ij})$ if $U_{ij} \in U_=$, and that $\Psi_{ij} < \chi(U_{ij})$ if $U_{ij} \in U_<$.

To prove the claim, it suffices to show that $U_<$ and $U_>$ are empty (this is because $U_{ij} \in U_=$ is only possible if $U_{ij}$ chose $\ha_i + \ha_j$, since all the sum pairs of $I_{ij}$ are distinct).
The rough idea is as follows.
If each $U_{ij}$ chose the correct $\ha_i + \ha_j$, then each of them will incur a deletion cost of $\Psi_{ij}$ and end up satisfying that the number of all the vertices of the deletion set $X$ is at most $q$.  
If $U_<$ is non-empty, it incurs an extra deletion cost with respect to $\Psi_{ij}$ with no real benefit.  The complicated case is when $U_>$ is non-empty.  In this case, $U_{ij} - X$ incurs fewer deletions, which is $\chi(U_{ij})$, than if it had chosen $\ha_i + \ha_j$, which would have deleted $\Psi_{ij}$ vertices.  However, this needs to be compensated with extra deletions in $T_i$ and $T_j$.  By using a charging argument, we will show that the sum of extra deletions required for all the $U_>$ members outweighs the deletions saved in the $U_{ij}$'s of $U_>$.

First, let us consider the $U_>$ set. Suppose that it is non-empty. Recall that each $T_i$ requires at least $\hu_i$ deletions by Claim~\ref{cl:ld-module-8}. 
Denote by $X_0(T_i)$ an arbitrary subset of $X(T_i)$ containing exactly $\hu_i$ vertices, and denote $X_1(T_i) = X(T_i) \setminus X_0(T_i)$. (Of course, $X_1(T_i)$ could occasionally be an empty set.)  
We also use the notation $\chi_0(T_i) := |X_0(T_i)|$ and $\chi_1(T_i) := |X_1(T_i)|$.
The $X_0(T_i)$ corresponds to (part) deletions we had to do because of $S_i$, and $X_1(T_i)$ corresponds to ``extra'' deletions.  We will write $X_1 = \bigcup_{i \in [k]}X_1(T_i)$.

Consider some $U_{ij} \in U_>$. We know $U_{ij}$ chose $a'+ b'$. We define $\Delta_{ij} = t_{ij}(a' + b') - \hu_i - \hu_j$. We have the following statements.
\begin{enumerate}
    \item 
    If $\alpha = 0$, then  $t_{ij}(a' + b') \geq 2s +2l - a'-b'$ , which is larger than $\hu_i + \hu_j = 2s +2l - \ha_i - \ha_j$ since  $a'+b' < \ha_i + \ha_j$ in this case. Moreover, $\Delta_{ij}$ is at least $\ha_i + \ha_j - a'- b'$.

    \item 
    If $\alpha \in (0,1)$, then $t_{ij}(a' + b') \geq \alpha a' +  \alpha b' - |\beta| - 2$ ,  which is larger than $\hu_i + \hu_j = \alpha \ha_i + \alpha \ha_j - 2|\beta| -2$ since  $a'+b' > \ha_i + \ha_j$ in this case.  Moreover, $\Delta_{ij}$ is at least $\alpha( a'+ b' - \ha_i - \ha_j)$.
    
    \item 
    If $\alpha =1$, then $t_{ij}(a' + b') \geq a' + b'$, which is larger than $\hu_i + \hu_j = \ha_i + \ha_j$ since  $a'+b' > \ha_i + \ha_j$ in this case. Moreover, $\Delta_{ij}$ is at least $a'+ b' - \ha_i - \ha_j$.
\end{enumerate}
Thus, we always have $\Delta_{ij} > 0$ for every  $U_{ij} \in U_>$.
In addition, $(T_i \cup T_j) \cap X$ contains at least $t_{ij}(a' + b')$ vertices, which means that $\chi(T_i) + \chi(T_j) \geq t_{ij}(a' + b')$. Furthermore, we have $\chi_0(T_i) + \chi_0(T_j) = \hu_i + \hu_j$.  Therefore, $\chi_1(T_i) + \chi_1(T_j) \geq t_{ij}(a' + b') - (\hu_i + \hu_j) = \Delta_{ij}$. Moreover, $X_1$ is not empty since $U_>$ is not empty.

We define a (charging) function $\charge : X_1 \rightarrow \mathbb{R}$, illustrated in Figure~\ref{fig:bddcharge}.  We assume that each vertex of $X_1$ starts with $\charge(v) = 0$ and, in the description that follows, we describe how charges are added incrementally. For a $U_{ij} \in U_>$, we have $X_1(T_i) \cup X_1(T_j)$ contains at least $\Delta_{ij}$ vertices. Then, for each extra vertex $v \in X_1(T_i) \cup X_1(T_j)$ that ``related'' to $U_{ij}$, we increase $\charge(v)$ by $1/(k - 1)$, moreover, we repeat this for every $U_{ij} \in U_>$, and this concludes the charging procedure. Note that the ``related'' vertices here means that we need to add these vertices together with $ X_0(T_i) \cup X_0(T_j)$ to $X$ to make every vertex of $U_{ij} - X$ satisfies the linear inequality of the problem. Moreover, the number of the related vertices is \textit{exactly} $\Delta_{ij}$.

\begin{figure}
\centering
\begin{tikzpicture}
\filldraw[color=red, fill=red!3]  (2,2) rectangle (4,3.24); 
\filldraw[color=red, fill=red!3] (0,0) ellipse (2 and 1.236); 
\filldraw[color=red, fill=red!3] (6,0) ellipse (2 and 1.236); 

\draw[color=black] (-2.5,0) node {$T_i$};
\draw[color=black] (8.5,0) node {$T_j$};
\draw[color=black] (3,3.6) node {$U_{ij}\in U^>$};

\fill [color=black] (2.6,2.2) circle (1.5pt);
\draw (2.6,2.2) -- (2.2,2.7);
\draw (2.6,2.2) -- (2.4,2.7);
\draw (2.6,2.2) -- (3,2.7);

\fill [color=black] (2.2,2.7) circle (1.5pt);
\fill [color=black] (2.4,2.7) circle (1.5pt);
\fill [color=black] (2.6,2.7) circle (0.6pt);
\fill [color=black] (2.7,2.7) circle (0.6pt);
\fill [color=black] (2.8,2.7) circle (0.6pt);
\fill [color=black] (3,2.7) circle (1.5pt);

\fill [color=black] (3.4,2.45) circle (0.6pt);
\fill [color=black] (3.5,2.45) circle (0.6pt);
\fill [color=black] (3.6,2.45) circle (0.6pt);
\draw [decorate,decoration={brace,amplitude=4pt},xshift=0pt,yshift=0pt]
(2.2,2.79) -- (3,2.79) node [black,midway,xshift=-0.6cm]{};
\draw[color=black] (2.588,3.1) node[scale=0.7]{$a'+ b'$};

\filldraw[color=blue, fill=blue!20] (-1,0) ellipse (0.6 and 0.37);
\draw[color=black] (-1,0) node[scale=0.7] {$X^0(T_i)$};
\filldraw[color=blue, fill=blue!5] (0.6,0) ellipse (0.68 and 1.1);
\filldraw[color=blue, fill=blue!100] (0.6,0.3) circle(0.4);
\draw[color=black] (0.6,-0.35) node[scale=0.7] {$X^1(T_i)$};

\filldraw[color=blue, fill=blue!20] (6.9,0) ellipse (1 and 0.618);
\draw[color=black] (7.1,0) node[scale=0.7] {$X^0(T_j)$};
\filldraw[color=blue, fill=blue!5] (5.2,0) ellipse (0.49 and 0.8);
\draw[color=black] (5.2,-0.35) node[scale=0.7] {$X^1(T_j)$};
\filldraw[color=blue, fill=blue!100] (5.2,0.3) circle(0.28);

\draw[-stealth] (2.6,2.2) -- (0.4,0.65);
\draw[-stealth] (2.6,2.2) -- (0.65,0.7);
\draw[-stealth] (2.6,2.2) -- (0.96,0.12);

\fill [color=black] (1.1,0.9) circle (0.6pt);
\fill [color=black] (1.2,0.8) circle (0.6pt);
\fill [color=black] (1.3,0.7) circle (0.6pt);

\draw[-stealth] (2.6,2.2) -- (5.4,0.5);
\draw[-stealth] (2.6,2.2) -- (5.1,0.55);
\draw[-stealth] (2.6,2.2) -- (5,0.1);

\fill [color=black] (4.76,0.67) circle (0.6pt);
\fill [color=black] (4.7,0.58) circle (0.6pt);
\fill [color=black] (4.64,0.49) circle (0.6pt);
\end{tikzpicture}
\caption{Illustration of the charging procedure. The two deep blue circles consist of the extra vertices (deletions) related to $U_{ij}$. The total charge added by $U_{ij}$ to the vertices in the deep blue circles is exactly $\frac{1}{k-1} \Delta_{ij}$.}
\label{fig:bddcharge}
\end{figure}
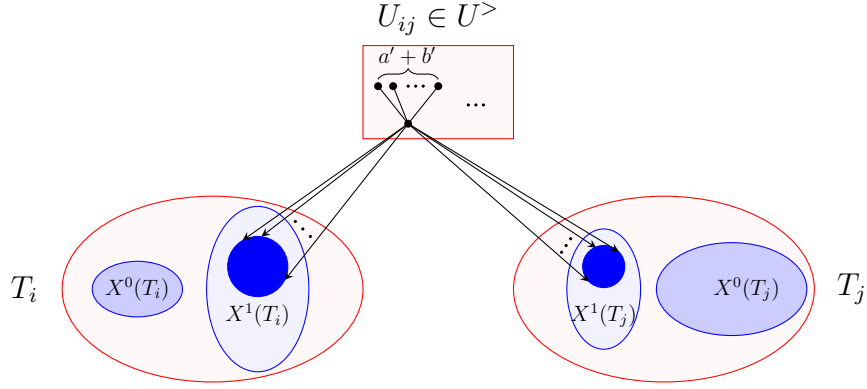

First observe that after the charging procedure is done, each $U_{ij}$ has charged at most $1/(k - 1)$ to each extra vertex $v \in X_1(T_i) \cup X_1(T_j)$ that related to $U_{ij}$.  
Moreover, for any $i \in [k]$ and any $v \in X_1(T_i)$, at most $k - 1$ of the $U_{ij}$ sets can charge that $1/(k - 1)$ to $v$, since there are only $k - 1$ of the $U_{ij}$ sets that have $i$ in their subscripts.
This means that $\charge(v) \leq (k - 1) \cdot 1/(k - 1) = 1$, and therefore that 
\begin{align}
    \sum_{v \in X_1} \charge(v) \leq |X_1|.
    \label{eq:charges-small-ld}
\end{align}

We further observe that, for any distinct $i,j\in  [k]$, the total charge added by $U_{ij}$ across the extra deletion vertices $X_1(T_i) \cup X_1(T_j)$ is exactly $\frac{1}{k-1} \Delta_{ij}$. Thus, the total charge added by all elements of $U_>$ across the extra deletion vertices $X_1$ is 
\begin{align}
\sum_{U_{ij} \in U^>}  \frac{1}{k-1} \Delta_{ij} = \sum_{v \in X_1} \charge(v).
\label{eq:total-charge-from-u}
\end{align}
Moreover, the number of deletions in $U_{ij}$ is $\chi(U_{ij}) = \wp_{ij} + c_{ij}(a' + b')$.  
This saves us some deletions compared to the number $\Psi_{ij}$, but this save in cost is much less than the total charge spread by $U_{ij}$, moreover, this fact will be demonstrated as follows.

First, we demonstrate the saving in each $U_{ij}$ as follows. 
\begin{align}
&~\Psi_{ij} - \chi(U_{ij}) \nonumber \\
=&~\wp_{ij} + c_{ij}(\ha_i + \ha_j) - \wp_{ij} - c_{ij}(a' + b') \nonumber \\
=&~ c_{ij}(\ha_i + \ha_j)  - c_{ij}(a' + b'). \label{the diference of u and hatu}
\end{align}

If $\alpha = 0$, then the formula (\ref{the diference of u and hatu}) equals
\begin{align*}
&~\frac{1}{2(k-1)}(\ha_i + \ha_j - l_{ij}) - \frac{1}{2(k-1)}(a' + b' - l_{ij})\\
=&~\frac{1}{k-1}(\ha_i + \ha_j - a' - b') - \frac{1}{2(k-1)}(\ha_i + \ha_j - a' - b')\\
\leq &~ \frac{1}{k-1}
\Delta_{ij} - \frac{1}{2(k-1)}(\ha_i + \ha_j - a' - b')\\
\leq &~ \frac{1}{k-1}
\Delta_{ij}- \frac{1}{2(k-1)} r\\
= &~ \frac{1}{k-1}
\Delta_{ij}- (k-1)k^3.
\end{align*}

If $\alpha = (0,1)$,  then the formula (\ref{the diference of u and hatu}) equals
\begin{align*}
&~  \frac{1}{k-1} \left(\left\lceil \frac{\alpha a'}{2} \right\rceil + \left\lceil \frac{\alpha b'}{2} \right\rceil\right) - \frac{1}{k-1} \left(\left\lceil \frac{\alpha \ha_i}{2} \right\rceil + \left\lceil \frac{\alpha \ha_j}{2} \right\rceil \right)\\
< &~  \frac{\alpha(a' + b' - \ha_i - \ha_j) + 4}{2(k-1)} \\
= &~  \frac{ \alpha (a' + b' - \ha_i - \ha_j)}{k-1} - \frac{\alpha (a' + b' - \ha_i - \ha_j) - 4}{2(k-1)}\\
\leq &~ \frac{1}{k-1}
\Delta_{ij} - \frac{\alpha (a' + b' - \ha_i - \ha_j) - 4}{2(k-1)}\\
\leq &~ \frac{1}{k-1}
\Delta_{ij} - \frac{1}{2(k-1)} (\alpha r - 4)\\
= &~ \frac{1}{k-1}
\Delta_{ij} - \frac{1}{2(k-1)} \left( \alpha \left\lceil\frac{10k(|\beta|+10)}{\alpha(1-\alpha)}\right\rceil^{10}  - 4 \right)\\
<&~ \frac{1}{k-1}
\Delta_{ij} -  \left(\frac{10k(|\beta|+10)}{\alpha(1-\alpha)}\right)^{9}.
\end{align*}

If $\alpha = 1$, then the formula (\ref{the diference of u and hatu}) equals
\begin{align*}
&~ \frac{1}{k-1} \left(\left\lceil \frac{a'}{2} \right\rceil + \left\lceil \frac{b'}{2} \right\rceil\right)   -  \frac{1}{k-1} \left(\left\lceil \frac{\ha_i}{2} \right\rceil + \left\lceil \frac{\ha_j}{2} \right\rceil \right) \\
=&~ \frac{ a' + b' - \ha_i - \ha_j}{k-1} -  \frac{a' + b' - \ha_i - \ha_j}{2(k-1)}  \\
\leq &~ \frac{1}{k-1}
\Delta_{ij} -  \frac{1}{2(k-1)} \left( a' + b' - \ha_i - \ha_j \right) \\
\leq &~ \frac{1}{k-1}
\Delta_{ij} -  \frac{r}{2(k-1)}\\
= &~ \frac{1}{k-1}
\Delta_{ij} -  (k-1)k^3.
\end{align*}
The second line is due to that $a',b',\ha_i$, and $\ha_j$ are even numbers.
Moreover, we define $r_1$ as follows. If $\alpha \in \{0,1\}$, then  $r_1= (k-1)k^3$. If $\alpha \in (0,1)$, then  $r_1= \left(\frac{10k(|\beta|+10)}{\alpha(1-\alpha)}\right)^{9}$. Thus, for any $U_{ij} \in U_>$ and any $\alpha \in [0,1]$, we have
\begin{align}
\Psi_{ij} - \chi(U_{ij}) \leq  \frac{1}{k-1}
\Delta_{ij} -  r_1. \label{the-differenct-pai-uij}
\end{align}
Then, summing over every $U_{ij} \in U_>$, we get that the all saving deletions are 
\begin{align}
&~\sum_{U_{ij} \in U_>} \left(\Psi_{ij} - \chi(U_{ij}) \right) \nonumber \\
\leq &~ \sum_{U_{ij} \in U_>}  \left(\frac{1}{k-1} \Delta_{ij} - r_1 \right) \nonumber\\
\leq &~ \sum_{U_{ij} \in U_>}  \frac{1}{k-1} \Delta_{ij} - r_1 \nonumber\\
= &~ \sum_{v \in X_1} \charge(v) - r_1 \nonumber\\
\leq&~ |X_1|  - r_1 \label{x1-is-lager-than-u},
\end{align}
where the second line is based on formula (\ref{the-differenct-pai-uij}), and the last and penultimate lines are based on formulas (\ref{eq:charges-small-ld}) and (\ref{eq:total-charge-from-u}), respectively.

Secondly, let us consider the $U_<$ set. Suppose that it is non-empty and some $U_{ij} \in U_<$. Let us consider the value of $\chi(U_{ij}) = \wp_{ij} + c_{ij}(a' + b')$.
If $\alpha =0 $, then $a' + b' \geq \ha_i + \ha_j + r$. We have 
\begin{align*}
\chi(U_{ij}) = &~\wp_{ij}+ \frac{1}{2(k-1)}(a' + b' - l_{ij})\\
\geq &~\wp_{ij}  + \frac{1}{2(k-1)}(\ha_i + \ha_j - l_{ij}) + \frac{1}{2(k-1)}r\\
= &~\Psi_{ij} + \frac{1}{2(k-1)}r.
\end{align*} 
If $\alpha \in (0,1)$, then $a' + b' \leq \ha_i + \ha_j - r$ and $\wp_{ij} = 0$. We have
\begin{align*}
\chi(U_{ij}) = &~ t -  \frac{1}{k-1} \left( \left\lceil \frac{\alpha a'}{2} \right\rceil + \left\lceil \frac{\alpha b'}{2} \right\rceil\right)\\
>&~  t -  \frac{1}{k-1} \left( \frac{\alpha a'}{2}  +  \frac{\alpha b'}{2} +2 \right)\\
 \geq&~ t -  \frac{1}{2(k-1)} \left( \alpha (\ha_i + \ha_j - r) + 4 \right)\\
\geq &~ t -  \frac{1}{k-1} \left( \left\lceil \frac{\alpha \ha_i}{2} \right\rceil + \left\lceil \frac{\alpha \ha_j}{2} \right\rceil\right) +  \frac{1}{2(k-1)} \left( \alpha r - 4 \right)\\
= &~ c_{ij}(\ha_i + \ha_j) +  \frac{1}{2(k-1)} \left( \alpha r - 4 \right)\\
= &~ \wp_{ij} + c_{ij}(\ha_i + \ha_j) +  \frac{1}{2(k-1)} \left( \alpha r - 4 \right)\\
= &~ \Psi_{ij} +  \frac{1}{2(k-1)} \left( \alpha r - 4 \right).
\end{align*} 
If $\alpha =1$, then $a' + b' \leq \ha_i + \ha_j - r$ and $\wp_{ij} = 0$. We have
\begin{align*}
\chi(U_{ij}) = &~ t -  \frac{1}{k-1} \left( \left\lceil \frac{a'}{2} \right\rceil + \left\lceil \frac{b'}{2} \right\rceil\right)\\
=&~  t -  \frac{1}{k-1} \left( \frac{a'}{2}  +  \frac{b'}{2} \right)\\
 \geq&~ t -  \frac{1}{2(k-1)} \left(  \ha_i + \ha_j -r \right)\\
\geq &~ t -  \frac{1}{k-1} \left( \left\lceil \frac{ \ha_i}{2} \right\rceil + \left\lceil \frac{\ha_j}{2} \right\rceil\right) +  \frac{r}{2(k-1)}\\
= &~ c_{ij}(\ha_i + \ha_j) +  \frac{r}{2(k-1)} \\
= &~ \wp_{ij} + c_{ij}(\ha_i + \ha_j) +  \frac{r}{2(k-1)}\\
= &~ \Psi_{ij} +  \frac{r}{2(k-1)}.
\end{align*} 
By replacing the values of $r$, we have that $\chi(U_{ij}) \geq \Psi_{ij} + (k-1)k^3$ if $\alpha \in \{0,1\}$, and that $\chi(U_{ij}) > \Psi_{ij} + \left(\frac{10k(|\beta|+10)}{\alpha(1-\alpha)}\right)^{9}$ if $\alpha \in (0,1)$. This means that $\chi(U_{ij}) \geq \Psi_{ij} + r_1$.
Then, summing over every $U_{ij} \in U_<$, we get have
\begin{align}
\sum_{U_{ij} \in U_<} \chi(U_{ij})  \geq \sum_{U_{ij} \in U_<} \left(\Psi_{ij} + r_1 \right) 
\geq \sum_{U_{ij} \in U_<} \Psi_{ij} + r_1.   \label{u-small-is-big}
\end{align}

We can now finish the proof of the claim.
We have that the size of $X \cap V_2$ is 
\begin{align}
    \sum_{i \in [k]} (\chi(S_i) + \chi(T_i)) + \sum_{1 \leq i < j \leq k} \chi(U_{ij}).
    \label{formula x in v2}
\end{align}

If $U_>$ is not empty, then formula (\ref{formula x in v2}) equals

\begin{align}
    & \sum_{i \in [k]} (\chi(S_i) + \chi_0(T_i)) + \sum_{i \in [k]}\chi_1(T_i) + \sum_{1 \leq i < j \leq k} \chi(U_{ij}) \nonumber \\ 
    =& \sum_{i \in [k]} (\chi(S_i) + \hu_i) + |X_1| + \sum_{1 \leq i < j \leq k} \chi(U_{ij}) \nonumber  \\ 
    \geq& \sum_{i \in [k]} \he_i + |X_1|  + \sum_{1 \leq i < j \leq k} \chi(U_{ij}) \nonumber \\ 
    \geq& \sum_{i \in [k]} \he_i + r_1 + \sum_{U_{ij} \in U_>} \left( \Psi_{ij} - \chi(U_{ij}) \right) + \sum_{1 \leq i < j \leq k} \chi(U_{ij}) \nonumber \\
    = & \sum_{i \in [k]} \he_i + r_1 + \sum_{U_{ij} \in U_>} \left( \Psi_{ij} - \chi(U_{ij}) \right) + \sum_{U_{ij} \in U_>} \chi(U_{ij}) + \sum_{U_{ij} \in U_{\leq}} \chi(U_{ij}) \nonumber \\
    =& \sum_{i \in [k]} \he_i + r_1 + \sum_{U_{ij} \in U_>}  \Psi_{ij}  + \sum_{U_{ij} \in U_{\leq}} \chi(U_{ij}) \nonumber  \\
    \geq & \sum_{i \in [k]} \he_i + r_1 + \sum_{U_{ij} \in U_>}  \Psi_{ij}  + \sum_{U_{ij} \in U_{\leq}} \Psi_{ij} \nonumber \\
    = & \sum_{i \in [k]} \he_i + r_1 + \sum_{U_{ij} \in U} \Psi_{ij}. \label{final-formula-is-large}
\end{align}
The third line and the fourth line are based on formulas (\ref{two-inequality-t-u-e}) and (\ref{x1-is-lager-than-u}), respectively. The penultimate line is due to inequality (\ref{u-small-is-big}), and it equals the last third line when $U_<$ is empty. Next, we demonstrate that formula (\ref{final-formula-is-large}) is larger than $q_2$.
According to the proof in Lemma \ref{ldlem:first-direction}, we have that, if $\alpha \in \{0,1\}$, then 
\begin{align*}
    \sum_{i \in [k]} \he_i + r_1 + \sum_{1 \leq i < j \leq k}  \Psi_{ij} = q_2 + r_1 > q_2,
\end{align*}
a contradiction.
Let us consider the case $\alpha \in (0,1)$. According to formula (\ref{use-in-last-uij-size}) in Lemma \ref{ldlem:first-direction}, we have
\begin{align*}
    \sum_{1 \leq i < j \leq k}  \Psi_{ij} 
    =  \sum_{1 \leq i < j \leq k}\left(t -  \frac{1}{k-1} \left(\left\lceil \frac{\alpha}{2} \ha_i \right\rceil + \left\lceil \frac{\alpha}{2} \ha_j \right\rceil \right)\right)
    = \binom{k}{2}t - \sum_{i \in [k]} \left\lceil \frac{\alpha}{2} \ha_i \right\rceil .
\end{align*}
Moreover, we have
\begin{align*}
    &\sum_{i \in [k]} \he_i + r_1\\
    = & \sum_{i \in [k]} \left( t + \frac{\alpha}{2} \ha_i - |\beta| - 2\right) + r_1\\
    > & \sum_{i \in [k]} \left( t + \left\lceil \frac{\alpha}{2} \ha_i \right\rceil - |\beta| - 3\right) + r_1\\
   \geq & \sum_{i \in [k]} \left( t + \left\lceil \frac{\alpha}{2} \ha_i \right\rceil\right) - k |\beta| - 3k +  \left(\frac{10k(|\beta|+10)}{\alpha(1-\alpha)}\right)^{9}\\
   > & \sum_{i \in [k]} \left( t + \left\lceil \frac{\alpha}{2} \ha_i \right\rceil\right) + 2k|\beta|. 
\end{align*}
Therefore, we have 
\begin{align*}
    &\sum_{i \in [k]} \he_i + r_1  + \sum_{1 \leq i < j \leq k}  \Psi_{ij}\\
    >& \sum_{i \in [k]} \left( t + \left\lceil \frac{\alpha}{2} \ha_i \right\rceil\right) + 2k|\beta| + \binom{k}{2}t - \sum_{i \in [k]} \left\lceil \frac{\alpha}{2} \ha_i \right\rceil\\
    = & k t + \binom{k}{2}t  + 2k|\beta|=  q_2,
\end{align*}
which is again a contradiction.

If $U_>$ is empty but $U_<$ is non-empty, then formula (\ref{formula x in v2}) equals
\begin{align*}
    & \sum_{i \in [k]} (\chi(S_i) + \chi(T_i)) + \sum_{U_{ij} \in U_=} \chi(U_{ij}) + \sum_{U_{ij} \in U_<} \chi(U_{ij}) \\
    &\geq \sum_{i \in [k]} \he_i + \sum_{U_{ij} \in U_=} \Psi_{ij} + \sum_{U_{ij} \in U_<} \chi(U_{ij})\\
    &\geq \sum_{i \in [k]} \he_i + \sum_{U_{ij} \in U_=} \Psi_{ij}  +\sum_{U_{ij} \in U_<} \Psi_{ij} + r_1\\
    &= \sum_{i \in [k]} \he_i + \sum_{U_{ij} \in U} \Psi_{ij} + r_1 \\
    & > q_2,
\end{align*}
a contradiction. Note that the penultimate line is the same as formula (\ref{final-formula-is-large}) that is larger than $q_2$, and that the second and third lines are based on formula (\ref{two-inequality-t-u-e}) and formula (\ref{u-small-is-big}), respectively.
Overall, both $U_<$ and $U_>$ are empty, and thus every $U_{ij}$ chose $\ha_i$ and $\ha_j$.
\end{proof}

This is all we need to construct a multicolored clique.  
To this end, we define $C = \{f^{-1}_{i}(\ha_i) : i \in [k] \mbox{ and $S_i$ chose $\ha_i$} \}$.
We claim that $C$ is a clique.  
By Claim~\ref{cl:ld-module-6}, each $S_i$ chooses some $\ha_i$ and thus $|C| = k$.
Now let $f^{-1}_{i}(\ha_i), f^{-1}_{j}(\ha_j)$ be two vertices of $C$, where $i<j$.  
Then $\ha_i, \ha_j$ were chosen by $S_i$ and $S_j$, respectively, and by Claim~\ref{cl:ld-module-10}
we know that $U_{ij}$ chose $\ha_i + \ha_j$.  
By the construction of the $U_{ij}$ solution table, this is only possible if both $(f^{-1}_{i}(\ha_j),$ $ f^{-1}_{j}(\ha_i))$ and  $(f^{-1}_{i}(\ha_i), f^{-1}_{j}(\ha_j))$ are in $E(G)$.  Therefore, $(f^{-1}_{i}(\ha_i), f^{-1}_{j}(\ha_j))\in E(G)$ and $C$ is a clique. 
\end{proof}

\end{document}